\documentclass[a4paper, 12pt]{article}

\usepackage[sort&compress]{natbib}
\bibpunct{(}{)}{;}{a}{}{,} 

\usepackage{amsthm, amsmath, amssymb, mathrsfs, multirow, url, subfigure}
\usepackage{graphicx} 
\usepackage{ifthen} 
\usepackage{amsfonts}
\usepackage[usenames]{color}
\usepackage{fullpage}
\usepackage{tikz}

\theoremstyle{plain} 
\newtheorem{thm}{Theorem}
\newtheorem{cor}{Corollary}
\newtheorem{prop}{Proposition}
\newtheorem{lem}{Lemma}

\theoremstyle{definition}
\newtheorem{defn}{Definition}

\newtheorem*{prinrep}{Principle of Representation}
\newtheorem*{prinexp}{Principle of Expressiveness}
\newtheorem*{prinpl}{Principle of Plausibility}
\newtheorem*{princom}{Principle of Minimal Complexity}
\newtheorem{ex}{Example}
\newtheorem{remark}{Remark}
\newtheorem*{question}{Question}

\theoremstyle{remark}

\newcommand{\prob}{\mathsf{P}} 
\newcommand{\E}{\mathsf{E}}

\newcommand{\bin}{{\sf Bin}}
\newcommand{\unif}{{\sf Unif}}
\newcommand{\nm}{{\sf N}}
\newcommand{\expo}{{\sf Exp}}
\newcommand{\gam}{{\sf Gamma}}

\newcommand{\chisq}{{\sf ChiSq}}
\newcommand{\mult}{{\sf Mult}}
\newcommand{\bet}{{\sf Beta}}
\newcommand{\cau}{{\sf Cauchy}}
\newcommand{\ber}{{\sf Ber}}

\newcommand{\RR}{\mathbb{R}}

\newcommand{\LL}{\mathbb{L}}
\newcommand{\XX}{\mathbb{X}}
\newcommand{\YY}{\mathbb{Y}}
\newcommand{\ZZ}{\mathbb{Z}}

\newcommand{\TT}{\mathbb{T}}

\renewcommand{\SS}{\mathbb{S}}

\newcommand{\iid}{\overset{\text{\tiny iid}}{\,\sim\,}}
\newcommand{\ind}{\overset{\text{\tiny ind}}{\,\sim\,}}

\newcommand{\model}{\mathscr{P}}
\newcommand{\prior}{\mathsf{Q}}
\newcommand{\credal}{\mathscr{Q}}
\newcommand{\cred}{\mathscr{C}}

\newcommand{\lPi}{\underline{\Pi}}
\newcommand{\uPi}{\overline{\Pi}}

\newcommand{\uprior}{\overline{\mathsf{Q}}}

\newcommand{\lOmega}{\underline{\Omega}}
\newcommand{\uOmega}{\overline{\Omega}}
\newcommand{\uGamma}{\overline{\Gamma}}
\newcommand{\lGamma}{\underline{\Gamma}}

\newcommand{\lprob}{\underline{\mathsf{P}}}
\newcommand{\uprob}{\overline{\mathsf{P}}}

\title{Valid and efficient imprecise-probabilistic inference with partial priors, II. General framework}
\author{Ryan Martin\footnote{Department of Statistics, North Carolina State University, {\tt rgmarti3@ncsu.edu}}}
\date{\today}

\begin{document}

\maketitle 

\begin{abstract}
Bayesian inference requires specification of a single, precise prior distribution, whereas frequentist inference only accommodates a vacuous prior.  Since virtually every real-world application falls somewhere in between these two extremes, a new approach is needed.  This series of papers develops a new framework that provides valid and efficient statistical inference, prediction, etc., while accommodating partial prior information and imprecisely-specified models more generally.  This paper fleshes out a general inferential model construction that not only yields tests, confidence intervals, etc.~with desirable error rate control guarantees, but also facilitates valid probabilistic reasoning with  de~Finetti-style no-sure-loss guarantees.  The key technical novelty here is a so-called outer consonant approximation of a general imprecise probability which returns a data- and partial prior-dependent possibility measure to be used for inference and prediction.  Despite some potentially unfamiliar imprecise-probabilistic concepts in the development, the result is an intuitive, likelihood-driven framework that will, as expected, agree with the familiar Bayesian and frequentist solutions in the respective extreme cases.  More importantly, the proposed framework accommodates partial prior information where available and, therefore, leads to new solutions that were previously out of reach for both Bayesians and frequentists.  Details are presented here for a wide range of examples, with more practical details to come in later installments.


\smallskip

\end{abstract}

\vfill

\pagebreak

\tableofcontents

\pagebreak

\section{Introduction}
\label{S:intro}

\begin{quote}
{\em As Bayes perceived, the concept of Mathematical Probability affords a means, in some cases, of expressing inferences from observational data, involving a degree of uncertainty, and of expressing them rigorously, in that the nature and extent of the uncertainty is specified with exactitude, yet it is by no means axiomatic that the appropriate inferences...
should always be rigorously expressible in terms of this same concept.} \citep[][p.~40]{fisher1973}
\end{quote}

As stated in Part~I \citep{martin.partial}, the overarching goal of this series is to develop a new and unified ``probabilistic'' framework in which valid and efficient statistical inference, prediction, decision-making, model assessment, etc., can be carried out.  This framework contains the familiar frequentist and Bayesian schools as (extreme) special cases, in addition to providing new solutions to problems that these existing frameworks can't solve.  This very optimistic goal will be achieved by leveraging the previously-untapped (by statisticians at least) power of {\em imprecise probability}.  In its simplest form, when genuine partial prior information is available---too much to justify ignoring and carrying out a frequentist analysis but too little to justify a single prior distribution and corresponding Bayesian analysis---this new framework accommodates exactly the information available in the form of an imprecise prior distribution.  This creates opportunities for efficiency gain compared to a frequentist solution without the risk of bias resulting from the strong assumptions that a Bayesian solution requires.  Part~I of the series laid the groundwork, so the goal here in Part~II is to show specifically {\em how} the aforementioned opportunities can be realized, that is, specifically how can partial prior information be combined with a posited statistical model and observable data in a way that simultaneously makes intuitive sense and guarantees statistical inference, etc.~is valid and efficient.  


To set the scene, suppose a statistical model---mathematically, a collection of probability distributions---is posited to describe the sampling variability in the observable data $Y \in \YY$.  These probability distributions are denoted by $\prob_{Y|\theta}$, depending on a parameter $\theta \in \TT$ that's {\em uncertain}.  Throughout this work I'll use an upper-case $\Theta$ to denote the uncertain variable, and a lower case $\theta$ to denote generic values of it.  Then $\prob_{Y|\theta}$ is effectively the conditional distribution of $Y$, given the uncertain variable $\Theta$ takes on value $\theta$, but no (precise) joint distribution is required. The goal is to quantify uncertainty about $\Theta$ after the data $Y=y$ has been observed.  More formally, the goal is to construct what I call an {\em inferential model} or IM, which is a mapping that takes data $y \in \YY$, model assumptions, and perhaps partial prior information as input, and returns an {\em imprecise probability distribution} $(\lPi_y, \uPi_y)$ on $\TT$ that can be used for uncertainty quantification and inference; see Section~\ref{S:background} below for a recap of notation and terminology from Part~I, in particular, what I mean by ``imprecise probability distribution'' and ``partial prior.''  From this imprecise-probabilistic uncertainty quantification, one can easily derive hypothesis tests, confidence regions, and other statistical decision procedures if desired, and the validity and efficiency of the IM will carry over to the derived procedures.

The IM definition is very flexible, allowing for familiar kinds of precise probabilities, such as Bayesian \citep{bernardo.smith.book, ghosh-etal-book}, fiducial \citep{fisher1935a, fisher1973, zabell1992, seidenfeld1992}, generalized fiducial \citep{hannig.review, williams.intro.gfi}, confidence distributions \citep{schweder.hjort.book, xie.singh.2012, thornton.xie.cd}, etc., and some less familiar imprecise probabilities, such as belief functions in the spirit of \citet{dempster1967, dempster1968a, dempster1968b, dempster2008}, \citet{shafer1976, shafer1982}, and others, including \citet{denoeux2014}, \citet{denoeux.li.2018}, and \citet{wasserman1990b}, and generalized Bayes as in \citet{walley1991, walley2002}, \citet{augustin.etal.bookchapter}, and elsewhere.  The point of my hyper-inclusive definition is to set a baseline objective and to remove any bias or stigma\footnote{For example, many associate fiducial inference with ``Fisher's biggest blunder'' \citep{efron1998} and would dismiss it offhand based on that comment alone. But Efron goes on to say that ``something like fiducial inference will play an important role'' in the 21st century, so he's not suggesting that it be dismissed. In fact, ``There is still more to be learned from Fisher, even when he seems clearly wrong, than from any other contributor to statistical thinking'' \citep{dawid1991}.} associated with their names.  If all of these ideas fall under the same IM umbrella and are aiming to achieve the same general objectives, then the discussion can focus, first, on what kind of properties IMs should have and, second, on determining which, if any, of the available candidates have those desirable properties.  

Most of the previous work in this area has focused on the case where the available prior information about $\Theta$ is vacuous, i.e., literally nothing is known about $\Theta$ other than $\Theta \in \TT$.  Even Bayesians focus on, e.g., credible interval coverage probabilities, asymptotic posterior concentration, etc.~in the vacuous prior setting.  Despite being a relatively extreme situation in the big picture, the vacuous prior case is of practical importance because researchers want to work on new problems in which they have little experience.  In \citet{imbasics, imbook} and, more recently, \citet{martin.nonadditive, imchar}, I argued that IMs ought to satisfy a so-called {\em validity property} which, in the case of a vacuous prior, takes the form \eqref{eq:valid.vacuous} below.  On the one hand, validity differs from the usual frequentist considerations because it takes as primitive that the IM output will be used for ``probabilistic reasoning''---e.g., small (upper) probability assigned to a hypothesis $A \subset \TT$ about $\Theta$ suggests that there's little support for its truthfulness in the data---and ensures that the associated logic is reliable.  On the other hand, validity completely aligns with frequentism because it implies that statistical procedures derived from the IM achieve the usual error rate control properties.  So, one can't criticize validity as being irrelevant or too weak.  The criticism begins only after the validity condition is enforced and we see which IMs pass the test.  \citet{balch.martin.ferson.2017} showed that {\em validity fails} for any IM whose output is an ordinary countably additive probability distribution for $\Theta$, which includes (subjective, empirical, or default prior) Bayes, (generalized) fiducial, confidence distributions, etc.; see, also, \citet{prsa.response}.  While these additive IMs can be used to construct statistical procedures with frequentist error rate control properties, they can't provide reliable probabilistic reasoning.\footnote{In retrospect, this isn't surprising. A probability distribution must carry an enormous amount of information in order to provide any summary requested of it with complete precision, so it's unrealistic to expect that all this can be done reliably using only data, model, and thin air.  As \citet{fraser.copss} said ``[Bayes's formula] does not create real probabilities from hypothetical probabilities.'' This limitation can be overcome by lowering our expectations of how informative the IM's output should be.}  Thus the controversy: to achieve valid probabilistic reasoning, one {\em must} abandon the comfort and familiarity of ordinary/precise probability.  This is scary, I agree. Fortunately, what's needed to achieve valid probabilistic reasoning isn't very deep into the unfamiliar territory of imprecise probability. 

The fact is, imprecise probabilities are everywhere in statistics.  For example, p-values correspond to imprecise probabilities \citep[e.g.,][]{impval} and make up an important part of what I referred to above as valid probabilistic reasoning.  More generally, {\em every statistical procedure} that provably controls error rates at the nominal level corresponds to a valid IM and an imprecise probability \citep{imchar}.  So, I'm not asking folks to give up their go-to methods, e.g., t-tests and analysis of variance, only to see them through an imprecise-probabilistic lens.  This change of perspective will prove to be generally valuable moving forward.  For one thing, it's well-documented that our obsession with ``null hypothesis significance testing,'' and statistical decision procedures more generally, has and is negatively affecting scientific progress \citep[e.g.,][]{mayo.book.2018, lash2017}, so a new perspective is needed.  Also, the perspective that I'm suggesting both simplifies and clarifies the interpretation of the aforementioned go-to methods: p-values measure the {\em plausibility} of the null hypothesis and confidence intervals contain all those parameter values that are sufficiently {\em plausible}.  These concepts are simple, it's the textbook-style mantra---``if the method is repeatedly applied to many hypothetical data sets, then...''---that creates confusion, along with other fallacies such as
\begin{quote}
{\em the [p-value] fails to provide the probability of the null hypothesis, which is needed to provide a strong case for rejecting it.} \citep{pvalue.ban}
\end{quote}
\citet[][Ch.~3]{fisher1973} readily acknowledges that p-values and other forms of quantitative inference aren't sufficiently informative to justify {\em probabilistic} statements, and he argues that this isn't a shortcoming.  In a similar vein, my claim is that p-values, etc.~are informative enough to justify {\em imprecise-probabilistic}, or {\em possibilistic} statements, which is all that's needed for sound statistical inference; see Section~\ref{S:background}.   

The situation described above, with vacuous prior information about $\Theta$, is common in the frequentist literature, but it misses almost all real-world applications.  The same can be said for the other case common in the Bayesian literature, where a precise prior distribution for $\Theta$ is singled out.  The fact is, the range of problems that naturally fit under the umbrellas associated with the Bayesian or frequentist schools is quite narrow.  Virtually every real problem falls somewhere between these two extremes: there is some genuine prior information about $\Theta$ available that shouldn't be ignored, as it would be by frequentists, but not enough to justify adoption of a complete prior probability distribution for $\Theta$ as Bayesians would require.  To fill this gap, Part~I considers representing {\em a priori} uncertainty about $\Theta$ as a {\em partial prior distribution} that takes the mathematical form of an imprecise probability defined on $\TT$.  The idea behind this partial prior is to capture exactly those cases between the two extremes where some information about $\Theta$ is available, but not enough to identify and wholeheartedly support a complete prior probability distribution.  In what real problem would a subject-matter expert not at least be able to provide information of the form ``I'm 100$x$\% certain that $\Theta$ is in $K \subset \TT$''?  Currently, information of this form is either ignored or artificially embellished upon, thereby pushing the application to one of the two extremes where it doesn't belong.  But this is exactly the situation imprecise probability is designed to accommodate.  

Previous efforts focused on the construction of valid IMs when the prior information is vacuous, but the availability of partial prior information changes things dramatically, and rightfully so.  Aside from the obvious practical question of how to incorporate this partial prior information into the IM construction, there's a more basic---and arguably more important---question of what properties should a partial-prior-dependent IM satisfy.  It's this latter question that was the primary focus of Part~I of this series.  In particular, Part~I proposed a new version of the validity property (see Definition~\ref{def:validity} below), generalizing that previously considered for the vacuous prior case, and explored its statistical and behavioral consequences; more on both of these below.  Part~I also explored a few relatively simple strategies for constructing a partial-prior-dependent IM and, while some of those constructions (e.g., Walley's generalized Bayes) achieved this new validity property, none of them did so efficiently.  Therefore, a more effective construction of valid, partial-prior IMs is needed and the goal here in Part~II is to meet this need.  

Unlike in the vacuous prior case, when partial prior information is available, the validity condition described above need not imply that statistical procedures derived from the valid IM satisfy all the desired properties.  For example, confidence regions---see Equation~\eqref{eq:conf.region} below---derived from an IM that's only valid need not achieve the nominal coverage.  Fortunately, there's a stronger condition, namely {\em strong validity} (Definition~\ref{def:strong.validity} below), from which all the relevant statistical properties are consequences; see Theorem~\ref{thm:valid} below.  Certain behavioral properties, in the spirit of de~Finetti's no-sure-loss, are also consequences of strong validity; see Theorem~\ref{thm:coherent} below.  Another relevant take-away from Part~I is that strong validity can effectively only be achieved by IMs that possess a particular mathematical structure, namely, that of a necessity--possibility measure pair \citep[e.g.,][]{dubois.prade.book} or, equivalently, a consonant belief--plausibility function pair \citep[e.g.,][Ch.~10]{shafer1976}, which are among the simplest of imprecise probability models.  Thanks to the new (complete-class-like) result in Section~\ref{S:faq} below, 
{\em I focus exclusively here on IMs whose output is consonant}.  This means the simplicity and strong validity properties enjoyed by consonant IMs are effectively without loss of generality.  

For the construction of consonant IMs, the chief technical novelty here is the use of a so-called {\em outer consonant approximation} \citep[e.g.,][]{dubois.prade.1990, dubois.etal.2004, hose2022thesis}.  I'll start with a very general setup in which the data analyst posits an imprecise joint distribution for the pair $(Y,\Theta)$.  This generalizes the situation described above where there's an imprecise prior for $\Theta$ and a precise probability distribution for $Y$, given $\Theta=\theta$.  That this model specification is allowed to be imprecise is crucial: it'd be too much of a burden on the data analyst to require that he specify a precise joint distribution.  The present framework is flexible and allows the data analyst to specify only what he's able to justify, so it's fully expected that this specification would be imprecise.  From this imprecise joint distribution for $(Y,\Theta)$, the proposal here is, roughly, to construct a consonant IM by simply approximating the joint imprecise probability distribution (from above) by a simpler, consonant belief--plausibility function pair.  In the literature, this kind of approximation is most commonly used to simplify computation, so its use here to achieve certain statistical objectives is new to my knowledge.  Unbeknownst to me at the time, this consonant approximation was the workhorse behind what I referred to in Part~I as {\em validification}.  Details are presented in Section~\ref{S:blocks}. 

The above description of my proposed IM construction---just find an outer consonant approximation of the imprecise joint distribution for $(Y,\Theta)$---accurately captures the intuition but overlooks two nontrivial challenges.  
\begin{itemize}
\item There's unavoidable ambiguity in the construction of this outer consonant approximation, due to the choice of ``plausibility order.''  The justification for my choice of ordering in Section~\ref{SS:minimal}, which closely follows \citet{hose2022thesis}, is based largely on the following basic principle:
\begin{equation}
\label{eq:dubois.quote}
\text{{\em what is probable must be plausible}. \citep[][p.~121]{dubois.prade.book}} 
\end{equation}
\item The proposed outer approximation approximates the imprecise joint distribution of $(Y,\Theta)$, but the goal is quantifying uncertainty about $\Theta$ relative to the observed value of $Y$, so a choice has to be made between {\em conditioning} or {\em focusing} on $Y=y$.  In Section~\ref{SS:construction}, I argue that the latter is the right choice for the present context, but this raises some technical follow-up questions that need to be addressed.  
\end{itemize} 

The power of this proposed IM construction is revealed in Section~\ref{SS:properties}, where strong validity follows automatically from general results on outer consonant approximations.  Strong validity itself implies that the IM-based procedures, e.g., tests and confidence regions, control frequentist-style error rates, which settles most of the relevant statistical questions.  De~Finetti-style no-sure-loss results can also be established for this new consonant IM, and the proof here is informative in that it relies heavily on the properties of the consonant approximation and how it relates to the posited imprecise joint distribution for $(Y,\Theta)$.  Section~\ref{SS:efficiency} gives a high-level discussion of the efficiency of this new consonant IM.  There I put forward a general {\em Principle of Minimal Complexity} as a guideline to assist in the construction of an efficient IM.  

Starting in Section~\ref{S:precise}, the focus is on specializing the general IM construction to some common statistical settings.  The first is that most-familiar special case of a precise statistical model and imprecise prior.  Looking at the new IM construction in this context reveals some interesting structure in the construction itself.  In particular, one can see that the imprecision in the prior enters into the IM itself it two ways: a {\em calibration} step that's responsible for (strong) validity and a {\em regularization} step that helps this IM be more efficient than one that assumes prior information is vacuous.  These are important insights to understanding how and why the IM works.  There the Principle of Minimal Complexity is related to the familiar dimension-reduction steps based on sufficient statistics and conditioning on ancillaries.  I also address the case where the prior information is complete, i.e., there's just a single prior for $\Theta$ as in the classical Bayesian setting.  There the Principle of Minimal Complexity leads to a consonant IM that depends only on the Bayesian posterior distribution but retains strong validity, etc.  

Section~\ref{SS:examples} presents a number of examples, with numerical illustrations, to show how the proposed solution can be put into practice and to give some concrete examples of what partial prior information might look like and how it can be incorporated.  The short Section~\ref{S:imprecise} aims to address the case in which the statistical model itself is imprecise, and there I focus on cases for dealing with missing or otherwise coarse data.  Deeper investigation into specific practical problems of interest will be considered elsewhere.  Section~\ref{S:discuss} wraps up Part~II of the series and lists a number of (what I think are) interesting open questions and promising future directions to pursue.

\section{Background}
\label{S:background}

As above, I'll consider here a pair of uncertain variables $(Y,\Theta) \in \YY \times \TT$.  The most common situation in the statistical literature is that where $\Theta$ represents the parameter of a precise statistical model for $Y$, i.e., where $(Y \mid \Theta=\theta) \sim \prob_{Y|\theta}$, with $\prob_{Y|\theta}$ an ordinary or precise probability distribution on $\YY$.  Here I'll focus (almost) exclusively on this classical case; but see Section~\ref{S:imprecise} for a brief discussion of the case where the statistical model is imprecise in one or more ways.  It can happen, however, that the inference problem doesn't originate with a parametric statistical model.  In fact, most machine learning applications are of this type---where the to-be-inferred unknown is defined as the minimizer of some expected loss function.  Fortunately, the framework developed here can be extended to such cases, but these details will be presented in a later installment in the series.


Regardless of how $\Theta$ is defined, when the conditional distribution of $Y$, given $\Theta=\theta$, is precise, then an imprecise joint distribution for $(Y,\Theta)$ is determined by an imprecise or partial prior distribution for $\Theta$.  This corresponds to a coherent  lower/upper probability pair $(\lprob_\Theta, \uprob_\Theta)$ on $\TT$.  Coherence ensures that there's a corresponding (closed and convex) collection of precise probability distributions on $\TT$ such that $\lprob_\Theta$ and $\uprob_\Theta$ are the lower and upper envelopes, respectively.  That is, there's a prior {\em credal set} 
\[ \cred(\uprob_\Theta) = \{\prob_\Theta: \prob_\Theta(A) \leq \uprob_\Theta(A) \text{ for all $A \subseteq \TT$}\}, \]
and the lower/upper envelope connection is given by 
\[ \lprob_\Theta(A) = \inf_{\prob_\Theta \in \cred(\uprob_\Theta)} \prob_\Theta(A) \quad \text{and} \quad \uprob_\Theta(A) = \sup_{\prob_\Theta \in \cred(\uprob_\Theta)} \prob_\Theta(A), \quad A \subseteq \TT. \]

Compared to Part~I, here the starting point is a bit more general.  I'll assume that the data analyst specified an ``imprecise joint distribution'' for $(Y,\Theta)$ in the form of a coherent lower/upper probability pair $(\lprob_{Y,\Theta}, \uprob_{Y,\Theta})$.  As above, coherence ensures that there's a corresponding credal set of precise joint distributions on $\YY \times \TT$, such that $\lprob_{Y,\Theta}$ and $\uprob_{Y,\Theta}$ are the lower and upper envelopes, respectively.  That is, 
\[ \cred(\uprob_{Y,\Theta}) = \{\prob_{Y,\Theta}: \prob_{Y,\Theta}(E) \leq \uprob_{Y,\Theta}(E) \text{ for all $E \subseteq \YY \times \TT$}\}, \]
and 
\[ \lprob_{Y,\Theta}(E) = \inf_{\prob_{Y,\Theta} \in \cred(\uprob_{Y,\Theta})} \prob_{Y,\Theta}(E) \quad \text{and} \quad \uprob_{Y,\Theta}(E) = \sup_{\prob_{Y,\Theta} \in \cred(\uprob_{Y,\Theta})} \prob_{Y,\Theta}(E). \]
Virtually all of the imprecise probability models used in applications are coherent, so there is no loss of generality in taking this as the basic setup here. 

The rationale behind my focus on imprecision is two-fold.  First, since an imprecise probability is, by definition, less specific than a precise probability, the former ought to be {\em easier} for the data analyst to specify than the latter in the sense that a less specific model puts less of a burden on the data analyst.  Second, modern efforts in statistics and machine learning focus on developing methods that are ``model agnostic,'' i.e., their construction doesn't require specification of a model and/or the methods' properties are relatively insensitive to variations in the posited model.  This is motivated by the
\begin{quote}
{\em Law of Decreasing Credibility.} The credibility of inference decreases with the strength of the assumptions maintained. \citep[][p.~1]{manski.book}
\end{quote}
The model agnostic perspective is overly conservative---it's rare the data analyst wouldn't be able to narrow down at all the class of plausible models in their application.  Manski's law suggests that there's a spectrum of model assumptions over which the credibility of inferences varies.  By allowing the model specification to be imprecise, I'm giving data analysts the flexibility that the all-or-nothing, precise-or-agnostic perspective lacks.  

The goal is to leverage the connection between $Y$ and $\Theta$ defined by the (imprecise) model to make inference about $\Theta$ based on the observed $Y=y$.  As described in Part~I and the references therein, an {\em inferential model} (IM) is a mapping from $y \in \YY$ to a coherent lower and upper probability pair $(\lPi_y, \uPi_y)$ depending explicitly on the observed $y$ and implicitly on the model assumptions, etc.  Coherence implies that these are genuine lower and upper probabilities in the sense that they're lower and upper envelopes, respectively, of the class of compatible (precise) probability distributions on $\TT$, i.e., 
\[ \lPi_y(A) = \inf_{\Pi \in \cred(\uPi_y)} \Pi(A) \quad \text{and} \quad \uPi_y(A) = \sup_{\Pi \in \cred(\uPi_y)} \Pi(A), \quad A \subseteq \TT. \]
Again, asking for the IM output to be coherent in this sense is no practical restriction because most, if not all, of the imprecise probability models discussed in the literature are coherent, e.g., $K$-alternating upper probabilities. 

The IM output being coherent for each fixed $y$ is important for its interpretation; see Section~2.2 in Part~I.  But coherence for each fixed $y$ isn't sufficient for my purposes here.  To ensure that the IM-based probabilistic reasoning to be reliable, it's necessary to enforce a certain calibration between the IM's output and the posited model.  By ``probabilistic reasoning,'' I mean making inferences about assertions concerning the unknown $\Theta$ based on the magnitudes of the IM's imprecise probability output.  Roughly speaking, if $\uPi_y(A)$ is small, then I'd infer $A^c$ in the sense that I'd be inclined to believe that ``$\Theta \not\in A$.''  While that's an intuitively reasonable approach, if $\uPi_Y(A)$ tends to be small, as a function of $Y$, even when $\Theta \in A$, then this type of inference wouldn't be reliable.  There's a dual to the above based on lower probabilities, but I'll not discuss this here.
To avoid this lack of reliability, I introduced in Part~I the following notion of {\em validity}.  

\begin{defn}
\label{def:validity}
An IM $y \mapsto (\lPi_y, \uPi_y)$ is {\em valid} with respect to $\uprob_{Y,\Theta}$ if 
\begin{equation}
\label{eq:validity}
\uprob_{Y,\Theta}\{ \uPi_Y(A) \leq \alpha, \, \Theta \in A \} \leq \alpha, \quad \text{for all $A \subseteq \TT$ and all $\alpha \in [0,1]$}. 
\end{equation}
There's an equivalent condition in terms of the lower probability component of the IM output, but this won't be needed below. 
\end{defn}

The intuition behind this definition is that the joint event, ``$\uPi_Y(A) \leq \alpha, \Theta \in A$,'' is one where there's risk of erroneous inference, so the bound \eqref{eq:validity} ensures that the data analyst can control the $\uprob_{Y,\Theta}$-probability of this undesirable event.  For the vacuous-prior cases considered in \citet{imchar} and earlier references, the validity property in \eqref{eq:validity} reduces to the following:
\begin{equation}
\label{eq:valid.vacuous}
\sup_{\theta \in A} \prob_{Y|\theta}\{ \uPi_Y(A) \leq \alpha \} \leq \alpha, \quad \text{for all $A \subseteq \TT$ and all $\alpha \in [0,1]$}. 
\end{equation}
The condition in the above display implies that the IM-based procedures satisfy all the desired statistical properties.  In the more general case considered in Part~I, a stronger notion of validity was needed to establish some of the desired statistical properties.  In this paper, I focus exclusively on strongly valid IMs. 

\begin{defn}
\label{def:strong.validity}
An IM $y \mapsto (\lPi_y, \uPi_y)$, with contour function $\pi_y(\theta) = \uPi_y(\{\theta\})$, is {\em strongly valid} with respect to $\uprob_{Y,\Theta}$ if
\begin{equation}
\label{eq:strong.validity}
\uprob_{Y,\Theta}\{ \pi_Y(\Theta) \leq \alpha \} \leq \alpha, \quad \alpha \in [0,1].
\end{equation}
\end{defn}

For an explanation of why Definition~\ref{def:strong.validity} is stronger than Definition~\ref{def:validity}, see Part~I and \eqref{eq:uniform} below.  As indicated above, my focus here will be on IMs whose output $(\lPi_y, \uPi_y)$ is {\em consonant} \citep[e.g.,][Ch.~10]{shafer1976}.  That is, there exists a function $\pi_y$, called the {\em contour}, such that $\sup_{\theta \in \TT} \pi_y(\theta)=1$ for each $y \in \YY$ and the upper probability 
\[ \uPi_y(A) = \sup_{\theta \in A} \pi_y(\theta), \quad A \subseteq \TT; \]
the lower probability is determined by conjugacy.  The resulting upper probability is $\infty$-alternating---hence, coherent---and, moreover, has the special structure of a possibility measure \citep[e.g.,][]{dubois.prade.book, hose2022thesis, imposs}.

It's worth saying a few words here about the relative simplicity of possibility measures and their connections to statistical inference.  As is clear from the above display, a possibility measure, $\uPi_y$, is fully determined by its contour function, $\pi_y$.  The latter is an ordinary map from $\TT$ to $[0,1]$, which is much simpler than the former which is a map from $2^\TT$ to $[0,1]$.  This is analogous to the simplification that occurs in ordinary/precise probability when calculations can be carried out using a probability mass or density function.  Possibility measures are the only coherent imprecise probabilities with such a simple characterization.  Beyond their simplicity, some similarities between possibility measures and p-values are apparent---if not from the discussion above, then this will become more clear in the following pages.  In particular, possibilities measures have deep connections to both p-values and confidence regions; see Remark~\ref{re:inclusion} below. Fisher makes several remarks about the distinction between p-values/significance tests and probability, e.g., 
\begin{quote}
{\em It is more primitive, or elemental than, and does not justify, any exact probability statement about the proposition.} \citep[][p.~46]{fisher1973}
\end{quote}
But Fisher doesn't offer a concrete alternative to probability as a means for quantifying uncertainty for inference.  Later,\footnote{The quote from Fisher is from the 1973 edition of {\em Statistical Methods for Scientific Inference}, but his ideas and arguments were developed more than 20 years earlier.} Shackle, who was among the first to consider formal alternatives to probability for uncertain reasoning in general, writes  
\begin{quote}
{\em Possibility is an entirely different idea from probability, and it is sometimes, we maintain, a more efficient and powerful uncertainty variable, able to perform semantic tasks which the other cannot.} \citep[][p.~103]{shackle1961}
\end{quote}
How might the state of the foundations of statistics be different had the Fisher--Shackle connection been made earlier?  Better late than never.  There are even important connections between possibility theory and Fisher's fiducial inference (see the brief discussion on page \pageref{page:fiducial.in.credal} below) but I'll save these details for another occasion.

\section{Background-related FAQs}
\label{S:faq}

In this section, I address two ``frequently asked questions'' related to the basic setup in the previous sections.  Specifically, the two questions I address here concern my choice to focus on IMs that are (a)~strongly valid and (b)~consonant.

\subsection{Why strong validity?}

The strong validity property is admittedly very specific and, as will be shown below, leads to some restriction on the mathematical form of the IMs under consideration. So it's perfectly natural to ask for my rationale behind this choice.  In the vacuous prior case, validity and strong validity are equivalent so there was no need to place any special emphasis on the strong version.  In Part~I of this series, which was my first serious consideration of the effects of incorporating partial prior information, I wasn't advocating for one property over another, so again there was no need to justify an emphasis on the strong version.  But it's important that I address this question here since I'm restricting my attention to strongly valid IMs.  

Validity itself ensures that probabilistic reasoning, i.e., basing inferences on the magnitudes of the IM's lower and upper probabilities, can be done in a simple, natural, and reliable way.  But there are other summaries that one would surely want to extract from the IM, and these summaries should be reliable too.  In particular, if valid IMs don't readily provide reliable confidence regions, then users will inevitably look for alternative frameworks.  In the vacuous-prior case, where validity and strong validity are equivalent, I've already established a characterization \citep[][Theorem~6]{imchar} of confidence regions in terms of (strongly) valid IMs, so, in that case, there's nowhere else for users to look.  In the more general case considered in Part~I, where the two notions of validity aren't equivalent, there are challenges that arise with the construction of confidence regions based on IMs that are valid but not strongly valid, which I explain next.  This is my motivation for focusing on strong validity.

Consider an IM whose upper probability output is $\uGamma_y$, and suppose that it's valid but not strongly valid.  Next, for each $y \in \YY$, consider a collection of {\em credible regions}, $\{C_\alpha(y): \alpha \in [0,1]\} \subseteq 2^\TT$, derived from $\uGamma_y$, with the following basic structure:
\begin{itemize} 
\item it's nested in the sense that $C_\alpha(y) \subseteq C_\beta(y)$ for all $\alpha,\beta \in [0,1]$ with $\alpha \geq \beta$; 
\vspace{-2mm}
\item and $\uGamma_y\{C_\alpha(y)^c\} \leq \alpha$ for all $\alpha \in [0,1]$. 
\end{itemize} 
These are very reasonable properties, satisfied by, say, the standard highest-posterior-density credible regions commonly used in (precise-prior) Bayesian analysis.  Beyond these basic structural properties, it's expected that the credible regions are also confidence regions in the sense that 
\begin{equation}
\label{eq:non.coverage}
\uprob_{Y,\Theta}\{ C_\alpha(Y) \not\ni \Theta\} \leq \alpha, \quad \text{for all $\alpha \in [0,1]$}. 
\end{equation}
When the credible regions satisfy \eqref{eq:non.coverage} I'll say that they're {\em credibility--confidence calibrated}.  This calibration between the IM's output and the posited model $\uprob_{Y,\Theta}$ is not automatic.  But this is roughly the kind of calibration that validity aims to achieve, so there's good reason to expect that the confidence property is possible.

For concreteness, consider the generalized Bayes IM \citep[e.g.,][Ch.~6.4]{walley1991}, where $\uGamma_y$ is the upper envelope of the collection of (precise) Bayesian posterior probabilities corresponding to the credal set of (precise) priors.  This was shown in Part~I (Corollary~3) to be valid relative to the posited partial-prior Bayes model.  To my knowledge, there's no general guarantee that a credible region extracted from a valid IM would satisfy \eqref{eq:non.coverage}, but this can be established for credible regions extracted from the generalized Bayes IM; see the discussion following Corollary~3 in Part~I.  

The above credibility--confidence calibration is a desirable property but, unfortunately, it's not at all clear how to determine the generalized Bayes IM's credible regions.  Intuitively, for each $(\alpha,y)$, the set $C_\alpha(y)$ would be defined as the ``smallest'' $A$ such that $\lGamma_y(A) \geq 1-\alpha$.  In the single-prior Bayes case, the highest-posterior-density regions are known to be size-optimal, but there's no analogous result for the generalized Bayes case.  For strongly valid IMs, the extraction of credible regions is straightforward and the credibility--confidence calibration is automatic, as I'll show in Section~\ref{SSS:statistical}.  

Continuing here with my analysis, suppose we have a valid IM $\uGamma_y$ with credible regions, $\{C_\alpha(y): \alpha \in [0,1]\}$, that satisfy \eqref{eq:non.coverage}.  Define the function
\[ \pi_y(\theta) = \sup\{\alpha \in [0,1]: C_\alpha(y) \ni \theta\}, \quad \theta \in \TT. \]
Since the credible regions are nested, it follows that $\sup_\theta \pi_y(\theta) = 1$ for each $y$.  In that case, I can define a (consonant) IM via the rule 
\[ \uPi_y(A) = \sup_{\theta \in A} \pi_y(\theta), \quad A \subseteq \TT. \]
This IM is strongly valid because 
\[ \uprob_{Y,\Theta}\{ \pi_Y(\Theta) \leq \alpha \} = \uprob_{Y,\Theta}\{ C_\alpha(Y) \not\ni \Theta \} \leq \alpha, \]
where the inequality follows from the assumed credibility--confidence calibration \eqref{eq:non.coverage}.  Moreover, the inferences drawn by the derived strongly valid IM $\uPi_y$ are similar to those drawn by the primitive valid IM $\uGamma_y$, e.g., 
\begin{itemize}
\item if $A \subseteq C_\alpha(y)^c$, then both $\uGamma_y(A) \leq \alpha$ and $\uPi_y(A) \leq \alpha$, and 
\vspace{-2mm}
\item if $A \supseteq C_\alpha(y)$, then both $\lGamma_y(A) \geq 1-\alpha$ and $\lPi_y(A) \geq 1-\alpha$.
\end{itemize}
The point is that for many relevant assertions concerning $\Theta$, the two IMs won't differ substantially in their conclusions.  So, given a valid IM with credible regions that satisfy \eqref{eq:non.coverage}, one can construct a strongly valid IM that has the same credible regions and no substantial difference in inferences.  From this observation, and the fact that strongly valid IMs lead to a more direct construction of confidence regions, there's no obvious reason not to restrict attention to strongly valid IMs.  

There is one potential concern about deviating from the generalized Bayes IM that's valid but not strongly valid: coherence in the sense of \citet[][Sec.~6.5]{walley1991}, which concerns the preservation of internal rationality when updating prior to posterior beliefs, is (basically) only achieved by generalized Bayes.  As I'll discuss in Section~\ref{SSS:behavioral} below, it turns out that {\em virtually} all of the coherence properties achieved by generalized Bayes can also be achieved by a strongly valid IM.  In addition to strong validity and virtually all the same coherence properties, the IMs that I'm proposing here are considerably more efficient than those that can be built using generalized Bayes; see, e.g., Figure~\ref{fig:walley.binom}.

\subsection{Why consonance?}

It's imperative that I make the following point clear.  My choice to focus on consonant IMs is not arbitrary, not motivated by simplicity, not because \citet[][Ch.~10]{shafer1976} said it's reasonable, and not a restriction in any way.  As I show below, the only efficient way to achieve the desired strong validity property is with a consonant IM.  

Recall the strong validity property in \eqref{eq:strong.validity}, where $\pi_y(\theta) = \uPi_y(\{\theta\})$ is the contour function.  An interpretation of strong validity is that the random variable $\pi_Y(\Theta)$ is stochastically no smaller than $\unif(0,1)$ under any joint distribution $\prob$ for $(Y,\Theta)$ in the credal set $\cred(\uprob)$.  This requires that the function $(y,\theta) \mapsto \pi_y(\theta)$ can take values arbitrarily close to 1, but non-consonant capacities would typically have contour functions that are bounded away from unity.  For example, the IM construction in \citet{walley2002} is ingenious, but he can't rule out the possibility that his contour function is bounded away from 1,\footnote{The example in his Section~4.4 is the easiest one to see this. That same example also shows that Walley's confidence intervals can suffer from a loss of efficiency compared to classical intervals, though he argues there's good reason for this; see Section~\ref{SS:vacuous} below.} so the result he proves isn't quite as strong as strong validity. 

As the previous paragraph indicates, while most capacities have contours bounded away from unity, there are some non-consonant upper probabilities for which the contour function can take values arbitrarily close to 1. These correspond to capacities determined by what \citet{dubois.prade.1990} call {\em consistent} random sets, i.e., those whose focal elements have non-empty intersection.  It turns out, however, that I can safely dismiss those IMs with the structure of a capacity induced by a consistent random set because there always exists a consonant IM that's at least as efficient.  

The following is a variation on Proposition~1 in \citet{dubois.prade.1990}; see, also, \citet{dubois.prade.1986.set}.  It also closely resembles Theorem~4.3 in \citet{imbook} which, in the context of random sets, shows that there's no loss of generality or efficiency in focusing on those with nested focal elements. It can be viewed as a sort of {\em complete-class theorem} that says for any strongly valid IM there exists another that's consonant and no less efficient in the sense that the latter's contour is pointwise smaller than the former.  Equivalently, the latter is no less specific than the former, in the sense of, e.g., \citet{dubois.prade.1986}, so would be preferred; see, also, the Principle of Expressiveness in Section~\ref{S:blocks} below.  

\begin{lem}
\label{lem:complete.class}
For any strongly valid IM with upper probability $\uGamma_y$, there exists a consonant IM with upper probability $\uPi_y$ such that $\uPi_y(A) \leq \uGamma_y(A)$ for all $A \subseteq \TT$. 
\end{lem}

\begin{proof}
For a given $\uGamma_y$, note that its contour function $\gamma_y$ satisfies $\sup_{\theta \in \TT} \gamma_y(\theta) = 1$ because it's induced by a consistent random set.  Moreover, for general capacities, with $\gamma_y(\theta) = \uGamma_y(\{\theta\})$, monotonicity implies that 
\[ \sup_{\theta \in A} \gamma_y(\theta) \leq \uGamma_y(A), \quad \text{for all $A \subseteq \TT$}. \]
Define $\uPi_y$ to be consonant with contour $\pi_y = \gamma_y$. Then the above display implies $\uPi_y(A) \leq \uGamma_y(A)$ for all $A$, just as the theorem claims.
\end{proof}

The above result implies that it suffices to focus on strongly valid IMs that are consonant.  It also shouldn't be ignored that consonant IMs are among the simplest computationally.  So, there's a ``best of both worlds'' conclusion that can be made here about consonant IMs, in other words, {\em consonance is king}.  The remainder of the paper, therefore, focuses on strategies for constructing strongly valid, consonant IMs. 

A related ``frequently asked question'' that I've heard \citep[e.g.,][]{cui.hannig.im} concerns a unique feature of possibility measures, namely, that the lower probability assigned to $A$ or to $A^c$ must be 0 for all $A$.  In other words, the lower probability being positive implies the upper probability is 1 and, in the other direction, the upper probability being less than 1 implies the lower probability is 0.  This property is well-known and is discussed in \citet[][p.~221]{shafer1976} and even in \citet{shackle1961}.  Then the question I've been asked is this: if $(\lPi_y, \uPi_y)$ is a consonant IM, then 
\[ \text{{\em isn't the gap between $\lPi_y(A)$ and $\uPi_y(A)$ too wide?}} \]
My response is ``too wide'' compared to what?  The width of this gap characterizes the IM's degree of imprecision, which is indirectly related to its efficiency.  The goal is to be as precise/efficient as possible while maintaining validity, but it's not clear at all how much imprecision this requires.  So, how wide is ``too wide''?  This gap makes people uncomfortable because they're imagining that there's a ``true precise posterior probability'' that the IM is trying to capture between $\lPi_y$ and $\uPi_y$.  If that were the case, then a wide gap between the IM's lower and upper probabilities would be a concern.  But that's not what's going on here---{\em there's no ``true precise posterior probability'' that the IM is approximating/estimating}.  So, while I can interpret the IM output as a set of precise probabilities tucked between the lower and upper bounds, which is valuable for probabilistic reasoning, the question of whether and how tightly those bounds contain a ``true precise posterior probability'' isn't applicable.  As will be shown below, there are cases where a consonant IM, one whose gap is allegedly ``too wide,'' yields inference which is maximally efficient, so it must be that its width is {\em just right} and, therefore, this criticism doesn't hold up.  To demonstrate that the gap between the consonant IM's lower and upper bounds is ``too wide'' would require finding a non-consonant IM, with tighter bounds, for which even just validity can be proved.  In my experience, validity of non-consonant IMs is difficult to prove and, in the cases where it can be proved, the inference tends to be inefficient, suggesting that, in fact, it's the valid, non-consonant IM's gap that's demonstrably ``too wide.''

\section{Building blocks}
\label{S:blocks}

\subsection{Outer consonant approximations}

In this section, I'll consider a generic uncertain variable $X$ taking values in $\XX$.  For now the reader can think of $X$ as just a simplified notation for $(Y,\Theta)$.  But it should be clear, at least intuitively (see Section~\ref{SS:construction} for more details), that the roles played by $Y$ and by $\Theta$ are very different, so I can't capture the nuance of the statistical inference problem if I replace $(Y,\Theta)$ by $X$.  The details I want to present here, however, aren't specific to the statistical inference problem, so I think it helps to make the notation as simple as possible.  Moreover, seeing how the statistical inference case differs from this general case is potentially helpful for understanding and appreciating what follows. 

Let $\uprob=\uprob_X$ be a coherent upper probability, as described in Section~\ref{S:background} above, defined on $\XX$; there's a corresponding lower probability $\lprob$ but this won't be needed here. If $\uprob$ itself is consonant, then the details in the remainder of this section are trivial.  However, for the statistical inference applications I have in mind, the corresponding $\uprob$ will virtually never be consonant (see Section~\ref{S:general}), so the non-consonant case is of primary relevance.  Henceforth, I'll assume that $\uprob$ is not consonant.  

It will be of interest below to approximate the non-consonant $\uprob$ by another upper probability that is consonant.  My particular motivation for this is unique, but this question of approximating one imprecise probability by another simpler one is common in the literature.  Since possibility measures are consonant, and are the simplest of the imprecise probabilities, these are often the models that serve as the approximant; see, e.g., \citet{dubois.prade.1990}, \citet{dubois.etal.2004}, \citet{baroni2004}, and \citet{hose.hanss.2021}. To fix notation, let $\uOmega$ denote a generic consonant upper probability on $\XX$, intended to serve as an approximation of $\uprob$.  There are a variety of ways to formulate the approximation of $\uprob$ by a consonant $\uOmega$, but the one I specifically have in mind is motivated by core principles often appealed to in the imprecise probability literature.  My presentation of these core principles and related results---here and in the next subsection---closely follows that in Dominik Hose's PhD thesis \citep[][Sec.~2.3.2.1--2.3.2.2]{hose2022thesis}; I'm showing most of the details since few readers will have immediate access to Dr.~Hose's thesis.


The first principle is what Hose calls the {\em Principle of Representation}, based on a slight tweak of the fundamental premise \eqref{eq:dubois.quote}, i.e., 
\[ \text{``what is probable must be [possible].''} \]  
I'll state the principle first and then explain its meaning and implications. 

\begin{prinrep}
Choose $\uOmega$ such that $\cred(\uprob) \subseteq \cred(\uOmega)$.
\end{prinrep} 

A possibility measure $\uOmega$ that satisfies the above condition is called an {\em outer consonant approximation} of $\uprob$.  The reason for this name is that the statement ``$\cred(\uprob) \subseteq \cred(\uOmega)$,'' which of course means that any $\prob$ in $\cred(\uprob)$ is also in $\cred(\uOmega)$, implies that if $\prob \leq \uprob$, then $\prob \leq \uOmega$.  So $\uOmega$ is like an upper bound on $\uprob$.  The implication is, as the fundamental premise says, if $\uOmega$ is interpreted as a measure of possibility, then an event concerning the value of $X$ can be no less possible than it is probable.  Since $\uprob$ is presumably a description of what is known/believed about $X$, then a connection like the one in the Principle of Representation is needed in order for the uncertainty quantification based on (the subjective) $\uOmega$ to be relevant in real-world applications.  

If the reader accepts the Principle of Representation, so that the goal is to construct an outer approximation $\uOmega$ of $\uprob$, then the next obvious question is: how to do the construction?  The construction I adopt here, referred to as the {\em imprecise-probability-to-possibility transform} in \citet{hose.hanss.2020, hose.hanss.2021}, proceeds as follows.  Take $h: \XX \to \RR$ to be a measurable function and define the upper probability 
\[ \uOmega_h(E) = \uprob\{ h(X) \leq \textstyle\sup_E h \}, \quad E \subseteq \XX. \]
The corresponding lower probability, $\lOmega_h$, is defined via conjugacy, but won't be needed here. The claim is that $\uOmega_h$ is an outer consonant approximation of $\uprob$.  

\begin{prop}
\label{prop:outer}
For any mapping $h$ as described above, the upper probability $\uOmega_h$ is an outer consonant approximation of $\uprob$ in the sense that $\cred(\uOmega_h) \supseteq \cred(\uprob)$. 
\end{prop}

\begin{proof}
The first step is to show that $\uOmega_h$ is consonant.  This is more-or-less immediate just by looking at the contour function corresponding to $\uOmega_h$, i.e., 
\begin{equation}
\label{eq:contour}
\omega_h(x) := \uOmega_h(\{x\}) = \uprob\{ h(X) \leq h(x)\}, \quad x \in \XX. 
\end{equation}
Clearly $\omega_h(x)$ is in $[0,1]$ for all $x$; it remains to show that $\sup_{x \in \XX} \omega_h(x) = 1$.  That this supremum is $\leq 1$ is obvious.  To show the opposite inequality, and hence equality, take a particular $\prob$ in $\cred(\uprob)$ and notice that 
\begin{align*}
\sup_{x \in \XX} \omega_h(x) \geq \sup_{x \in \XX} \prob\{h(X) \leq h(x)\} = \prob\Bigl\{ h(X) \leq \sup_{x \in \XX} h(x) \Bigr\},
\end{align*}
where the last equality is due to monotonicity of $\prob$.  Since the right-hand side is clearly equal to 1, it follows that $\uOmega_h$ is consonant.  

To prove that $\uOmega_h$ is an outer approximation, it suffices to show that 
\begin{equation}
\label{eq:outer}
\uprob(E) \leq \uOmega_h(E), \quad \text{for all $E \subseteq \XX$}.
\end{equation}
Of course, $x \in E$ implies $h(x) \leq \sup_E h$ which, in turn, implies $E \subseteq \{x: h(x) \leq \sup_E h\}$.  Then \eqref{eq:outer} follows from the definition of $\Omega_h$ and monotonicity of $\uprob$. 
\end{proof}

Since $\uOmega_h$ is consonant, it's determined by its corresponding possibility contour, $\omega_h$, defined in \eqref{eq:contour} above, through the relation
\[ \uOmega_h(E) = \sup_{x \in E} \omega_h(x), \quad E \subseteq \XX. \]
It'll often be convenient to work directly with the contour function, $\omega_h$.  In particular, the following elementary property of the contour is of fundamental importance.  In fact, this property is the reason why consonance so important to the IM framework.  

\begin{prop}
\label{prop:valid}
If $\omega_h$ as defined above is measurable, then 
\begin{equation}
\label{eq:X.valid}
\uprob\{ \omega_h(X) \leq \alpha \} \leq \alpha, \quad \text{for all $\alpha \in [0,1]$}. 
\end{equation}
\end{prop}

\begin{proof}
A short indirect proof of this result is given in Remark~\ref{re:inclusion} below.  For the sake of completeness, however, I'll also give a direct proof.  By coherence of $\uprob$, I can write
\[ \omega_h(x) = \uprob\{h(X) \leq h(x)\} = \sup_{\prob' \in \cred(\uprob)} \prob'\{h(X) \leq h(x)\}. \]
Then 
\begin{align*}
\uprob\{\omega_h(X) \leq \alpha\} & = \sup_{\prob \in \cred(\uprob)} \prob_X\Bigl[ \sup_{\prob' \in \cred(\uprob)} \prob_{X'}'\{h(X') \leq h(X)\} \leq \alpha \Bigr] \\
& \leq \sup_{\prob \in \cred(\uprob)} \prob_X\bigl[ \prob_{X'}\{h(X') \leq h(X)\} \leq \alpha \bigr],
\end{align*}
where the subscripts ``$X$'' and ``$X'$'' indicate which random variable the probability calculation is with respect to.  The inequality follows by replacing the supremum over $\prob'$ with $\prob' \equiv \prob$.  The $\prob$-probability on the last line is $\leq \alpha$ by standard results from probability theory, i.e., since the distribution function of a random variable applied to itself is stochastically no smaller than $\unif(0,1)$, so the supremum over $\prob$ is $\leq \alpha$ too.  
\end{proof}

\begin{remark}
\label{re:inclusion}
It turns out that Propositions~\ref{prop:outer} and \ref{prop:valid} are equivalent, thanks to a well-known characterization of $\cred(\uOmega_h)$ in terms of sub-level sets of the contour $\omega_h$ \citep[e.g.,][]{cuoso.etal.2001, dubois.etal.2004}.  Indeed, these characterizations state that  
\[ \prob \in \cred(\uOmega_h) \iff \prob\{\omega_h(X) \leq \alpha\} \leq \alpha \quad \text{for all $\alpha \in [0,1]$}. \]
So, if $\cred(\uprob) \subseteq \cred(\uOmega_h)$, then \eqref{eq:X.valid} holds and, hence, Proposition~\ref{prop:valid}.  Conversely, if \eqref{eq:X.valid} holds, then every $\prob$ in $\cred(\uprob)$ must also be in $\cred(\uOmega_h)$, which implies Proposition~\ref{prop:outer}.  
\end{remark}

The reader surely sees a connection between the property in Proposition~\ref{prop:valid} and the strong validity property in Definition~\ref{def:strong.validity} above.  That this connection is natural and immediate is largely the reason why consonance in general and this construction in particular is fundamental to the IM formulation.  This also hints at  similarities between the so-called {\em Validity Principle} in \citet{imbook, sts.discuss.2014} and the Principle of Representation.  But the direct connection between the result in Proposition~\ref{prop:valid} and the IM's strong validity property is not immediate---this is related to the nuance that's missed when $(Y,\Theta)$ is replaced by $X$---so some effort is needed to tie everything together; see Section~\ref{S:general}.

\subsection{Minimal outer consonant approximations}
\label{SS:minimal}

Since the goal is approximating $\uprob$ by a simpler, consonant $\uOmega$, it only makes sense to aim for the ``best possible'' approximation, in some sense.  There are too many degrees of freedom in the specification of an outer consonant approximation so formulating this as an optimization problem that can be readily solved is out of reach.  Instead, I'll again proceed by following some core principles.  

The next principle is what Hose calls the {\em Principle of Expressiveness}, derived from another small tweak of the fundamental premise \eqref{eq:dubois.quote}, i.e., 
\[ \text{``what is [less] probable must be [less possible].''} \]
This is akin to the {\em Principle of Maximum Specificity}  \citep[e.g.,][]{dubois.prade.1986}. The point is that the outer approximation should be as tight as possible.  In particular, a vacuous or near-vacuous $\uOmega$ that ignores all or most of the structure in $\uprob$ is no good. 

\begin{prinexp}
If $\uOmega_1$ and $\uOmega_2$ are two outer consonant approximations of $\uprob$ with $\cred(\uOmega_1) \subseteq \cred(\uOmega_2)$, then $\uOmega_1$ is preferred to $\uOmega_2$. 
\end{prinexp}

Consider the class of outer approximations $\uOmega_h$ indexed by the function $h$.  The role played by $h$ is to determine a (partial) ordering on the domain $\XX$, what I refer to here as a {\em plausibility order}.  It's easy to see \citep[e.g.,][Prop.~12]{hose2022thesis} that the contour $\omega_h$ of $\uOmega_h$ has the same plausibility order as $h$, that is, 
\[ h(x) > h(x') \implies \omega_h(x) > \omega_h(x'). \]
There are many consonant $\uOmega$'s whose contour has the same plausibility order determined by $h$; denote the collection of those consonant upper probabilities as $\mathbb{O}_h$, where ``O'' stands for ``Omega.''  For example, if $g$ equals $h$ composed with a monotone transformation, then $\uOmega_g \in \mathbb{O}_h$, but there are many others as well.  

\begin{prop}
\label{prop:order}
Fix $h$.  If $\uOmega \in \mathbb{O}_h$, then $\cred(\uOmega_h) \subseteq \cred(\uOmega)$.
\end{prop}

\begin{proof}
Take $\uOmega \in \mathbb{O}_h$ and let $\omega$ be its contour function.  Fix any $x' \in \XX$, and set $\alpha = \omega(x') \in [0,1]$.  Define the sub-level sets of $\omega$:
\[ S_\beta = \{x: \omega(x) \leq \beta\}, \quad \beta \in [0,1]. \]
By the well-known results described in Remark~\ref{re:inclusion}, it follows that $\uprob(S_\alpha) \leq \alpha$.  By definition, $\omega$ has the same plausibility order as $h$, so that 
\[ \{x: h(x) \leq h(x')\} \subseteq \{x: \omega(x) \leq \omega(x')\} = S_\alpha. \]
Therefore, by monotonicity of $\uprob$, 
\[ \omega_n(x') := \uprob\{h(X) \leq h(x')\} \leq \uprob(S_\alpha) \leq \alpha = \omega(x'). \]
It follows that, for all $\uOmega \in \mathbb{O}_h$, with contour function $\omega$, the pointwise inequality holds: $\omega_h \leq \omega$.  This implies $\uOmega_h$ is no less specific than $\uOmega$ and, therefore, $\cred(\uOmega_h) \subseteq \cred(\uOmega)$. 
\end{proof}

Therefore, if the plausibility order $h$ is given, then the minimal outer approximation as called for by the Principle of Expressiveness is given by $\uOmega_h$.  More generally, Proposition~\ref{prop:order} implies that it's enough to consider a choice of $\uOmega$ from the class $\{\uOmega_h\}$ indexed by $h$.  The downside is that, in general, there's no $h$ that leads to the ``best'' $\uOmega_h$.  To resolve this ambiguity concerning the choice of $h$, again I'll appeal to some core principles.

The third principle, what Hose calls the {\em Principle of Plausibility}, is exactly the fundamental premise \eqref{eq:dubois.quote}: ``what is probable must be plausible.''  The point is that the plausibility order in $\uOmega$ should represent that of $\uprob$.  

\begin{prinpl}
If $\eta: \XX \to \RR$ represents the plausibility order inherent in $\uprob$, then $\uOmega_\eta$ is preferred to any other $\uOmega_h$.
\end{prinpl}

If $\uprob$ were an ordinary/precise probability distribution $\prob$, then one would typically take the function $\eta$ above to be the probability mass or density function of $\prob$.  Then the $\uOmega_h$ obtained by following the Principle of Plausibility, with $h=\eta$, would be the\footnote{That it's ``the'' unique minimal outer consonant approximation requires that $p$ be unimodal; otherwise, there may be other choices of $h$ for which the minimality holds.} {\em minimal outer approximation} of $\prob$ as shown in \citet{dubois.etal.2004}.  Intuitively, the choice of $h$ as the density/mass function as suggested above is not surprising.  It's closely related to the familiar result in probability theory that the smallest sets (in terms of counting/Lebesgue measure) having a fixed $\prob$-probability content are the mass/density function level sets.  And for the statistical applications I have in mind here, this sheds light on the fundamental role played by the likelihood function.  

More generally, I propose to define the plausibility order of $\uprob$ by
\begin{equation}
\label{eq:eta}
\eta(x) =  \uprob(\{x\}), \quad x \in \XX. 
\end{equation}
Aside from reducing to the precise-probability case described above (at least for discrete $\prob$), this choice can be further justified as follows.  Consider the case where $\uprob$ itself is consonant.  This case isn't practically relevant, but it's the one case where the consonant approximation could be exact and, therefore, makes a good test case: if the proposed strategy is any good, then the approximation $\uOmega_\eta$ ought to match $\uprob$ exactly.  

\begin{prop}
\label{prop:exact}
If $\uprob$ is consonant with contour $\eta(x) = \uprob(\{x\})$, then the approximation $\uOmega_\eta$ is exact in the sense that $\omega_\eta(x) = \eta(x)$ for all $x$.
\end{prop}

\begin{proof}
First, since $\uprob$ is a possibility measure, if $\eta$ is defined as in \eqref{eq:eta}, then it's the possibility contour of $\uprob$, so the associated sub-level sets are 
\[ S_\alpha = \{x: \eta(x) \leq \alpha\}, \quad \alpha \in [0,1]. \]
Now it's immediately clear that
\[ \omega_\eta(x) = \uprob( S_{\eta(x)} ) = \sup_{x' \in S_{\eta(x)}} \eta(x') \leq \eta(x). \]
Second, since $\uprob$ is coherent, it follows that
\[ \eta(x) = \uprob(\{x\}) = \sup_{\prob \in \cred(\uprob)} \prob(X=x) \leq \sup_{\prob \in \cred(\uprob)} \prob\{ \eta(X) \leq \eta(x) \} = \omega_\eta(x). \]
Therefore, $\omega_\eta(x) = \eta(x)$, so the approximation is exact.
\end{proof}

Again, the case where $\uprob$ itself is consonant isn't of practical interest.  I'm making the general proposal to use the outer consonant approximation 
\[ \uOmega_\eta, \quad \text{with} \quad \eta(\cdot) = \uprob(\{\cdot\}), \]
regardless of what structure $\uprob$ has.  Proposition~\ref{prop:exact} helps to justify this proposal in the following sense: the only case in which a consonant approximation could be exact is when $\uprob$ is itself consonant and, in such cases, my proposal is exact.  Besides being based on some basic, fundamental principles, I honestly don't see any other viable alternatives.

\section{General IM framework}
\label{S:general}

\subsection{Construction}
\label{SS:construction}

Consider a general imprecise probability model $\uprob_{Y,\Theta}$ for $(Y,\Theta)$, where $Y$ is observable and $\Theta$ is the quantity of interest.  The IM construction and strong validity result that follows is based on having a ``joint distribution'' for $(Y,\Theta)$ that's consonant, and I emphasized before that $\uprob_{Y,\Theta}$ itself would virtually never be consonant.  To see why this is so, consider the typical case---discussed in detail in Section~\ref{S:precise}---where $\uprob_{Y,\Theta}$ corresponds to an ordinary or precise probability distribution for $Y$, given $\Theta=\theta$, and an imprecise partial prior for $\Theta$.  For simplicity, also suppose that $Y$ is discrete.  Then the contour for $(Y,\Theta)$ would be the conditional mass function for $Y$ times the prior contour and it's (virtually) impossible for this product to equal 1 for some pair $(y,\theta)$.  Therefore, the outer consonant approximations described in Section~\ref{S:blocks} would be relevant here.  

That outer consonant approximation developed above can be applied directly in this particular case by interpreting $(Y,\Theta)$ here as the $X$ there.  That is, get an outer consonant approximation $\uOmega_h$ of $\uprob_{Y,\Theta}$ on $\YY \times \TT$ by first defining the contour function 
\begin{equation}
\label{eq:omega}
\omega_h(y,\theta) = \uprob_{Y,\Theta}\{ h(Y,\Theta) \leq h(y,\theta)\}, \quad (y,\theta) \in \YY \times \TT, 
\end{equation}
and then setting 
\[ \uOmega_h(E) = \sup_{(y,\theta) \in E} \omega_h(y,\theta), \quad E \subseteq \YY \times \TT, \]
where $h: \YY \times \TT \to \RR$ is a measurable function.  Guidance on the choice of $h$ was given above, but I'll hold off on applying that guidance for now.  This yields a good/principled outer consonant approximation of the joint distribution of $(Y,\Theta)$.  

Remember that the goal is to make inference on $\Theta$ for fixed $Y=y$, so the joint distribution of $(Y,\Theta)$ isn't of primary interest.  This is the nuance that's lost as a result of treating $(Y,\Theta)$ as a generic $X$.  From the ability to construct an outer approximation of the joint distribution, there are at least two ways to proceed with the construction of a consonant IM for $\Theta$ given $Y=y$:
\begin{itemize}
\item {\em Conditioning.} Carry out the construction of $\uOmega_h$ as described above, with the $h$ determined by those motivating principles, and then apply a suitable conditioning operation to transform that ``joint possibility distribution'' for $(Y,\Theta)$ to a ``conditional possibility distribution'' for $\Theta$, given the observed $Y=y$. 
\vspace{-2mm}
\item {\em Focusing.} Carry out the construction of $\uOmega_h$ as described above, but with $h$ chosen to satisfy those motivating principles {\em as a function of $\theta$ for each fixed $y$}.  This determines a collection of outer approximations indexed by $y$, and then I just choose---or focus on---the one corresponding to the observed $Y=y$. 
\end{itemize} 

The conditioning strategy, on the one hand, might seem more natural, but there are issues.  Even in situations involving precise probabilities, conditioning can be problematic; Borel's paradox \citep[e.g.,][Ch.~15]{jaynes2003} is an example of this.  Related difficulties carry over into the imprecise probability setting: there are different strategies for carrying out the relevant conditioning operations, and certain ``paradoxes'' \citep[e.g.,][]{gong.meng.update} can pop up.  Even just in possibility theory, there are a number of different proposals of conditioning operations---e.g., \citet{zadeh1978}, \citet{hisdal1978}, \citet{nguyen1978.conditional}, \citet{dubois.prade.1984, dubois.prade.1990.conditioning}, and \citet{cooman.poss2}---that each have their own emphasis and priorities.  Most importantly, this process of first getting a joint approximation and then conditioning doesn't provide the answers that I'm looking for; see Remark~\ref{re:construction}.2.  So, if there's an alternative to conditioning, then I'm open to it.  

The focusing strategy, on the other hand, starts with the idea of reinterpreting the joint contour $(y,\theta) \mapsto \omega_h(y,\theta)$ as a collection of contours $\theta \mapsto \omega_h(y,\theta)$ indexed by $y$.  If this new interpretation could be justified, then the IM construction amounts to, as explained above, choosing---or {\em focusing on}---the contour corresponding to the observed $Y=y$.  In particular, no conditioning operation is required in this process.  The challenge, however, is that there's no immediate guarantee that $\theta \mapsto \omega_h(y,\theta)$ would meet the conditions of a possibility contour for each fixed $y$.  So, for the focusing strategy to be successful, some care is needed in the choice of $h$. Towards this, reconsider the situation where quantification of uncertainty about a generic $X$ was the goal, and recall that the Principle of Plausibility suggested letting the plausibility order in $\uprob_{Y,\Theta}$, expressed in terms of $x \mapsto \uprob(\{x\})$, determine the plausibility order in the consonant approximation.  In the present case, with an eye towards focusing, I want $\omega_h(y,\theta)$ to reflect a meaningful plausibility order in $\theta$ relative to $y$.  This suggests the choice 
\[ \eta(y,\theta) = c(y) \, \uprob_{Y,\Theta}(\{y,\theta\}), \quad (y,\theta) \in \YY \times \TT, \]
where the proportionality constant (in $\theta$) is a normalizing factor that's intended to cancel out any $y$-specific effects on the plausibility order induced directly by $\uprob_{Y,\Theta}(\{y,\theta\})$.  There is some ambiguity in how this normalizing factor is chosen but I propose 
\[ c(y) = \Bigl[\sup_{t \in \TT} \uprob_{Y,\Theta}(\{y,t\}) \Bigr]^{-1}, \quad y \in \YY. \]
First, this choice clearly cancels out all the $y$-specific effects in $\eta$, letting the plausibility order be determined by the relationship between $\theta$ and $y$ alone.  Second, as is easy to verify, this choice ensures that $\theta \mapsto \omega_\eta(y,\theta)$ is a possibility contour for each fixed $y$.  For more on this choice of normalization, see Remark~\ref{re:normalization} below. 

To summarize, my proposed IM construction is as follows.  Start with the given imprecise ``joint distribution'' $\uprob_{Y,\Theta}$ of $(Y,\Theta)$ and define the plausibility ordering 
\begin{equation}
\label{eq:eta.y}
\eta(y,\theta) = \frac{\uprob_{Y,\Theta}(\{y,\theta\})}{\sup_{t \in \TT} \uprob_{Y,\Theta}(\{y,t\})}, \quad (y,\theta) \in \YY \times \TT. 
\end{equation}
Next, define the outer consonant approximation $\uOmega_\eta$, with possibility contour $\omega_\eta$, as defined in Section~\ref{S:blocks} above.  Finally, define the IM by focusing the above construction on the observed $Y=y$; in particular, set the possibility contour for $\Theta$, given $Y=y$, as 
\begin{align}
\pi_y(\theta) & = \omega_\eta(y,\theta) \notag \\
& = \uprob_{Y,\Theta}\{ \eta(Y,\Theta) \leq \eta(y,\theta)\}, \quad \theta \in \TT, \label{eq:pi.y}
\end{align}
and the corresponding consonant upper probability 
\begin{equation}
\label{eq:consonant.upper}
\uPi_y(A) = \sup_{\theta \in A} \pi_y(\theta), \quad A \subseteq \TT. 
\end{equation}
It's important to point out the two key roles played by the ``partial prior'' for $\Theta$, which is baked into $\uprob_{Y,\Theta}$.  First, in \eqref{eq:eta.y} it's clear that those regions in $\TT$ that aren't supported by the partial prior are discounted.  This {\em regularization} effect makes $\eta(y,\theta)$ in the expression \eqref{eq:pi.y} is smaller for ``outlier'' $\theta$, which, in turn, at least intuitively, makes $\pi_y(\theta)$ smaller there too, hence a potential efficiency gain.  Second, the partial prior also has a {\em calibration} effect through the ``$\uprob_{Y,\Theta}$'' calculation in \eqref{eq:pi.y}, which is what leads to the IM's strongly validity properties.  

The notation ``$\uPi_y$'' above doesn't reflect the dependence on the plausibility order $\eta$, but I don't think this is necessary.  To see why, recall that the choice of plausibility order in \eqref{eq:eta} was well-motivated and, in fact, really the only viable option.  Here there's technically a choice of normalization but, again, what I've suggested in \eqref{eq:eta.y} above is the only choice that ensures $\theta \mapsto \pi_y(\theta)$ is a possibility contour function for each fixed $y$.  So the suggested $\eta$ in \eqref{eq:eta.y} is a central part of the IM construction, not an option whose particular choice needs to be highlighted in the notation.  Of course, the $\eta$ in \eqref{eq:eta.y} depends on $\uprob_{Y,\Theta}$ so its particular form will vary from one problem to the next.  Also, certain context-specific adjustments may be warranted for the sake of computational or statistical efficiency; that's the topic of Section~\ref{SS:efficiency} below. 

\begin{remark}
\label{re:construction}
There are a number of relevant technical points that can be made about the above construction, so I collect those points here. 
\begin{enumerate}
\item The above discussion focused on the case where $Y$ is discrete.  When $Y$ is continuous, then there's a risk that the contour $(y,\theta) \mapsto \uprob_{Y,\Theta}(\{y,\theta\})$ is identically 0; this risk is only serious in the (common) case involving a precise statistical model for $(Y \mid \Theta)$.  A strong argument can be made that all data is discrete, that the continuous models for $Y$ are all just mathematically convenient approximations, so the discrete-$Y$ case is all that really matters at a foundational level.  But that's not fully satisfactory since continuous-$Y$ models are very common in practice.  Fortunately, it's easy to extend the principled approach here in the discrete-$Y$ case to the continuous-$Y$ case by analogy: just use the continuous density function where the discrete mass function would go; see Equation \eqref{eq:baseline} below. 
\vspace{-2mm} 
\item I made the claim that a conditioning-based IM construction can lead to some unexpected and sub-optimal results.  Here's a quick illustration.  Suppose $Y$, given $\Theta=\theta$, has a $\gam(n, \theta)$ distribution with known shape $n$ and unknown scale $\theta$; suppose the prior for $\Theta$ is vacuous.  A conditioning-based procedure would use $\eta(y,\theta)$ equal to the likelihood function to get the joint possibility contour $\omega_\eta(y,\theta)$, then get the corresponding conditional contour for $\Theta$, given $Y=y$, via, e.g., the formula in Proposition~4.6 of \citet{cooman.poss2}.  Alternatively, the focusing-based approach takes $\eta(y,\theta)$ equal to the relative likelihood.  In both cases, there's a pivotal structure in $\eta(Y,\Theta)$, so no sophisticated Choquet integration is required.  Figure~\ref{fig:exp.contour.compare} shows a plot of the two contour functions based on the conditioning and focusing approaches, where $Y \approx 30$ is the observed value and $n=5$ is known.  Notice that the two contours don't have the same core, i.e., the maximum plausibility values are different.  The focusing- or relative likelihood-based contour is maximized at the maximum likelihood estimator $\hat\theta = Y/n \approx 6$.  Surprisingly, the conditioning- or likelihood-based contour is maximized at a different point, not the maximum likelihood estimator.  Since there's no justification for the maximum plausibility to be attained at a point other than the maximum likelihood estimator, I conclude that the conditioning-based IM construction is generally inferior to the recommended focusing-based construction. 
\vspace{-2mm}
\item Note that the $\uOmega_\eta$ just constructed is still an outer consonant approximation of $\uprob$.  It's just been constructed in such a way that $\uPi_y(A) = \uOmega_\eta(\{y\} \times A)$ is a consonant upper probability on $\TT$ for each $y$.  This sheds light on why I call it ``focusing.'' 
\vspace{-2mm} 
\item The IM construction above is {\em not} the same as first getting a ``conditional'' upper probability for $\Theta$, given $Y=y$, via the generalized Bayes rule \citep[e.g.,][Ch.~6.4]{walley1991}, and then finding a corresponding outer consonant approximation.  To see this, consider the case of a precise statistical model for $Y$, given $\Theta=\theta$, and a vacuous prior on $\Theta$.  The upper conditional distribution for $\Theta$, given $Y=y$, would also be vacuous \citep[e.g.,][Theorem~4.8]{gong.meng.update} and, hence, so would the outer consonant approximation.  The IM would be valid but useless, thanks to the extreme inefficiency.  As I demonstrate in Section~\ref{SS:vacuous}, however, construction of a valid and efficient IM is possible using the proposed approach. 
\end{enumerate} 
\end{remark}

\begin{figure}[t]
\begin{center}
\scalebox{0.65}{\includegraphics{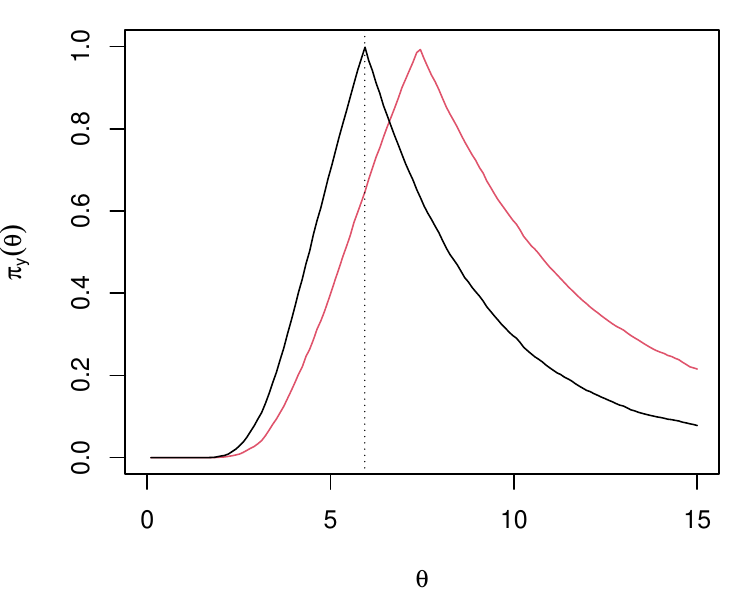}}
\end{center}
\caption{Plots of the possibility contour corresponding to the conditioning-based (red) and focusing-based strategies described in the text.  This is for an observation $Y \approx 30$ from a gamma distribution with shape $n=5$ and unknown scale $\Theta$.  Vertical line is at the maximum likelihood estimator, $\hat\theta = Y/n \approx 6$.}
\label{fig:exp.contour.compare}
\end{figure}


\subsection{Key properties}
\label{SS:properties}

There are a number of interesting and practically important properties that result from the above IM construction using outer consonant approximations.  Like in Part~I, here I'll split these into two categories: {\em statistical} and {\em behavioral} properties.  First, the statistical properties will closely resemble those commonly found in the frequentist (and Bayesian) statistics literature, i.e., error rate control for testing rules and confidence sets.  There are, however, some key differences between these results and those commonly found in the literature, deriving from my particular emphasis on an imprecise-probabilistic model formulation, which I'll explain.  Second, the behavioral properties are of a subjective nature and closely resemble those no-sure-loss/coherence properties emphasized in the imprecise probability literature.  This property is important because, at least intuitively, reliability and rationality should go hand-in-hand: if a framework for quantifying uncertainty is reliable/valid, then it shouldn't be irrational and, conversely, if the framework is rational, then it shouldn't be unreliable.  The results below together demonstrate that, indeed, the proposed IM framework achieves this dual reliability--rationality.  

\subsubsection{Statistical}
\label{SSS:statistical}

The driver behind the IM's statistical properties is the strong validity property that motivated my insistence that the IM have a consonant structure.  Since strong validity was a primary goal, it should be no surprise that this property holds.  And I basically already proved it in Proposition~\ref{prop:valid}. 

\begin{thm}
\label{thm:valid}
The IM defined above, with contour given by \eqref{eq:pi.y}, is strongly valid for inference on $\Theta$, relative to the posited model $(\lprob,\uprob)$, in the sense that 
\[ \uprob_{Y,\Theta}\{ \pi_Y(\Theta) \leq \alpha \} \leq \alpha, \quad \text{all $\alpha \in [0,1]$}. \]
\end{thm}

\begin{proof}
Since $\pi_y(\theta) = \omega_\eta(y,\theta)$ for the function $\eta$ defined in \eqref{eq:eta.y}, the result follows immediately from Proposition~\ref{prop:valid} above.
\end{proof}

As seen in Part~I and indicated above, the strong validity property immediately leads to some practically relevant error rate control properties of statistical procedures derived from the IM output, such as hypothesis tests and confidence regions.  

\begin{cor}
\label{cor:valid}
Let $\theta \mapsto \pi_y(\theta)$ be the IM's contour function as defined in \eqref{eq:pi.y} above, with $A \mapsto \uPi_y(A)$ the corresponding possibility measure.  Then the following properties hold.
\begin{enumerate}
\item For any $A \subseteq \TT$ and any $\alpha \in [0,1]$, the test 
\[ \text{{\em reject the hypothesis ``$\Theta \in A$'' if and only if $\uPi_Y(A) \leq \alpha$}} \] 
controls the upper false-rejection probability at level $\alpha$, i.e., 
\begin{equation}
\label{eq:size}
\uprob_{Y,\Theta}\{ \uPi_Y(A) \leq \alpha, \, \Theta \in A \} \leq \alpha. 
\end{equation}
\item For any $\alpha \in [0,1]$, the $100(1-\alpha)$\% plausibility region 
\begin{equation}
\label{eq:conf.region}
C_\alpha(y) = \{\theta \in \TT: \pi_y(\theta) > \alpha\}, \quad y \in \YY, 
\end{equation}
controls the upper non-coverage probability at level $\alpha$, i.e., 
\begin{equation}
\label{eq:coverage}
\uprob_{Y,\Theta}\{ C_\alpha(Y) \not\ni \Theta \} \leq \alpha. 
\end{equation}
\end{enumerate} 
\end{cor}

\begin{proof}
Both conclusions are immediate consequences of Theorem~\ref{thm:valid}.
\end{proof} 

It's important to emphasize that the results in Corollary~\ref{cor:valid} aren't the usual frequentist error rate control properties.  The difference is that the probability calculations are with respect to the posited (upper) joint distribution of $(Y,\Theta)$.  Of course, in the special case where the posited model is based on specification of a vacuous prior for $\Theta$, the above conclusions do reduce to the usual frequentist properties.  More specifically, in the vacuous prior case, the result in \eqref{eq:coverage} can be rewritten as
\begin{equation}
\label{eq:coverage.vac}
\sup_\theta \prob_{Y|\theta}\{ C_\alpha(Y) \not\ni \theta\} \leq \alpha, 
\end{equation}
which is the familiar textbook formula for controlling the non-coverage probability at level $\alpha$.  Beyond just being a mathematical generalization, this has practical consequences as well.  The point is that it makes no sense to require strong uniform bounds on the non-coverage probability as in \eqref{eq:coverage.vac} if genuine (albeit partial) prior information is available and can be encoded in $\uprob_{Y,\Theta}$.  This not only makes the conclusions drawn based on the analysis more meaningful, it also creates an opportunity for efficiency gain.  For example, the uniform bound in \eqref{eq:coverage.vac} is the most difficult to achieve, so, at least intuitively, $C_\alpha(y)$ would have to be relatively large (in some stochastic sense) in order to achieve it.  By incorporating available prior information into both the IM construction {\em and} the properties it's required to satisfy, the set $C_\alpha(y)$ has an opportunity to be smaller (in the same stochastic sense) and, therefore, can provide sharper and more efficient inference.  

To help shed light on why it would make sense to describe the set $C_\alpha(y)$ as a ``confidence'' region, I present an alternative definition in terms of the IM's lower probability $\lPi_y$. This depends crucially on the underlying consonance structure.  The set $C_\alpha(y)$ in \eqref{eq:conf.region} can be re-expressed as 
\[ C_\alpha(y) = \bigcap \{A \subseteq \TT: \lPi_y(A) \geq 1-\alpha\}. \]
In other words, $C_\alpha(y)$ represents the smallest $A \subseteq \TT$ to which the IM assigns at least $1-\alpha$ lower probability, belief, or {\em confidence}.  This complements the already-meaningful interpretation based on the original expression: $C_\alpha(y)$ is the collection of candidate values $\theta$ of $\Theta$ that are individually sufficiently plausible in the sense that $\pi_y(\theta) > \alpha$.  

From the consonant and strongly valid IM for $\Theta$ constructed above, it is straightforward to construct another consonant and strongly valid IM for any relevant feature $\Phi = \phi(\Theta)$ of $\Theta$.  This is just a direct application of what \citet{zadeh1975a, zadeh1978} referred to as the {\em extension principle}; see, also, \citet[][Sec.~3.2.3]{hose2022thesis}.  In particular, the IM for $\Psi$ has a contour function given by 
\[ \pi_y^\phi(\varphi) := \sup_{\theta: \phi(\theta)=\varphi} \pi_y(\theta), \quad \varphi \in \phi(\TT). \]
Then the corresponding IM's upper probability for $\Phi$ is determined by optimization:
\[ \uPi_y^\phi(B) = \sup_{\varphi \in B} \pi_y^\phi(\varphi), \quad B \subseteq \phi(\TT). \]
Note that this is exactly the same result obtained by directly considering the assertion $A=\{\phi(\Theta) \in B\}$ based on the original IM for $\Theta$.  That strong validity of the IM for $\Theta$ implies the same of the new IM for $\Phi$ is easy to see.  If $\phi$ is a one-to-one mapping, then clearly $\pi_y^\phi(\phi(\theta)) = \pi_y(\theta)$, in which case strongly validity is obvious.  More generally, because of the supremum in the definition, $\pi_y^\phi(\phi(\theta)) \geq \pi_y(\theta)$, so 
\[ \uprob_{Y,\Theta}\{ \pi_Y^\phi(\phi(\Theta)) \leq \alpha\} \leq \uprob_{Y,\Theta}\{ \pi_Y(\Theta) \leq \alpha \} \leq \alpha. \]

Finally, the property in Part~1 of Corollary~\ref{cor:valid} is actually stronger than indicated there, but the extra strength is subtle and deserves to be addressed separately.  An equivalent way to define strong validity (see Definition~3 in Part~I) of $y \mapsto \uPi_y$ is 
\begin{equation}
\label{eq:uniform}
\uprob_{Y,\Theta}\{\uPi_Y(A) \leq \alpha \text{ for some $A$ that contains $\Theta$} \} \leq \alpha, \quad \alpha \in [0,1]. 
\end{equation}
There's a similar expression in terms of the lower probability, but I'll not need it here.  The difference is that, in \eqref{eq:size}, the set $A$ is fixed in advanced, whereas the ``for some $A$...'' in \eqref{eq:uniform} suggests a sort of uniformity in $A$ is being achieved. In particular, the bound in \eqref{eq:uniform} implies that the IM testing procedure controls the false-rejection rate even when the investigator peeks at the data first before defining the hypothesis to be tested.  More generally, these considerations would be relevant in the context of selective inference, i.e., where the data are used to determine which features of $\Theta$ are to be tested.  I'll save these details to be investigated elsewhere.

\subsubsection{Behavioral}
\label{SSS:behavioral}

As discussed briefly above, my primary goal with these investigations into the IM's behavioral properties is simply to demonstrate that the statistical reliability prioritized above doesn't put the data analyst at risk of severe irrationality.  In other words, I aim to justify the statement that, if an uncertainty quantification procedure is reliable, then it's use can't be grossly irrational.  The result below is similar to the corresponding result in Part~I, but with a key difference: here I'm using a particular IM construction that depends explicitly on the posited $\uprob$. So the result below is more insightful than that in Part~I because the connection between validity and coherence comes directly from the specification of $\uprob$ and the outer approximation-based IM construction.  

Based on the specification $\uprob_{Y,\Theta}$ for the joint distribution of the pair $(Y,\Theta)$, a marginal specification of the prior beliefs about $\Theta$ can be derived.  For notational convenience, let me write this marginal specification as 
\[ \uprob_\Theta(A) = \uprob_{Y,\Theta}(\YY \times A), \quad A \subseteq \TT. \]
There's a corresponding $\lprob_\Theta$ too, defined in the obvious way.  Then the IM construction, which relies crucially on the input $(\lprob_{Y,\Theta},\uprob_{Y,\Theta})$ and, indirectly, the marginal specification $(\lprob_\Theta,\uprob_\Theta)$, could be interpreted as an {\em updating rule}.  That is, the IM construction described above offers a rule by which the quantification of prior uncertainty, $(\lprob_\Theta,\uprob_\Theta)$, derived from $(\lprob_{Y,\Theta},\uprob_{Y,\Theta})$, is updated to a new quantification of uncertainty, $(\lPi_y,\uPi_y)$, that incorporates the observed data $Y=y$ in a particular way.  From this perspective, it's important to ask if the updating rule puts the user at risk of certain irrational judgments.  These questions go back at least to \citet{definetti1937}.  The rationality notion that I'll be concerned with here is what I'm calling {\em half-coherence}; see Remark~\ref{re:coherent}.  I'll say that the updating rule $\{(\lprob_\Theta, \uprob_\Theta), y\} \mapsto (\lPi_y, \uPi_y)$ is half-coherent if 
\begin{equation}
\label{eq:coherent}
\inf_{y \in \YY} \lPi_y(A) \leq \lprob_\Theta(A) \quad \text{and} \quad \sup_{y \in \YY} \uPi_y(A) \geq \uprob_\Theta(A) \quad \text{for all $A \subseteq \TT$}. 
\end{equation}
The intuition behind this is as follows.  Suppose, for example, that the second inequality fails to hold for some $A$, so that $\uPi_y(A)$ is strictly less than $\uprob_\Theta(A)$ for all observations $y$.  In this case, you know that the price at which I'm willing to sell gambles on the event ``$\Theta \in A$'' will go down as soon as $Y=y$ is observed, no matter what $y$ is, so you'll surely be in a better position if you wait until $y$ is revealed to make your purchase.  This doesn't make me a sure loser in the usual sense---I could still win, but you're strictly better off if you purchase the gamble after $Y$ is observed instead of before. It'd be silly for me to give you a risk-free strategy to improve your circumstances, hence condition \eqref{eq:coherent}. 

What I refer to above as half-coherence is effectively the same as what \citet[][Lemma~4.1]{gong.meng.update} prove is achieved by the generalized Bayes rule.  This is a stronger requirement than the no-sure-loss property \citep[e.g.,][Def.~3.3]{gong.meng.update} which, in my notation, says 
\[ \inf_{y \in \YY} \lPi_y(A) \leq \uprob_\Theta(A) \quad \text{and} \quad \sup_{y \in \YY} \uPi_y(A) \geq \lprob_\Theta(A) \quad \text{for all $A \subseteq \TT$}. \]
If, say, the latter of the two inequalities above fails, then you can first sell gambles on ``$\Theta \in A$'' to me for $\lprob_\Theta(A)$ and then buy them back from me at a lower price after $Y$ is observed, hence making me a sure loser.  

\begin{thm}
\label{thm:coherent}
The updating rule determined by the IM construction described in Section~\ref{SS:construction} above is half-coherent in the sense of \eqref{eq:coherent}. 
\end{thm}

\begin{proof}
This follows immediately and directly from the IM construction.  I'll do the verification for the upper probability.  For any $A \subseteq \TT$, 
\begin{align*}
\sup_{y \in \YY} \uPi_y(A) & = \sup_{y \in \YY} \sup_{\theta \in A} \pi_y(\theta) \\
& = \sup_{y \in \YY} \sup_{\theta \in A} \omega_\eta(y,\theta) \\
& = \sup_{(y,\theta) \in \YY \times A} \omega_\eta(y,\theta) \\
& = \uOmega_\eta(\YY \times A) \\
& \geq \uprob_{Y,\Theta}(\YY \times A) \\
& = \uprob_\Theta(A),
\end{align*}
where the inequality is due to the fact that $\uOmega_\eta$ is an outer consonant approximation of $\uprob_{Y,\Theta}$.  Then the right-most inequality in \eqref{eq:coherent} holds, hence half-coherence.  
\end{proof}

Half-coherence itself doesn't make the IM's uncertainty quantification ``good'' or meaningful in any practical sense, but failing to satisfy \eqref{eq:coherent} would be a clear sign that something ``bad'' is happening.  It's the combination of statistical and behavioral properties that gives the IM solution its dual reliability--rationality appeal.  

A result similar to that in Theorem~\ref{thm:coherent} was proved in Part~I, but there's a key difference.  In Part~I, the theorem states that validity (Definition~\ref{def:validity}) implies half-coherence, but the reason why it's coherent is somewhat mysterious.  Of course, the IM above is itself valid, so the conclusion already follows from Theorem~2 in Part~I.  The difference, however, is that here I'm considering a particular construction of a strongly valid IM that makes explicit use of the underlying model $\uprob_{Y,\Theta}$.  So, now it's clear {\em why} half-coherence holds: the basic properties of $\uprob_{Y,\Theta}$ and the outer approximation don't allow otherwise. 

\begin{remark}
\label{re:coherent}
To be clear, half-coherence in \eqref{eq:coherent} is {\em not} the same as the notion of coherence between conditional and unconditional upper probabilities/previsions in, say, \citet[][Sec.~6.5.2]{walley1991}. As the name suggests, half-coherence is just one of the two-part necessary and sufficient conditions for coherence, namely, Walley's (C8).  Informally, Walley's second condition, (C9), states that 
\begin{equation}
\label{eq:coherent.2}
\lprob_\Theta(A) \leq \max\Bigl\{\lPi_y(A), \sup_{x \neq y} \uPi_x(A) \Bigr\}, \quad \text{for all $y \in \YY$ and all $A \subseteq \TT$}. 
\end{equation}
The intuition behind this condition is described in Part~I, so I won't repeat it here.  The key point is that \eqref{eq:coherent.2} will hold for the IM described above, but only in a trivial way, in cases where $y \mapsto \uPi_y(A)$ is continuous.  Otherwise, there's nothing special about the particular IM construction that would ensure \eqref{eq:coherent.2} holds, and there are cases where it will fail to hold \citep[][Sec.~6.5.4]{walley1991}. 
Only being able to verify half-coherence and the other half \eqref{eq:coherent.2} when it holds trivially is a bit disappointing, but it's a concession that I'm willing to make for the efficiency gains it affords; see Figure~\ref{fig:walley.binom} and its caption. 
\end{remark}

\begin{figure}[t]
\begin{center}
\scalebox{0.65}{\includegraphics{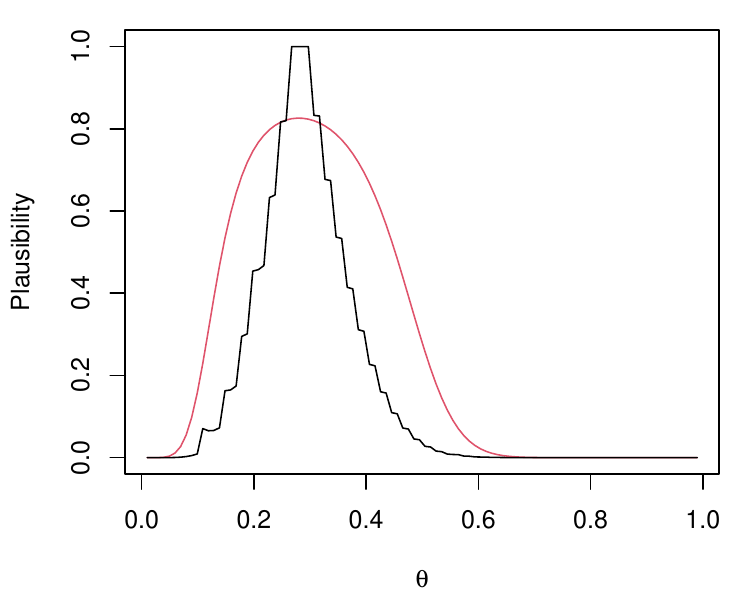}}
\end{center}
\caption{Suppose $(Y \mid \Theta=\theta) \sim \bin(n,\theta)$, vacuous prior on $\Theta$. Plot shows two ``plausibility contours'' based on $n=25$ and $y=7$: red is for the IM construction in \citet{walley2002} and black is for the IM constructed in Section~\ref{SS:vacuous}. The former has certain near-validity guarantees that I'll not describe here and is fully coherent, but it's not consonant so this plot is not a ``plausibility contour'' per se.  Tthe latter is strongly valid in the sense of Definition~\ref{def:strong.validity} but only half-coherent. My strongly valid IM is more efficient than Walley's, i.e., it has a more tightly concentrated contour, which implies narrower confidence intervals, so there's a practical price one pays for full coherence.}
\label{fig:walley.binom}
\end{figure}

\subsection{Marginalization}
\label{SS:marginal}

In the present setup, $\Theta$ is the unknown parameter in a posited statistical model.  It's often the case that inference on $\Theta$ itself is the objective, but it's also common for interest to be in some feature $\Phi = f(\Theta)$ for a known function $f$.  While {\em estimating} the feature $\Phi$ is easy, the task of {\em inference} on $\Phi$ is surprisingly challenging, for classical statisticians.  An advantage of ``probabilistic inference,'' e.g., Bayesian and fiducial, is that marginalization is at least conceptually straightforward---just apply the familiar probability calculus to get the corresponding marginal posterior distribution for $\Phi$, given $Y=y$.  Despite the conceptual simplicity, an underappreciated point is that whatever statistical properties are satisfied by the posterior for $\Theta$ need not carry over to the marginal posterior for $\Phi$.  A particularly striking example of this was first presented in \citet{stein1959}.  Part~III of the series focuses on the question of efficient marginal inference, so I'll not devote much space to this issue here.  The point is that the possibilistic marginalization that's suited for consonant IMs is both conceptually straightforward like the aforementioned probabilistic marginalization and preserves the IM's validity property.  

Let $\theta \mapsto \pi_y(\theta)$ be the IM's plausibility contour.  A hypothesis $B$ about a feature $\Phi = f(\Theta)$ corresponds to a hypothesis about $\Theta$ itself, i.e., 
\[ \Phi \in B \iff \Theta \in A_{f,B} := \{\theta \in \TT: f(\theta) \in B\}. \]
So, if $\uPi_y^f$ denotes the marginal/derived IM's consonant upper probability for $\Phi = f(\Theta)$, then it was explained in \eqref{eq:consonant.upper} how to evaluate this function:
\[ \uPi_y^f(B) = \uPi_y(A_{f,B}) = \sup_{\theta: f(\theta) \in B} \pi_y(\theta). \]
Taking $B=\{\phi\}$ to be a singleton gives the contour function for the marginal IM:
\begin{equation}
\label{eq:marginal.contour}
\pi_y^f(\phi) = \uPi_y^f(\{\phi\}) = \sup_{\theta: f(\theta)=\phi} \pi_y(\theta), \quad \phi \in f(\TT). 
\end{equation}
Then it's easy to check that 
\[ \sup_\phi \pi_y^f(\phi) = 1 \quad \text{and} \quad \uPi_y^f(B) = \sup_{\phi \in B} \pi_y^f(\phi), \quad B \subseteq f(\TT), \]
hence, the marginal IM is consonant too.  Interesting parallels can be drawn between the above style of marginalization and the more familiar probabilistic marginalization.  Indeed, both are based on applying some operation to a scalar function: integration of a density function in the probabilistic case and optimization of a contour function in the present imprecise-probabilistic case.  

Being able to carry out marginalization is only half the battle; the other half is demonstrating that the desired validity property enjoyed by the original IM for $\Theta$ is preserved in the marginal IM for $\Phi$.  This is crucial because it's exactly in the context of marginalization where issues with probabilistic inference become clear.  

\begin{cor}
Suppose that the IM for $\Theta$ is strongly valid relative to the posited model in the sense of Theorem~\ref{thm:valid}, and let $\Phi = f(\Theta)$ be a relevant feature.  Then the marginal IM defined above, with contour \eqref{eq:marginal.contour}, is also strongly valid, i.e., 
\[ \uprob_{Y,\Theta} \bigl\{ \pi_Y^f\bigl( f(\Theta) \bigr) \leq \alpha \bigr\} \leq \alpha, \quad \alpha \in [0,1]. \]
\end{cor}

\begin{proof}
Note that $\pi_y^f\bigl( f(\theta) \bigr) \geq \pi_y(\theta)$ for all $\theta$ and apply Theorem~\ref{thm:valid}. 
\end{proof}

This general recipe above leads to a valid marginal IM.  But if this strategy is able to preserve validity under no conditions, then one would guess that it's conservative.  That is, this basic marginalization generally doesn't lead to efficient---where ``efficient'' is in the sense of Section~\ref{SS:efficiency} below---marginal inference.  This loss of efficiency can even be seen from the above proof. That is, the supremum in \eqref{eq:marginal.contour} implies the marginal IM's contour can never be smaller, and typically will be larger, than that of the original IM.  And this is despite being focused on a general lower-complexity (e.g., lower-dimensional) unknown.  Therefore, efficient marginalization requires some extra care and, since this is relevant to almost all practical statistical problems, I'll devote my attention in Part~III of the series to marginalization.

\subsection{Efficiency considerations}
\label{SS:efficiency}

The results above are completely general and, therefore, very enlightening.  That is, the proposed framework provides a straightforward recipe for carrying out (imprecise) probabilistic inference under arguably the most general of statistical model assumptions while exactly controlling statistical error rates---no asymptotic approximations---and meeting the basic coherence requirements.  But a general strategy like this, one that achieves certain properties in every case, will rarely be most efficient.  Fortunately, in the typical cases one often encounters in practice, there are opportunities to improve efficiency through suitable {\em dimension reduction} steps, not unlike those recommended in \citet{imcond, immarg, imbook} and elsewhere.  I'll talk more about the dimension reduction details below and in the subsequent sections.  

Before moving forward, I need to say a bit more about the {\em efficiency} objective.  What I mean here by ``efficiency'' is intuitively clear but difficult to describe mathematically.  Situations in which efficiency is relatively easy to formulate are those in textbooks where the aim is to perform a specific task, such as point estimation, and optimality is expressed as solving a (constrained) optimization problem, e.g., minimum variance unbiased estimation.  In the present case, however, the aim of ``(imprecise-)probabilistic inference'' is intended to be sufficiently broad to cover most, if not all aspects of uncertainty quantification.  This bigger-picture focus has value, but there's no free lunch: spreading focus around to various aspects of uncertainty quantification necessarily makes formulation of a mathematically precise description of IM ``optimality'' or ``efficiency'' more challenging.\footnote{One idea would be to construct IMs that are optimal for various different statistical tasks \citep[see, e.g.,][Ch.~4]{imbook} and then fuse/patch them together in some suitable way, but this can't work for at least two reasons.  First, patching together separately-optimal IMs could lead to incoherence; see \citet{schervish1996} for an example where he tries to patch together p-values derived from two separate most-powerful tests.  Second, the patched-together IM can't be optimal in the sense that its separately optimal IMs components are.}  Fortunately, while a fully satisfactory mathematical formulation of IM efficiency is currently lacking, a lot can be done with a clear intuition.  

The focus is to keep $\theta \mapsto \pi_y(\theta)$ as small as possible in some (high level) sense, subject to the strong validity constraint.  It's necessary to say ``in some sense'' because, of course, being able to attain a minimum value at a given $y$ for all $\theta$ simultaneously can't be expected.  This focus makes sense from multiple points of view.
\begin{itemize}
\item First, from a statistical perspective, the magnitude of $\pi_y$ is directly related to the size of the IM's confidence regions and its p-values for tests for a given $y$, in particular, smaller $\pi_y$ suggests smaller confidence regions and smaller p-values.  Therefore, this magnitude is indirectly related to the expected volume of these confidence regions and the power of the tests.  So, aiming to make $\pi_y$ as small as possible is like trying to maximize efficiency of the IM's derived statistical procedures.
\vspace{-2mm}
\item Second, from an imprecise-probabilistic perspective, the magnitude of $\theta \mapsto \pi_y(\theta)$ is directly related to the size of the credal set $\cred(\uPi_y)$. That is, large $\pi_y$ makes $\cred(\uPi_y)$ large and vice versa, so a focus on minimizing the magnitude of $\pi_y$ is equivalent to an aim to limit the imprecision in the (imprecise-)probabilistic inference.  This aligns with the Principle of Expressiveness in Section~\ref{S:blocks} above. 
\end{itemize}

Perhaps not surprisingly, the steps needed to boost efficiency are largely problem- and model-specific.  But there is a general principle that can be applied across all cases to help improve efficiency.  The basic idea is to minimize the complexity of the Choquet integral calculation in \eqref{eq:omega}, the one that defines the contour $\omega_\eta(y,\theta)$ and, in turn, $\pi_y(\theta)$.  This complexity is controlled by, among other things, the dimension of its range of integration, so the principle below can be interpreted as an aim to mitigate the curse of dimensionality. By reducing this complexity/dimensionality, computation becomes more manageable without risking loss of statistical efficiency.  Moreover, at least in some cases (see Section~\ref{SS:vacuous}), reducing the complexity leads to a quantifiable improvement in statistical efficiency.

\begin{princom}
Reduce the dimension of the relevant variables as much as possible {\em before} defining the IM via the contour \eqref{eq:pi.y}. 
\end{princom}

Of particular interest are cases that involve a precise statistical model for $Y$, given $\Theta=\theta$, and partial prior information about $\Theta$.  There, classical notions like minimal sufficiency can be used to suitably reduce dimension and, hence, the complexity, with no associated loss of statistical efficiency.  Further dimension-reduction techniques, beyond that suggested by minimal sufficiency, will also be considered below (and in Part~III of the series). 
In certain ``extreme'' cases, apparently drastic complexity reduction steps are needed to achieve efficiency, which I'll discuss in Sections~\ref{SS:vacuous}--\ref{SS:complete}. 

There are some obvious connections between the Principle of Minimal Complexity and more familiar principles in the statistics literature.  In particular, both the {\em Sufficiency} and {\em Conditionality Principles} \citep[e.g.,][]{birnbaum1962} are rooted in the idea of reducing dimension/complexity where possible.  This connection will be seen more clearly in the sections that follow.  There's also an admittedly looser connection to the (controversial) {\em Likelihood Principle} but I'll save this explanation for Section~\ref{SS:complete}; see, also, Section~\ref{SS:vacuous}.  It's no coincidence that the Principle of Minimal Complexity is most closely related to the {\em Efficiency Principle} in \citet{imbook, sts.discuss.2014}.  In fact, the former is arguably more basic/fundamental because the latter's call for efficiency were always met through one or another complexity-reduction step \citep[e.g.,][]{imcond, immarg}.

\section{Precise statistical model}
\label{S:precise}

\subsection{General partial prior}
\label{SS:partial}

Consider the case where there exists a precise statistical model of the form $\prob_{Y|\theta}$ for the distribution of $Y$, depending on the value $\theta$ of $\Theta$.  Then a joint distribution for $(Y,\Theta)$ obtains by effectively multiplying the statistical model---treated as a conditional distribution in this case---by the prior distribution for $\Theta$.  Of course, uncertainty in the prior specification is encoded by an imprecise prior distribution with upper probability $\uprior = \uprob_\Theta$ on $\TT$.  Here I'll consider the case where the credal set $\credal = \cred(\uprior)$ is neither a singleton (complete prior information) nor everything (vacuous prior information); these two extreme cases will be treated separately in the next two subsections.  

What's special about the precise statistical model case is two-fold: first, it helps to shed light on what the construction does and why it works; second, it readily provides opportunities for dimension reduction and, in turn, potential efficiency gains.  Following the general developments leading to \eqref{eq:eta.y}, consider the baseline 
\begin{equation}
\label{eq:baseline}
\eta(y,\theta) = \frac{p_\theta(y) \, \uprior(\{\theta\})}{\sup_{t \in \TT} \{ p_t(y) \, \uprior(\{t\})\}}, \quad (y,\theta) \in \YY \times \TT, 
\end{equation}
where $p_\theta(y)$ is the likelihood function from the precise statistical model $\prob_{Y|\theta}$.  Then I'll proceed to construct a strongly valid IM for inference on $\Theta$ using the general formula \eqref{eq:pi.y}, which determines $y$-dependent possibility contour function 
\[ \pi_y(\theta) = \uprob_{Y,\Theta}\{ \eta(Y,\Theta) \leq \eta(y,\theta)\}, \quad \theta \in \TT, \]
with $\eta$ as in \eqref{eq:baseline}.  Since the nuance of the particular statistical application is mostly captured by the plausibility order \eqref{eq:baseline}, a few remarks about this formula are in order.  First, note that expression \eqref{eq:baseline} applies whether the $(Y \mid \Theta)$ model is discrete or continuous (see Remark~\ref{re:construction}.1).  Second, this expression looks superficially similar to Bayes's formula, i.e., proportional to likelihood times prior, although the interpretation of the ``prior'' portion and the normalization is different here in the imprecise case.  Third, if I set $r_y(\theta) = p_\theta(y) / \sup_{t \in \TT} p_t(y)$ to be the relative likelihood, bounded between 0 and 1, then \eqref{eq:baseline} can also be interpreted as the combination of two possibility measures according to a common fuzzy set intersection rule \citep[e.g.,][]{dubois.prade.1988}.  Finally, \eqref{eq:baseline} makes clear the {\em regularization} effect of incorporating $\uprior$: it clearly down-weights $\eta(y,\theta)$ in low-prior-support regions, which creates an opportunity for efficiency gain.  

\begin{remark}
\label{re:normalization}
One more technical remark is in order.  Recall the discussion of how to normalize the $\eta(y,\theta)$ function in Section~\ref{SS:construction}.  In particular, in the precise-model context, my recommendation was a supremum-based normalization, i.e., 
\[ c(y)^{-1} = \sup_{t \in \TT} \{ p_t(y) \, q(t) \}. \]
An alternative, Choquet-integral-based normalization would take 
\begin{align*}
c(y)^{-1} = \uprior L_y, 
\end{align*}
where $L_y(\theta) = p_\theta(y)$ is the likelihood function.  The appeal of this choice is its compatibility with the most familiar cases:
\begin{itemize}
\item with a vacuous prior, $\uprior L_y = \sup_t L_y(t)$ and then the corresponding $\eta(y,\theta)$ in \eqref{eq:eta.y} is the familiar relative likelihood;
\vspace{-2mm}
\item with a complete prior, $\uprior L_y = \int L_y(t) \, \prior(dt)$ and then the corresponding $\eta(y,\theta)$ in \eqref{eq:eta.y} is the Bayesian posterior density.
\end{itemize}
The problem, however, is that this choice of normalization doesn't ensure that the resulting $\pi_y(\theta)$ in \eqref{eq:pi.y} is a possibility contour in $\theta$ for each fixed $y$.  I've considered this choice carefully and the only two intuitively reasonable options seem to be the supremum and Choquet-integral-based normalization strategies.  Since only the former ensures the mathematical structure needed, the choice between the two is clear.
\end{remark}

Related to the proposed construction above, there are two important practical questions that arise immediately.  I'll state and address these both next.  

\begin{question}
{\em How can/should the partial prior information be quantified?}  There are lots of options to consider, including some classical ideas, such as contamination neighborhoods \citep[e.g.,][]{wasserman1990, huber1981, berger1984}.  I think that possibility measure priors can also come in handy.  For example, it's not out of the question that certain moments can be elicited from subject matter experts, in which case those possibility measure priors based on Markov's or Chebyshev's inequalities \citep[e.g.,][Sec.~4]{dubois.etal.2004} would be reasonable choices.  A discussion of how a prior possibility measure incorporating structural information about sparsity could be constructed was given in Part~I, and I'll revisit this specifically in a follow-up paper in the series.  To be clear, however, I can't recommend a ``default'' prior---it's the data analyst's responsibility to make this judgment.  The point is that the IM's properties are relative to the prior specification, and stronger properties can be achieved with stronger assumptions.  I can't tell the practitioner how strong of assumptions they can/should make, the user has to take some responsibility/ownership in their data analysis.\footnote{A contributing factor to the replication crisis in science is that statistics education has given the scientific community the impression that we've taken care of the data analysis, that their responsibility is just to pick an appropriate test to use.  I'm optimistically hoping that by {\em forcing} the scientists to take an active role in scientific inference---by having to make decisions about what, if anything, they specifically know and deserves to be incorporated in the analysis---then that will have an overall positive effect.} But note that {\em I'm not requiring} the user to have any prior information at all. He's free to proceed with a vacuous prior, with $\uprior(\{\theta\}) \equiv 1$, if indeed there's no prior information available or he doesn't know how to quantify it, so my proposed framework puts no constraints on the data analyst---in fact, it's far more flexible than the existing frequentist and Bayesian schools.  
\end{question}

\begin{question}
{\em How can the IM's output be computed?}  In many cases, the partial prior $\uprior$ for $\Theta$ can be expressed in terms of a possibility measure with a prior contour $q$.  Then the IM's contour function, which is given by a Choquet integral, takes a simpler looking form according to Proposition~7.14 in \citet{lower.previsions.book}:
\[ \pi_y(\theta) = \int_0^1 \sup_{\vartheta: q(\vartheta) > \alpha} \prob_{Y|\vartheta}\{ \eta(Y,\vartheta) \leq \eta(y,\theta)\} \, d\alpha. \]
From this expression, there's a very natural Monte Carlo-driven strategy available: 
\begin{equation}
\label{eq:partial.MC}
\pi_y(\theta) \approx \int_0^1 \max_{\vartheta_r: q(\vartheta_r) > \alpha} \Bigl[ \frac1M \sum_{m=1}^M 1\{\eta(Y_r^{(m)}, \vartheta_s) \leq \eta(y,\theta)\} \Bigr] \, d\alpha,
\end{equation}
where $\vartheta_r,\ldots,\vartheta_R$ is a fixed grid that spans all of (the relevant parts of) $\TT$ and $Y_r^{(1)},\ldots,Y_r^{(M)}$ are iid samples from $\prob_{Y|\vartheta_r}$, for $r=1,\ldots,R$.  Other more efficient strategies may be possible in certain cases.  See Section~\ref{SS:examples} for a few illustrations.  For further details on computation, see \citet{hose.hanss.martin.belief2022} and Section~\ref{S:discuss}.  
\end{question}

Often there's a non-trivial minimal sufficient statistic, $S(Y)$, for the model $\prob_{Y|\theta}$, of dimension no larger than that of $Y$ and no smaller than that of $\theta$.  With a slight abuse of notation, it follows from the factorization theorem that 
\[ p_\theta(y) = p_\theta(s) \, p(y \mid s), \quad s=S(y). \]
In the above display, $p_\theta(s)$ is the likelihood function for $\theta$ based on the marginal distribution of $S(Y)$, which depends on $\theta$, and $p(y \mid s)$ is the conditional density/mass function for $Y$, given $S(Y)=s$, which, by definition of sufficiency, doesn't depend on $\theta$.  Then the baseline \eqref{eq:baseline} immediately simplifies to 
\[ \eta(y,\theta) = \frac{p_\theta(s) \, \uprior(\{\theta\})}{\sup_{t \in \TT} \{ p_t(s) \, \uprior(\{t\})\}}, \]
which only depends on $y$ through the value $s=S(y)$ of the minimal sufficient statistic.  With another slight abuse of notation, if I write this as $\eta(s,\theta)$, then it readily follows that the IM's possibility contour can be expressed as 
\begin{align*}
\pi_s(\theta) & = \uprob_{Y,\Theta}\{ \eta(S(Y),\Theta) \leq \eta(S(y),\theta)\} \\
& = \uprob_{S,\Theta}\{ \eta(S,\Theta) \leq \eta(s,\theta)\} \\
& = \sup_{\prior \in \credal} \int \prob_{S|\vartheta}\{\eta(S,\vartheta) \leq \eta(s,\theta)\} \, \prior(d\vartheta), \quad s=S(y), 
\end{align*}
where $\prob_{S|\vartheta}$ is the image of $\prob_{Y|\vartheta}$ under mapping $S$ and, of course, the left-hand side of the above display depends on $y$ only through $s=S(y)$.  The point is that the complexity of the right-hand side has been reduced---the Choquet integral is with respect to the ``marginal'' upper probability for the lower-dimensional $(S,\Theta)$ induced by that for $(Y,\Theta)$ and the mapping $S$ that defines a minimal sufficient statistic.  The corresponding strong validity property now takes the form 
\[ \uprob_{S,\Theta}\{ \pi_S(\Theta) \leq \alpha \} \leq \alpha, \quad \alpha \in [0,1], \]
and the relevant consequences of strong validity presented in Section~\ref{S:general} still hold, just now they're for the lower-dimensional IM indexed by a minimal sufficient statistic.  To be clear, the above reduction in simplicity via sufficiency didn't actually change the IM in the sense that the baseline ``$\pi_y(\theta)$'' equals the reduced ``$\pi_{S(y)}(\theta)$'' for all $(y,\theta)$.  But the spirit of this dimension reduction process can be used in other cases (see below) and there it will generally affect the IM contour, thus the potential for efficiency gain. 

What's described above is entirely expected, so the it's easy to overlook what's driving the dimension reduction.  The fact that $\eta(y,\theta)$ only depends on $y$ through $s=S(y)$ is an immediate consequence of the factorization theorem, so I don't need to say any more about that.  The key observation that I want to draw the reader's attention to is the fact that the same factorization being used to simplify $\eta$ is also being applied to the $\uprob$-probability calculation that defines the IM output.  When the direct dependence of $\eta(y,\theta)$ on ``$y \mid s$'' drops out, that dimension in the $\uprob$-probability computation collapses, thereby reducing the dimension/complexity.  In the end, the IM construction is based on an outer consonant approximation to the joint distribution of $(S,\Theta)$, rather than of $(Y,\Theta)$, and now it's clear why all the same validity properties hold.  

In other cases, there might be opportunities for dimension reduction beyond what minimal sufficiency itself allows.  For example, in Cauchy models with an unknown location parameter, the order statistics are minimal sufficient and yet the parameter is a scalar.  In some of those cases, it's possible to express the minimal sufficient statistic as the pair $\{S(Y),U(Y)\}$ where $U(Y)$ is ancillary in the sense that the  distribution of $U=U(Y)$, given $\Theta=\theta$, doesn't depend on $\theta$.  With my now-familiar abuse of notation, this leads to the factorization 
\[ p_\theta(y) = p_\theta(s \mid u) \, p(u) \, p(y \mid s, u), \quad s=S(y), \quad u=U(y). \]
Only the first term depends on $\theta$, everything else cancels in the ratio \eqref{eq:baseline}.  That is, 
\[ \eta(y,\theta) = \frac{p_\theta(s \mid u) \, \uprior(\{\theta\})}{\sup_{t \in \TT} \{ p_t(s \mid u) \, \uprior(\{t\})\}}, \]
and the left-hand side only depends on $y$ through $(s,u)$.  Direct (but {\em naive}) application of the IM construction described above leads to 
\begin{align*}
\pi_y^\text{naive}(\theta) & = \uprob_{Y,\Theta}\{ \eta(Y, \Theta) \leq \eta(y, \theta)\} \\
& = \sup_{\prior \in \credal} \int \prob_{S,U|\vartheta}\{ \eta(S, U, \vartheta) \leq \eta(s,u,\theta)\} \, \prior(d\vartheta) \\
& = \sup_{\prior \in \credal} \int \E\bigl[ \prob_{S|U,\vartheta}\{ \eta(S,U,\vartheta) \leq \eta(s,u,\theta) \} \bigr] \, \prior(d\vartheta) , \quad (s,u)=(S,U)(y), 
\end{align*}
where the expectation inside the integral is with respect to the marginal distribution of $U$, which doesn't depend on $\vartheta$; also note that the left-hand side depends on $y$ only through $(s,u)$.  This IM construction is fine in the sense that strong validity holds.  I indicated above, however, that this construction is ``naive''---that's because I missed an opportunity to apply the Principle of Minimal Complexity.  In particular, since $\eta$ depends only on the {\em conditional} density of $S$, given $(u,\theta)$, I had an opportunity to take $u$ as fixed and thereby reduce the dimension/complexity of the IM construction.  This lower-complexity construction goes as follows:
\begin{align*}
\pi_{s|u}(\theta) & = \uprob_{S,\Theta|U=u}\{\eta(S,u,\Theta) \leq \eta(s,u,\theta)\} \\
& = \sup_{\prior \in \credal} \int \prob_{S|u,\vartheta}\{\eta(S, u, \vartheta) \leq \eta(s, u, \theta)\} \, \prior(d\vartheta), \quad (s,u)=(S,U)(y),
\end{align*}
where ``$\pi_{s|u}(y)$'' is just the notation I'm using for the expression in the right-hand side.  The key difference between this and the IM construction above is that the integration over the $U$-space has been eliminated, thus lowering the complexity.  

The same strong validity property holds for this lower-complexity IM construction, so there's no risk.  Except for a few the measure-theoretic technicalities, the explanation of this is clear.  Indeed, for almost all values $u$ of $U$, define 
\[ f_\alpha(u) = \uprob_{S,\Theta|u}\{ \pi_{S|u}(\Theta) \leq \alpha \}, \]
and define $f_\alpha(u)$ arbitrarily on the complementary null set.  Then the proof of strong validity above applied to this case implies that 
\begin{equation}
\label{eq:valid.ae}
f_\alpha(u) \leq \alpha, \quad \text{for all $\alpha \in [0,1]$ and almost all $u$}. 
\end{equation}
From here, an application of Fubini's theorem followed by Fatou's lemma gives 
\begin{align*}
\uprob\{ \pi_{S|U}(\Theta) \leq \alpha \} & = \sup_{\prior \in \credal} \int \prob_{S,U|\theta}\{ \pi_{S|U}(\theta) \leq \alpha \} \, \prior(d\theta) \\
& = \sup_{\prior \in \credal} \int \E\bigl[ \prob_{S|U,\theta}\{ \pi_{S|U}(\theta) \leq \alpha \} \bigr] \, \prior(d\theta) \\
& = \sup_{\prior \in \credal} \E \Bigl[ \int \prob_{S|U,\theta}\{ \pi_{S|U}(\theta) \leq \alpha \} \, \prior(d\theta) \Bigr] \\
& \leq \E\{ f_\alpha(U) \}.
\end{align*}
The upper bound is no more than $\alpha$, thanks to \eqref{eq:valid.ae}, hence strong validity.  

Since this ``condition-on-an-ancillary'' solution is, again, entirely expected, it's worth emphasizing what's the driving force behind this dimension reduction.  After carrying out the factorization and canceling those factors that don't depend on $\theta$, what remains is a {\em conditional distribution} of $S$, given $U$.  Following the Principle of Minimal Complexity, I take that $U$ in the conditional distribution fixed and effectively construct consonant approximation to the conditional distribution of $S$, given $(U,\Theta)$.  Consequently, in this case, the dimension/complexity is reduced in two steps: first, the ``$y \mid (s,u)$'' dimension is collapsed and, second, a particular slice in the ``$u$'' dimension is focused in on.

Aside from facilitating the proof of strong validity, the result in \eqref{eq:valid.ae} is of independent interest.  This shows that the IM is giving strongly valid inference even when focused exclusively on the same $U(y)$-specific subset of the population that the observed data falls into.  This was Fisher's original goal of conditional inference: to make the inference even more relevant to the data at hand by focusing any necessary probability calculations on the corresponding relevant subset.  See Section~\ref{SS:examples} below for a couple examples. 

The two extreme cases---complete and vacuous prior information---in the following two subsections, respectively, involve similarly extreme dimension reduction considerations, not unlike what was done above in dealing with the factorization involving $(S,U)$.  Further factorization tricks like this for reducing dimension will be discussed below in contexts where nuisance parameters are present; see Part~III of the series. 

\subsection{Vacuous prior}
\label{SS:vacuous}

Section~\ref{SS:partial} considered a general case on the partial prior spectrum presented in Part~I.  Here I focus on one of the two extremes on that spectrum, namely, where the prior information for $\Theta$ is {\em vacuous} and the corresponding prior credal set $\credal$ consists of all possible prior distributions.  This is the model for prior ignorance implicitly adopted when it's assumed that $\Theta$ is simply ``unknown.''  My reason for focusing on this case isn't that I think it's realistic or practical.  In fact, it must be rare that an investigator would be genuinely and completely ignorant about the quantity of interest.  But this is an interesting case at least for technical and historical reasons. 

As a special case of that in \eqref{eq:baseline}, consider the baseline 
\begin{equation}
\label{eq:rel.lik}
\eta(y,\theta) = \frac{p_\theta(y)}{\sup_{t \in \TT} p_t(y)}, \quad (y,\theta) \in \YY \times \TT, 
\end{equation}
where, in this case, the prior contour appears to be absent because it's constant equal to 1, i.e., $\uprior(\{\theta\}) = 1$ for all $\theta$.  This is what prior ignorance means: all values of $\Theta$ are equally and fully possible.  The right-hand side above is familiar---it's the likelihood ratio statistic for testing $\Theta=\theta$ or, more succinctly, the {\em relative likelihood}.  To simplify the notation for what follows, I'll assume that all the dimension reduction steps described in Section~\ref{SS:partial} have been taken and I've renamed the lower-dimensional statistic $S$ as $Y$.  Direct (but, again, {\em naive}) application of the IM construction described above leads to 
\begin{align*}
\pi_y^\text{naive}(\theta) & = \uprob_{Y,\Theta}\{ \eta(Y,\Theta) \leq \eta(y,\theta)\} \\
& = \sup_{\vartheta \in \TT} \prob_{Y|\vartheta}\{ \eta(Y,\vartheta) \leq \eta(y,\theta)\}, \quad \theta \in \TT, 
\end{align*}
where the supremum results from taking a Choquet integral with respect to the vacuous prior.  The corresponding IM is strongly valid but, as the reader can anticipate, what makes this construction ``naive'' is that I failed to take an opportunity to reduce dimension and improve efficiency.  As before, what signals an opportunity to reduce dimension is that the baseline $\eta$ above only depends on the {\em conditional} density of $Y$, given $\theta$.  By taking $\theta$ as fixed in the above calculation I get a new and more efficient IM:
\[ \pi_y(\theta) = \prob_{Y|\theta}\{ \eta(Y,\theta) \leq \eta(y,\theta)\}, \quad \theta \in \TT. \]
The right-hand side in the above display is familiar: it's just the (exact) p-value for testing $\Theta=\theta$ based on the likelihood ratio statistic, so the corresponding IM is obviously strongly valid.  Computation in this case is easier than in the genuine partial prior case discussed above, thanks to the dimension reduction step.  Indeed, there's no optimization or integration as in \eqref{eq:partial.MC}, so a simple Monte Carlo approximation would be
\begin{equation}
\label{eq:vacuous.MC}
\pi_y(\theta) \approx \frac1M \sum_{m=1}^M 1\{\eta(Y^{(m)}, \theta) \leq \eta(y,\theta)\}, \quad \text{where} \; Y^{(1)},\ldots,Y^{(M)} \iid \prob_{Y|\theta}. 
\end{equation}
Of course, more sophisticated approaches can be taken on a problem-specific basis.  Numerical illustrations are shown in Section~\ref{SS:examples}. 

It's also obvious that the latter contour determines an everywhere more efficient IM than the one I called ``naive'' above:
\[ \pi_y(\theta) \leq \pi_y^\text{naive}(\theta), \quad \text{for all $\theta \in \TT$}. \]
The above inequality holds since $\pi_y^\text{naive}$ involves a supremum while $\pi_y$ doesn't.  In this case, clearly the naive IM that ignores the dimension-reduction opportunity is inadmissible.  But despite the inefficiency-creating supremum, the ``naive'' IM may not differ too much from the other IM.  The reason is that the likelihood ratio statistic is often an exact pivot, i.e., the distribution of $\eta(Y,\theta)$ doesn't depend on $\theta$. Also Wilks's theorem states that, under regularity conditions, the likelihood ratio statistic is an asymptotic pivot.  Therefore, if $Y=(Y_1,\ldots,Y_n)$ is an iid sequence indexed by $n$, and if certain regularity conditions are satisfied, then $\pi_y^\text{naive}(\theta) \approx \pi_y(\theta)$ for all $\theta$ when $n$ is sufficiently large.  This same pivot structure also aids with computation, since the Monte Carlo samples in \eqref{eq:vacuous.MC} wouldn't have to depend on the particular $\theta$ on the left-hand side.  

The IM construction above produces a possibility measure, $\uPi_y$.  Since possibility measures are coherent, there's a corresponding credal set $\cred(\uPi_y)$ that consists of all those ($y$-dependent) probability distributions on $\TT$ that are dominated by $\uPi_y$.  The well-known characterization of this credal set (cf.~Remark~\ref{re:inclusion}) states that 
\[ \Pi_y \in \cred(\uPi_y) \iff \Pi_y\{C_\alpha(y)\} \geq 1-\alpha \;\; \text{for all $\alpha \in [0,1]$}, \]
where $C_\alpha(y)$ is the upper-$\alpha$ level set of $\pi_y$, the IM's possibility contour.  In words, all the probability distributions compatible with the IM's output assign probability at least $1-\alpha$ to the IM's $100(1-\alpha)$\% plausibility regions.  This credal set contains some extreme probability distributions, e.g., point mass distributions on points in the core of $\uPi_y$, but also some less extreme ones.  Obviously, the highly-concentrated distributions in $\cred(\uPi_y)$ wouldn't be good approximations but there are some more diffuse elements in $\cred(\uPi_y)$ that might be decent approximations. What about some of the familiar probabilistic inference solutions, such as default-prior Bayes, fiducial, etc.? \label{page:fiducial.in.credal} Are they contained in the IM's credal set?  There are some relatively simple special cases (e.g., location models) where it's not too difficult to show, as in \citet[][Sec.~3.2]{dubois.etal.2004}, that the fiducial distribution, default-prior Bayes posterior distribution, etc.~are the most diffuse elements in the IM's credal set; see \citet{martin.isipta2023} for details.  But I'm not aware of a general definition of ``most diffuse,'' let alone results about when such an element exists beyond these special cases.  This is an interesting open question because, in my view, the appropriate {\em definition} of a confidence distribution is as the ``most diffuse'' element in the credal set of a strongly valid IM.  This would make clear what role the confidence distribution plays, namely, as a simple, familiarity-motivated probabilistic approximation to the strongly valid possibilistic IM. 

I'll end this section with a brief discussion of the {\em likelihood principle} \citep[e.g.,][]{basu1975, birnbaum1962, bergerwolpert1984}.  The likelihood principle states that inference on $\Theta$ ought to depend only on, say, the relative likelihood.  Although the IM constructed above is driven by the relative likelihood $\eta$ in \eqref{eq:rel.lik}, it doesn't satisfy the likelihood principle because it depends on the model $\prob_{Y|\theta}$, which determines but isn't determined by the relative likelihood.  Some argue that failure to satisfy the likelihood principle is a shortcoming, but I want to offer a different perspective here.  I claim that {\em IMs shouldn't satisfy the likelihood principle}, at least not by default. Strong validity is my top priority so I'm not willing to sacrifice on this.  My conjecture is that it's impossible to achieve strong validity and the likelihood principle simultaneously in this setting.  It is possible, however, to simultaneously achieve a relaxed version of strong validity and the likelihood principle; the proposal in \citet{walley2002} is one example, and the ``relaxed'' strong validity property is what he calls the {\em fundamental frequentist principle}.  But even if my conjecture is wrong, or if I was willing to relax the strong validity condition, then I still wouldn't want to satisfy the likelihood principle, and here's why.  If the IM could achieve both the likelihood principle and strong validity, then the latter must hold {\em uniformly} over all those models that determine the same relative likelihood function.  Consequently, the IM would be less efficient as a result of the constraint imposed by the likelihood principle.  In Walley's case, for example, the confidence intervals he obtains are significantly wider (see Figure~\ref{fig:walley.binom}) than the standard intervals, and that inefficiency can be traced back precisely to the Bayesian aspects of his solution that ensure the likelihood principle is met.  As Walley argues, this additional width is to be expected since the interval's coverage properties must hold, e.g., over all stopping rules.  But what if the data analyst {\em knows} what stopping rule was used?  Why should he accept a statistical solution that ignores this information and gives less efficient inference as a result? It should be up to the data analyst to decide whether he should be concerned about different models (e.g., binomial or negative binomial) with equivalent relative likelihoods: if he's concerned about this, then he can incorporate this imprecision into the analysis himself and have control over the resulting loss of efficiency; if he's not concerned, then he can enjoy the efficiency that's afforded to him by not having that concern.  Ultimately, it's the model assumptions, etc., that determine the data analyst's IM, and the framework should return the most efficient IM as possible based on his/her inputs. The framework shouldn't indirectly control the data analyst's model assumptions.  See \citet{martin.basu} for more on the IM's ability to balance the likelihood principle and the reliability necessary for scientific inference.

\subsection{Complete prior} 
\label{SS:complete}

Next, I consider the opposite end of the partial prior spectrum, namely, where the prior information is {\em complete} or precise.  This is the classical Bayesian setup where the credal set contains just one prior distribution: $\credal = \{\prior\}$.  Like prior ignorance, the complete prior case is unrealistic, since virtually no applications would have such strong prior information available.  But, again, this is interesting for technical and historical reasons. 

Following the general developments leading to \eqref{eq:eta.y}, consider the baseline setup 
\[ \eta(y,\theta) = \frac{p_\theta(y) \, q(\theta)}{\sup_{t \in \TT} p_t(y) \, q(t)}, \quad (y,\theta) \in \YY \times \TT, \]
where $p_\theta(y)$ is the likelihood portion and $q(\theta)$ is the prior density function.  The numerator can be re-expressed as 
\[ p_\theta(y) \, q(\theta) = q_y(\theta) \, p(y), \]
where $q_y(\theta)$ is the posterior density function based on Bayes's formula and $p(y)$ is the marginal density of $y$ derived from the joint distribution of $(Y,\Theta)$.  That marginal density doesn't depend on $\theta$, so it cancels in the ratio, leading to 
\[ \eta(y,\theta) = \frac{q_y(\theta)}{\sup_{t \in \TT} q_y(t)}, \quad (y,\theta) \in \YY \times \TT. \]
Since this depends only on the {\em conditional} distribution of $\Theta$, given $Y=y$, there's an opportunity for dimension reduction to be had here by taking $y$ fixed.  This is exactly like how the ancillary statistic $U=u$ was taken as fixed in the previous subsection.  Then the Principle of Minimal Complexity suggests an IM construction with contour 
\[ \pi_y(\theta) = \prior_y\{ \eta(y,\Theta) \leq \eta(y,\theta)\}, \quad \theta \in \TT, \]
where $y$ is fixed and $\prior_y$ is the usual Bayesian posterior distribution.  In fact, $\eta$ above can be replaced by the unnormalized posterior density, i.e., likelihood times prior, since normalizing constants cancel when $y$ is fixed.  Computationally, this IM construction is relatively straightforward.  That is, if samples $\Theta_1,\ldots,\Theta_M$ can be drawn from the posterior distribution $\prior_y$, then a Monte Carlo approximation of the IM contour is
\[ \pi_y(\theta) \approx \frac1M \sum_{m=1}^M 1\{\eta(y,\Theta_m) \leq \eta(y,\theta)\}, \quad \theta \in \TT, \]
and the IM's upper probability at a generic hypothesis $A$ can be readily approximated by maximizing the above Monte Carlo approximation over a grid of $\vartheta$ points in $A$.  

The IM just described is, of course, different from the Bayesian solution---it's an outer consonant approximation of the usual posterior distribution.  There are advantages to this possibilistic representation of the Bayesian solution and these benefits were, more or less, already recognized by Pereira \& Stern (and collaborators) in a series of articles about what they call the {\em e-value} (``e'' for evidence) and its use in Bayesian significance testing \citep[e.g.,][]{MR4426408, MR3188081, MR2383252}. Pereira \& Stern's claim is that their e-value has a p-value-like calibration with respect to both the posterior and the sampling distribution, thus making it suitable for significance testing in Bayesian and frequentist settings simultaneously.  This calibration is exactly the strong validity property that follows from the general IM construction described above, as I explain in more detail below.  Pereira \& Stern also made connections to possibility theory, but not to the extent that those connections are made in the present paper.  

While the strong validity property follows by the general construction given above, in this case it's easy to verify directly.  Indeed, $\pi_y(\theta)$ is just the posterior distribution function of the random variable $\eta(y, \Theta)$, evaluated at $\eta(y,\theta)$, so it follows from standard arguments that $\pi_y(\Theta)$, as a function of $\Theta \sim \prior_y$, is stochastically no larger than $\unif(0,1)$.  This implies a sort of {\em strong conditional validity}, i.e., 
\begin{equation}
\label{eq:conditional.strong.validity}
\prior_y\{ \pi_y(\Theta) \leq \alpha \} \leq \alpha, \quad \text{for all $\alpha \in [0,1]$ and almost all $y$}. 
\end{equation}
From this, the strong validity property as in \eqref{eq:strong.validity} following immediately from the iterated expectations formula.  Like in the previous subsection, however, the strong conditional validity result is much more than just a step towards proving \eqref{eq:strong.validity}.  The reader can easily verify that, in this case, the IM depends only on the posterior distribution and, therefore, only on the relative likelihood; and this, of course, implies that the IM satisfies the likelihood principle. Indeed, it's the result in \eqref{eq:conditional.strong.validity} that explains this compatibility of strong validity and the likelihood principle in the special case of complete prior information.

\subsection{General partial prior, revisited}
\label{SS:part.naive}

The reader may notice that the recommended applications of the Principle of Minimal Complexity suggested in the general partial prior solution described in Section~\ref{SS:partial} above had nothing to do with the partial prior.  These were just basic sufficiency and condition-on-ancillary-statistics recommendations that apply no matter what kind of prior information is available, if any.  One would expect that certain partial priors induce structure that can be leveraged to further reduce complexity in the Choquet integral.  Unfortunately, I don't currently know how to do much more than what was already suggested in Section~\ref{SS:partial}.  Here I want to briefly explain the little bit more that I didn't already share above and highlight the opportunities that I see for efficiency gain.  

After seeing the complexity-reduction strategy for the vacuous prior case as described in Section~\ref{SS:vacuous}, it's natural to consider naively applying the same strategy in the general partial prior case.  That is, if the available prior information is quantified by, say, a possibility measure $\uprior$ with contour function $q$ and $\eta(y,\theta)=\eta_q(y,\theta)$ is the baseline plausibility ordering is as in \eqref{eq:baseline}, which depends on $q$, then define the IM for $\Theta$ to have contour 
\[ \pi_y(\theta) = \pi_{y,q}(\theta) = \prob_{Y|\theta}\{ \eta_q(Y,\theta) \leq \eta_q(y,\theta)\}, \quad \theta \in \TT. \]
This is different from the proposal in Section~\ref{SS:partial} because here I'm not calculating a Choquet integral.  Instead, I just naively reduce complexity by fixing the value of $\Theta$ at the hypothesized value $\theta$.  Why would I do this?  The two basic reasons are (a)~it's simpler to compute and (b)~it does incorporate the prior information $q$.  More compelling is that (c)~this naive, complexity-reduced IM is still strongly valid, i.e., 
\[ \uprob_{Y,\Theta}\{ \pi_{y,q}(\Theta) \leq \alpha \} \leq \alpha, \quad \alpha \in [0,1], \]
where $\uprob_{Y,\Theta}$ is the upper joint distribution of $(Y,\Theta)$ under the partial prior with contour $q$.  Strong validity implies, in particular, that tests and confidence regions derived from the IM have provable error rate control guarantees.  

The downside, however, is that this IM construction apparently doesn't make full use of the information in $q$.  This makes sense, intuitively, given that it's based on skipping the key Choquet integral step, but it's difficult (for me) to see the statistical consequences of this skipped-step.  As I show numerically in Example~\ref{ex:binomial} below (and have checked in other examples not shown here), this naive, complexity-reduced IM doesn't commit to the partial prior information as much as the partial-prior IM solution proposed in Section~\ref{SS:partial}.  That is, there appears to be opportunities for the former to be more aggressive---and hence more efficient---when partial prior information is available.  So, despite the advantages of this naive, complexity-reduced partial-prior IM, in particular, the strong validity, I don't believe this solution is fully satisfactory.  More work is needed to understand how the structure in the partial prior can be leveraged to allow for a more tailored dimension reduction that leads to the greatest efficiency gains; see Section~\ref{S:discuss}.

\section{Examples}
\label{SS:examples}

\subsection{Basic illustrations}

The first three examples are very simple, aimed to illustrate the proposed solution and, in particular, to show more concretely what I have in mind when I speak of a ``partial prior.''  Some more sophisticated examples are given further below. 

\begin{ex}[Binomial]
\label{ex:binomial}
``The fundamental problem of practical statistics'' \citep{pearson1920} concerns the case when the observable data is modeled as $(Y \mid \Theta=\theta) \sim \prob_{Y|\theta}=\bin(n,\theta)$, with $n$ a known positive integer.  The density for this model satisfies 
\[ p_\theta(y) = \binom{n}{y} \theta^n (1-\theta)^{n-y}, \quad y \in \YY, \quad \theta \in \TT, \]
where $\YY = \{0,1,\ldots,n\}$ and $\TT = [0,1]$.  The goal is to quantify uncertainty about the unknown $\Theta$, the binomial rate.  For the purpose of illustration, I'll consider three different forms of prior information---vacuous, complete, and partial---and present plots to visualize the IM's output and how the results depend on the form of the prior information.  To be clear, I'm not suggesting any ``default'' approaches for incorporation of complete or partial prior information; this choice needs to be made by the data analyst in their individual application.  Of course, one can take the vacuous prior approach as a ``default'' but this is only advisable if there's genuinely no reliable prior information on which to base the analysis.  As will be clear below, there's efficiency to be gained by incorporating at least some partial prior information when available.

Consider a very basic form of partial prior information, which can be easily described in words as {\em I'm 90\% sure that $\Theta$ is no more than 0.6}.  While it's rare in applications to have genuine and complete prior information, it's hard to imagine any application where genuine prior information of the simple form just presented isn't available.  But despite the simplicity and ubiquity of partial prior information of this form, this still has not been satisfactorily handled in the statistics literature.  Quantifying this partial prior information as a possibility measure simple too: just define the contour 
\[ q(\theta) = 1(\theta \leq 0.6) + 0.9 \, 1(\theta > 0.6), \quad \theta \in \TT,  \]
and the corresponding possibility measure $\uprior$ via consonance.  This can be readily incorporated into the analysis as described above, either via the full-blown partial-prior IM construction in Section~\ref{SS:partial} or the naive complexity-reduced construction in Section~\ref{SS:part.naive}.  



For this illustration, I consider six different scenarios: three different $n$ values, namely, $n \in \{8,16, 32\}$, and two different maximum likelihood estimator values, namely, $\hat\theta \in \{0.50, 0.75\}$.  These two $\hat\theta$ values are chosen so that one is compatible with the partial prior information while the other is borderline incompatible.  Figure~\ref{fig:binomial} shows the plausibility contours for the three IMs---vacuous prior, partial prior, and naive complexity-reduced partial prior---in each of the six $(n,\hat\theta)$ combinations.  Here are some key observations:
\begin{itemize}
\item As the sample size $n$ increases, the plausibility contours get more tightly concentrated around the $\hat\theta$ value as we go from top to bottom in Figure~\ref{fig:binomial}.  
\item In the left-hand column, where the data and prior are compatible, there's a varying degree of efficiency gain with the partial prior information compared to the vacuous prior.  On the one hand, the full-blown partial-prior IM gets some efficiency gain in the range $(0.6,1)$ compared to the vacuous-prior IM, there's an apparent efficiency loss in the $(0,0.6]$ range.  With the naive complexity-reduced partial-prior IM, the efficiency gain on $(0.6, 1)$ is minor at best but without the efficiency loss on $(0,0.6]$.  
\item The right-hand column shows the case where data and partial prior are borderline incompatible; of course, the data and vacuous prior can't be incompatible.  This kind of prior--data conflict can have counter-intuitive consequences.  When the data is relative uninformative, i.e., when $n=8$, the partial prior IM gives more weight to the prior, as expected.  As the data becomes more informative, the partial-prior IM concludes that both the data and prior are possible, and returns a two-mode plausibility contour.  For even larger $n$ (not shown), as expected, the mode at $\theta=0.6$ disappears and the partial-prior and vacuous-prior IMs merge.  Interestingly, the naive complexity-reduced partial prior IM seems to ``split the difference'' between the full-blown partial-prior IM and the vacuous-prior IM.  
\end{itemize} 
To me, neither of the two partial-prior IM solutions are fully satisfactory.  The initial proposal bets more on the available prior information (left column of Figure~\ref{fig:binomial}), which creates opportunities for efficiency gain, but appears that it could be more aggressive.  The naive complexity-reduced proposal shies away from gambling on the prior information, but takes the less-risky bet of closely following the vacuous-prior IM.  This illustrates the point made in Section~\ref{SS:part.naive}, i.e., that there ought to be a {\em via media} that balances the strengths of one and the weaknesses of the other.  
\end{ex}

\begin{figure}[p]
\begin{center}
\subfigure[$n=8$, $\hat\theta=0.50$]{\scalebox{0.55}{\includegraphics{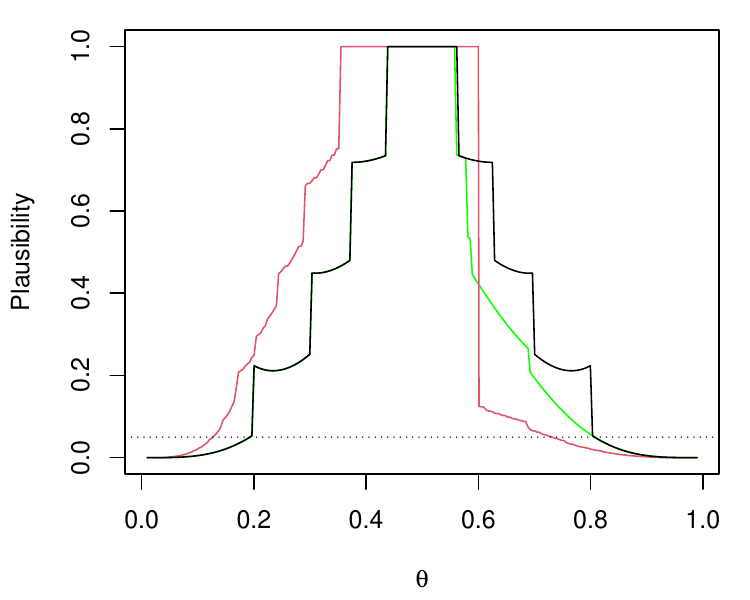}}}
\subfigure[$n=8$, $\hat\theta=0.75$]{\scalebox{0.55}{\includegraphics{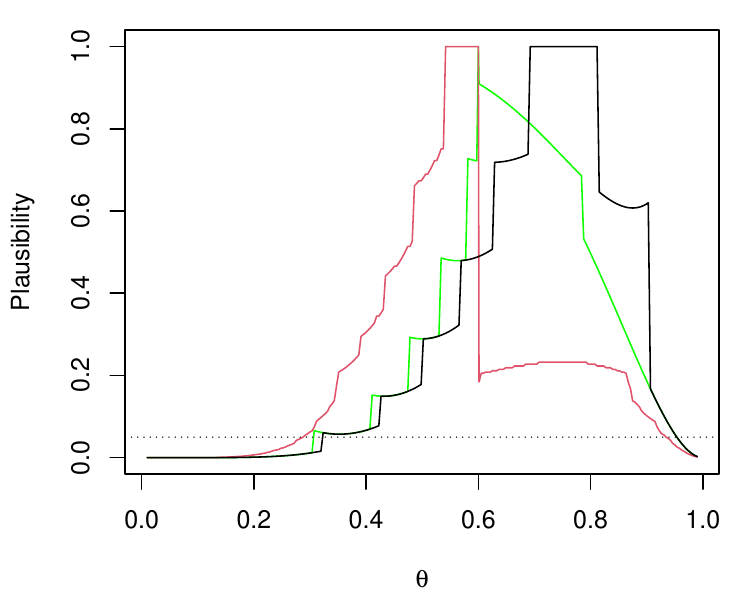}}}
\subfigure[$n=16$, $\hat\theta=0.50$]{\scalebox{0.55}{\includegraphics{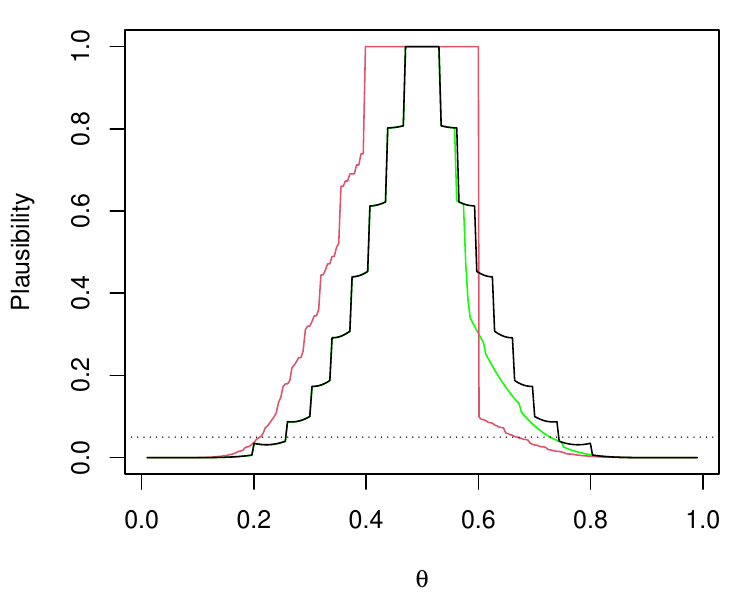}}}
\subfigure[$n=16$, $\hat\theta=0.75$]{\scalebox{0.55}{\includegraphics{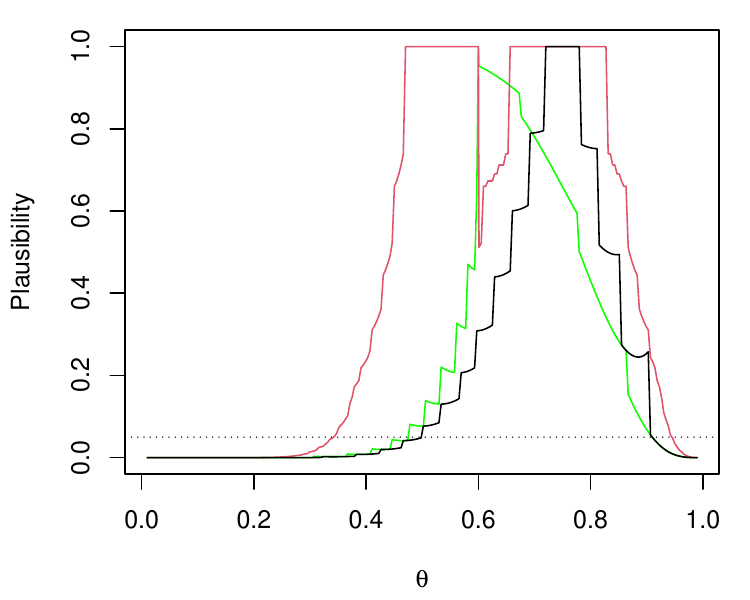}}}
\subfigure[$n=32$, $\hat\theta=0.50$]{\scalebox{0.55}{\includegraphics{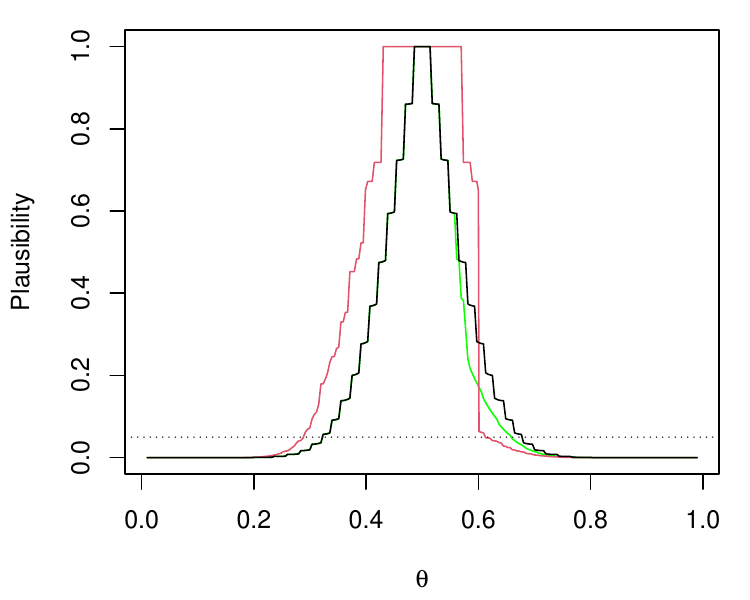}}}
\subfigure[$n=32$, $\hat\theta=0.75$]{\scalebox{0.55}{\includegraphics{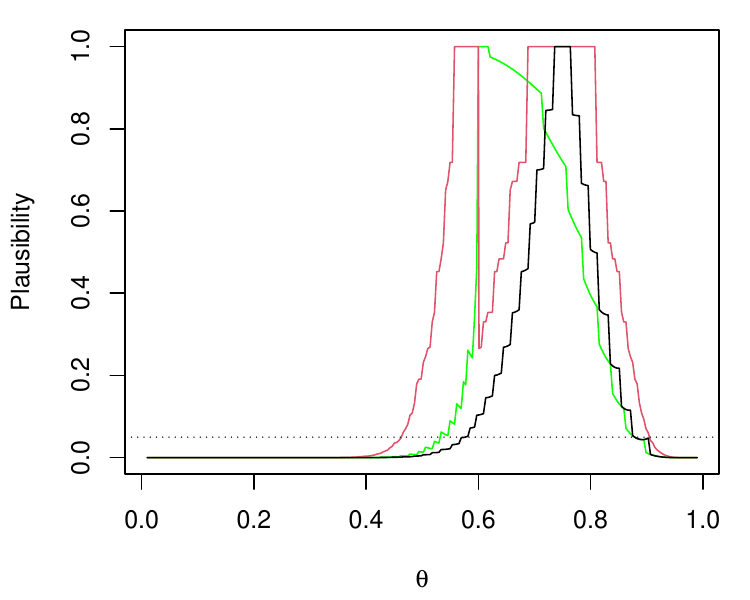}}}
\end{center}
\caption{Plausibility contours for the three IMs in each of the six $(n,\hat\theta)$ combinations in Example~\ref{ex:binomial}.  Black line is based on the vacuous prior, red line is based on the partial prior, and the green line is based on the naive complexity-reduced partial prior. In Panel~(b), the green line has a spike of height 1 at $\theta=0.6$, but this isn't visible in the plot.}
\label{fig:binomial}
\end{figure}

\begin{ex}[Normal]
Consider a simple case where $(Y \mid \Theta=\theta) \sim \prob_\theta=\nm(\theta, \sigma^2 n^{-1})$, where $\sigma > 0$ and $n \in \{1,2,\ldots\}$ are known. This is equivalent to starting with a collection of $n$ iid $\nm(\theta, \sigma^2)$ samples and reducing down via sufficiency.  As in the previous example, the goal here is to compare IMs based on vacuous, complete, and partial prior information.  Suppose there is reason to believe that $\Theta$ is near 0; this is common in cases where $\Theta$ represents, say, a difference between the mean responses corresponding to two similar treatments.  More concretely, suppose the available prior information says that, {\em a priori}, $\E(\Theta) = 0$ and $\E|\Theta| \leq K$, where $K > 0$ is a given constant that could, e.g., be elicited by asking subject matter experts how large they think the signal might be.  This prior information can be processed in various ways, but here's two:
\begin{itemize}
\item Take a precise prior that assumes $\Theta \sim \nm(0, \tau^2)$, where $\tau^2$ is chosen to match the elicited upper bound on the expected signal size, i.e., $\tau = K(\pi/2)^{1/2}$.  This is an embellishment on the prior information that was given, whose justification is based on making the computation of the Bayesian posterior distribution (and the complete-prior IM) straightforward. 
\item Take an imprecise prior that's consonant with contour 
\[ q(\theta) = \min\{1, K |\theta|^{-1} \}, \quad \theta \in \RR. \]
This is a so-called {\em Markov prior} based on the conversion of Markov's inequality into a class of imprecise probability distributions; see, e.g., \citet{dubois.etal.2004}.  This imposes far less structure on the prior distribution than that above.  It's also quite simple compared to the convoluted notion of ``weakly informative priors'' often advocated for in applied Bayesian analyses \citep[e.g.,][]{gelman.weak.2008}.  
\end{itemize} 
Figure~\ref{fig:normal} shows the plausibility contour for three different IMs: one that ignores the prior information, one that incorporates the complete prior in the first bullet point above, and one that incorporates the partial prior in the second bullet point.  There are two plots corresponding to two different observations, namely, $y=0.8$, which is consistent with the prior information, and $y=1.5$ which is just beyond what the prior would suggest is perfectly possible.  In the former case, the complete prior pulls the contour closer to the origin, whereas the partial prior IM keeps the peak at the observed $y$ while tightening the contour slightly compared to the vacuous prior IM.  In the latter case, both prior-dependent IMs pull the contour toward the origin, but the partial-prior IM shrinks less aggressively than the complete prior IM.
\end{ex}

\begin{figure}[t]
\begin{center}
\subfigure[$y=0.8$]{\scalebox{0.55}{\includegraphics{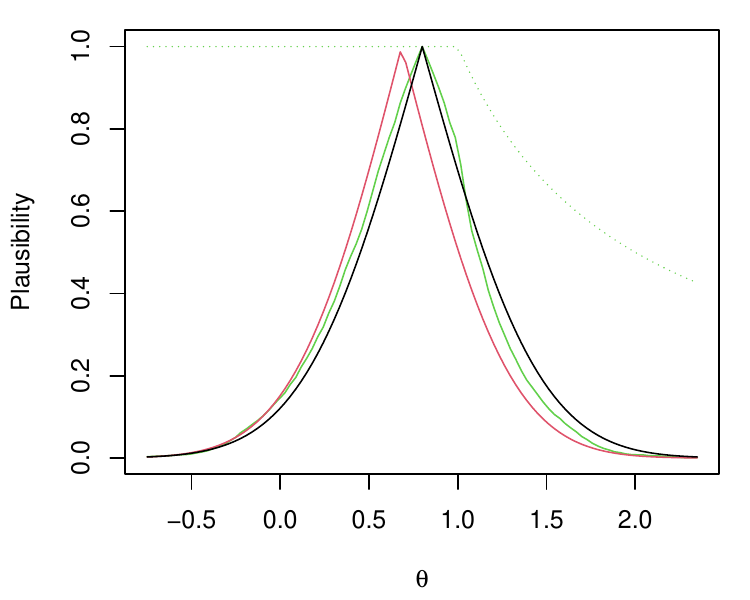}}}
\subfigure[$y=1.5$]{\scalebox{0.55}{\includegraphics{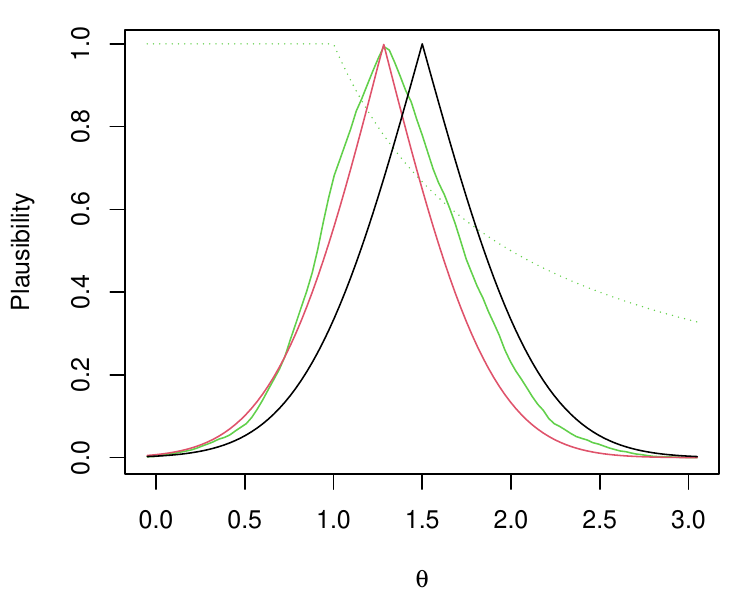}}}
\end{center}
\caption{Plots of the plausibility functions for the three IMs, based on different prior inputs, with $\sigma=2$, $n=15$, and $K=1$.  Black line is based on the vacuous prior, red line is based on the complete beta prior, and the green line is based on the partial prior; the dashed green line is the prior contour. The wiggliness in the plot of the partial prior IM contour is due to the Monte Carlo approximation.}
\label{fig:normal}
\end{figure}

\begin{ex}[Discrete uniform]
\label{ex:discrete.unif}
Suppose that, given $\Theta=\theta$, the data $Y=(Y_1,\ldots,Y_n)$ consists of iid random variables from a discrete uniform distribution, i.e., $\unif\{1,2,\ldots,\theta\}$, and the unknown parameter $\Theta$ takes values in $\TT=\{1,2,\ldots\}$, the (discrete) set of natural numbers.  Given $Y=y$, the likelihood function is 
\[ L_y(\theta) = \theta^{-n} \, 1(\theta \geq \hat\theta_y), \quad \theta \in \TT = \{1,2,\ldots\}, \]
where $\hat\theta_y$ is the sample maximum---the minimal sufficient statistic and the maximum likelihood estimator.  Start with the case of vacuous prior information about $\Theta$.  It's easy to show that the IM's plausibility contour has a closed-form expression:
\[ \pi_y(\theta) = \begin{cases}
( \hat\theta_y /\theta )^n & \text{for $\theta=\hat\theta_y, \hat\theta_y + 1, \ldots$} \\ 0 & \text{otherwise}. \end{cases} \]
Clearly this does not define a probability mass function since $\pi_y(\hat\theta_y)=1$; in fact, in the $n=1$ case, the series $\sum_{\theta=1}^\infty \pi_y(\theta)$ diverges.  Panels~(a) and (b) in Figure~\ref{fig:discrete.unif} display the IM's plausibility contour based on $\hat\theta_y = 10$ for $n=1$ and $n=2$, respectively.  The vertical lines are meant to emphasize that this is a function defined only on natural numbers, not on the entire real line.  With a sample of size $n=1$, the data is not especially informative, and this is reflected by how slow the plausibility contour decays as a function of $\theta$.  But with even just one more sample, i.e., $n=2$, it's possible to narrow down the range of sufficiently plausible values of $\theta$ considerably.  

Next, consider the potentially more interesting case where partial prior information about $\Theta$ is available.  As above, this is just one illustration of the kind of partial prior information that might be available, I'm not suggesting what follows as any sort of ``default'' choice for users.  The situation I have in mind is one where the investigator can elicit the function $\theta \mapsto \uprior(\Theta \geq \theta)$ and, for concreteness, suppose it's 
\[ \uprior(\Theta \geq \theta) = \frac{a}{a \vee \theta}, \quad \theta \in \TT, \quad \text{$a \in \TT$ fixed}. \]
The rationale behind this choice is as follows: recall that $\uprior$ serves as an upper bound on a precise prior distribution, so the above specification is saying simply that, {\em a priori}, I wouldn't pay more than \$$a/(a \vee \theta)$ for a gamble that pays \$1 if ``$\Theta \geq \theta$.''  One plausibility contour that's compatible with the above specification is 
\begin{equation}
\label{eq:discrete.unif.q}
q(\theta) = \frac{a}{a \vee \theta}, \quad \theta \in \TT. 
\end{equation}
This isn't the only choice of $q$ that's compatible with $\uprior$ above, but it's arguably the simplest---giving preference to the lower bound in the propositions ``$\Theta \geq \theta$'' is like Occam's razor.  The corresponding baseline plausibility ordering in \eqref{eq:eta.y} is given by 
\[ \eta(y,\theta) = \Bigl( \frac{\hat\theta_y}{\theta} \Bigr)^n \cdot \frac{a \vee \hat\theta_y}{a \vee \theta} \cdot 1(\theta \geq \hat\theta_y), \quad \theta \in \TT. \]
As a first step towards the partial prior-based IM plausibility contour, I get 
\begin{align*}
\prob_{Y|\vartheta}\{ \eta(Y,\vartheta) \leq \eta(y,\theta) \} & = \prob_{Y|\vartheta}\Big\{ \Bigl( \frac{\hat\theta_Y}{\vartheta} \Bigr)^n \frac{a \vee \hat\theta_Y}{a \vee \vartheta} \leq \Bigl( \frac{\hat\theta_y}{\theta} \Bigr)^n \frac{a \vee \hat\theta_y}{a \vee \theta} \Bigr\} \\
& = \sum_{u=1}^\vartheta 1\Big\{ \Bigl( \frac{u}{\vartheta} \Bigr)^n \, \frac{a \vee u}{a \vee \vartheta} \leq \Bigl( \frac{\hat\theta_y}{\theta} \Bigr)^n \, \frac{a \vee \hat\theta_y}{a \vee \theta} \Bigr\} \,  f_\vartheta(u), \quad \theta \geq \hat\theta_y, 
\end{align*}
where $f_\vartheta(u)$ is the mass function for $\hat\theta_Y$ under the iid $\unif\{1,2,\ldots,\vartheta\}$ model, i.e., 
\[ f_\vartheta(u) = \Bigl( \frac{u}{\vartheta} \Bigr)^n - \Bigl( \frac{u-1}{\vartheta} \Bigr)^n, \quad u=1,2,\ldots,\vartheta. \]
The expression in the penultimate display can be evaluated numerically for any $\vartheta$, hence the IM's plausibility contour can be evaluated using numerical integration:
\begin{align*}
\pi_y(\theta) & = \int_0^1 \Bigl[ \sup_{s: q(\vartheta) > s} \prob_{Y|\vartheta}\{ \eta(Y,\vartheta) \leq \eta(y,\theta) \}  \Bigr] \, ds 
\end{align*}
Panels~(c) and (d) in Figure~\ref{fig:discrete.unif} shows the plausibility contours for $n=1$ and $n=2$, respectively, with the partial prior depending on $a=7$.  The key observation here is that some efficiency is gained by being able to discount relatively large $\theta$-values thanks to the partial prior contour that vanishes (albeit slowly) in the tail.  
\end{ex}

\begin{figure}[t]
\begin{center}
\subfigure[vacuous prior, $n=1$]{\scalebox{0.55}{\includegraphics{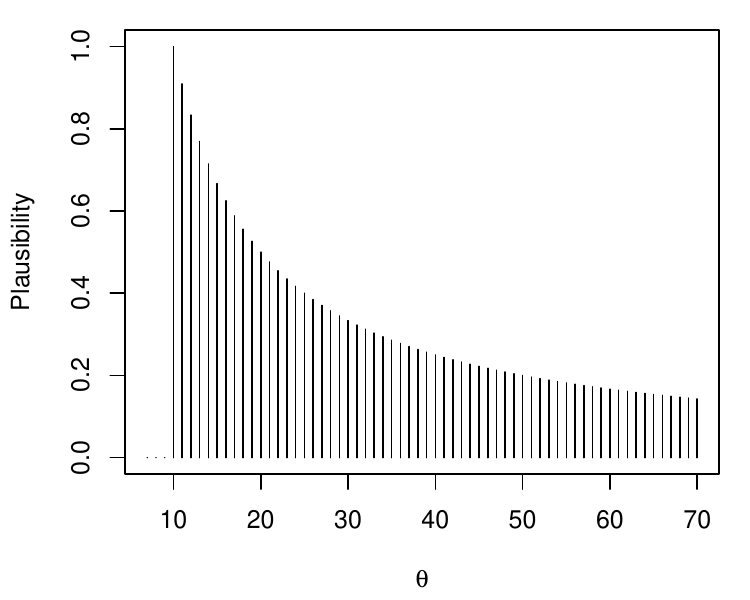}}}
\subfigure[vacuous prior, $n=2$]{\scalebox{0.55}{\includegraphics{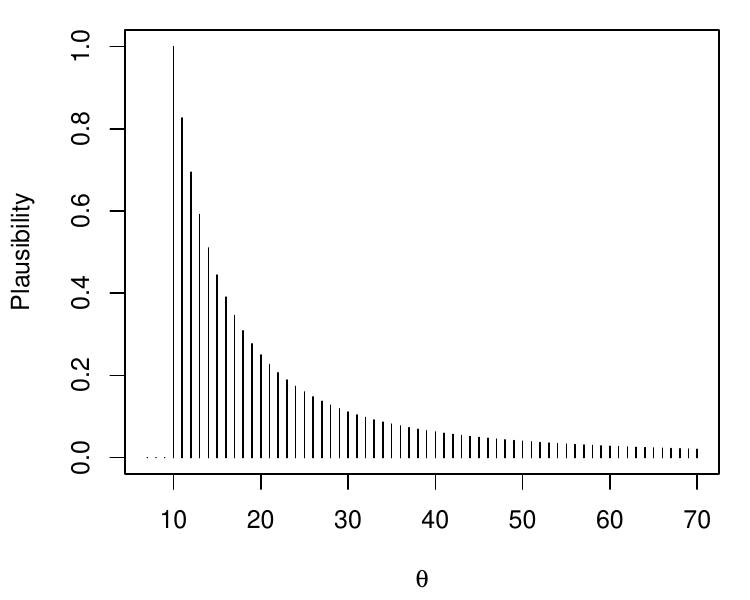}}}
\subfigure[partial prior, $n=1$]{\scalebox{0.55}{\includegraphics{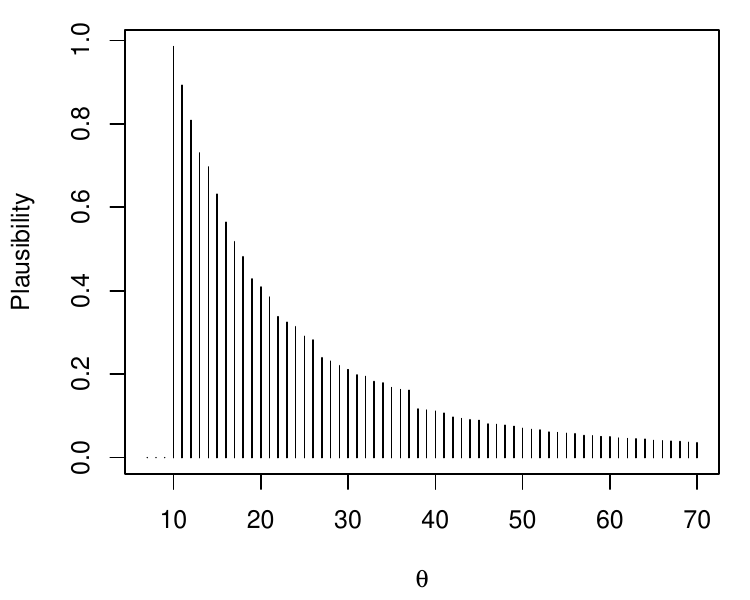}}}
\subfigure[partial prior, $n=2$]{\scalebox{0.55}{\includegraphics{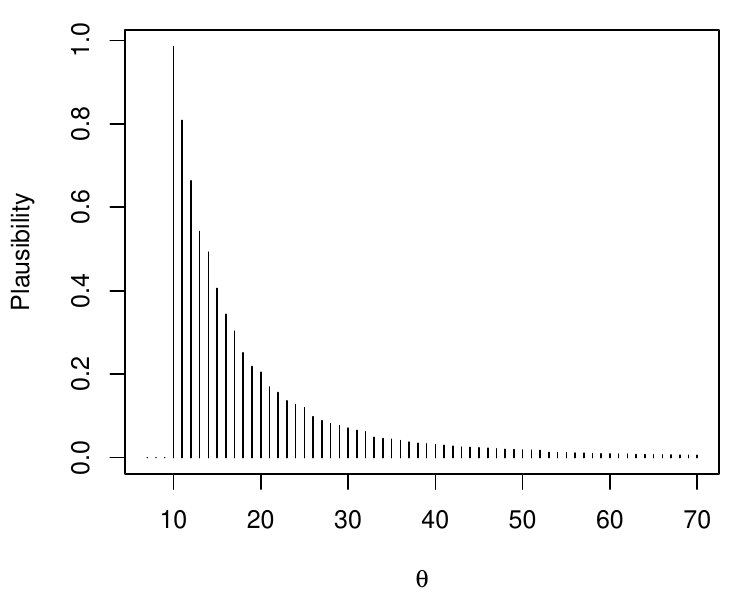}}}
\end{center}
\caption{Plots of the plausibility contour functions $\theta \mapsto \pi_y(\theta)$ in the discrete uniform case presented in Example~\ref{ex:discrete.unif}.  In Panels~(c) and (d), the partial prior being used has contour function $q$ in \eqref{eq:discrete.unif.q} depending on constant $a=7$.}
\label{fig:discrete.unif}
\end{figure}

The next two examples are special in the sense that, despite appearing rather simple at first look, there are lurking challenges.  Because of these lurking challenges, many statisticians would consider these two examples as foundational test cases---if a framework for statistical inference can't satisfactorily handle these (and perhaps other) cases, then the framework itself needs more work.  For example, in both of the examples below, the minimal sufficient statistic has dimension greater than that of the unknown parameter and, consequently, Fisher's original fiducial argument can't be applied; and since Fisher's proposed strategy for accommodating examples like these---namely, conditioning---also doesn't work here, these test cases raise serious questions about the general viability of Fisher's fiducial argument.  Generalized fiducial and default-prior Bayes but can only offer (asymptotically) approximate solutions.  The new likelihood-driven IM construction that I'm advocating here turns out to perform very well in these two examples, offering exact validity and efficiency right off the shelf.  For these two examples, and some of the others in the later section, I'll focus on the vacuous prior case since that's the easiest to compare against existing solutions in the literature. 

\begin{ex}[Uniform, linked endpoints]
\label{ex:uniform}
Let $X=(X_1,\ldots,X_n)$ be iid observables and assume a continuous uniform model $\unif\{a(\theta), a(\theta) + b(\theta)\}$, where the functions $a$ and $b > 0$ are known but the true $\Theta$ value is unknown.  Here I focus on the case where $b(\theta)$ is non-constant; the case with constant $b$ happens to be fundamentally different and will be considered in Example~\ref{ex:location} below.  The prototype is $\unif(\theta, \theta^2)$, with $\TT = (1,\infty)$, so that $a(\theta) = \theta$ and $b(\theta) = \theta^2 - \theta$.  Regardless of the form $a$ and $b$ take, the minimal sufficient statistic is $Y = (Y_1,Y_2)$, with $Y_1 = \min_i X_i$ and $Y_2 = \max_i X_i$.  In what follows, I'll treat the observable data as $Y$ and let $\prob_{Y|\theta}$ denote the distribution of $Y$ induced by that of $X$.  This is a non-regular problem because the distribution's support depends on the parameter, which explains the mismatch between the dimensions of $\Theta$ and of $Y$.  This non-regularity also poses some challenges for default-prior Bayesian analyses that are often based on Jeffreys prior, which depends on the Fisher information matrix that might not exist in such cases; see, e.g., \citet{bergerbernardosun2009} and \citet{shemyakin2014}.  Care was needed in \citet[][Ex.~4]{hannig.review} to develop an efficient generalized fiducial solution in this example, which, like the aforementioned default-prior Bayes solution, only offers asymptotically approximate confidence limits.  Interestingly, the proposed IM solution is easy, exactly valid, and efficient compared to these alternatives.

The likelihood function is simple in this case, though maximizing it for the purpose of finding the relative likelihood need not be straightforward; it depends on the form of $a$ and $b$.  In the prototypical version above, with $a(\theta)=\theta$ and $b(\theta)=\theta^2 - \theta$, the maximum likelihood estimator is $\hat\theta_y = y_2^{1/2}$, so the relative likelihood is 
\[ \eta(y, \theta) = \Bigl\{ \frac{y_2 - y_2^{1/2}}{\theta(1-\theta)} \Bigr\}^n \times 1\{y_1 \geq \theta, \, y_2 - y_1 \leq \theta^2-\theta\}, \quad \theta > 1. \]
There's no closed-form expression for the IM's possibility contour, but a Monte Carlo approximation is straightforward.  Indeed, under $\prob_\theta$, $Y_1$ isn't needed and $Y_2$ has the same distribution as $\theta + (\theta^2 - \theta) U_2$, where $U_2 \sim \bet(n,1)$, so a simple Monte Carlo approximation is, for $\theta \in [y_2^{1/2}, y_1]$, 
\[ \pi_y(\theta) \approx \frac1M \sum_{m=1}^M 1\bigl[ \theta + (\theta^2 - \theta)U_2^m - \{\theta + (\theta^2-\theta) U_2^m\}^{1/2} \leq y_2 - y_2^{1/2} \bigr], \]
where $\{U_2^m: m=1,\ldots,M\}$ are iid $\bet(n,1)$; $\pi_y(\theta) = 0$ otherwise.  This function is strictly decreasing on $\theta \in [y_2^{1/2}, y_1]$ so from here it's easy to compute $\uPi_y(A) = \pi_y(\inf A)$ for any relevant $A \subseteq \TT$ and to extract the IM confidence interval $[y_2^{1/2}, \pi_y^{-1}(\alpha)]$.  

An alternative IM for this problem was put forward in \citet{imunif} using the specialized tool developed in \citet{imcond}, in particular, their so-called localized IM construction.  This is a powerful general strategy, yielding valid and highly efficient inference, but the details are often non-trivial to work out.  \citet{imunif} showed that their construction in the $\unif(\theta,\theta^2)$ example is provably valid and produces confidence intervals that are no less efficient than the very good Bayesian and generalized fiducial confidence limits that only achieve the nominal coverage asymptotically.  An inspection of the solutions put forward here and in \citet{imunif} wouldn't suggest that there's any connection between the two but, remarkably, in my experiments I found that the two solutions are (at least numerically) the same.  So the thorough comparisons presented in \citet{imunif}, which show that their IM is provably valid and as efficient as existing methods, apply directly to the likelihood-based IM proposed here.  

As an illustration, I take the same data example from \citet[][Ex.~4]{hannig.review}, which was also used in \citet[][Sec.~3.3]{imunif}.  That is, the observed data is $y=(281.1, 9689.7)$ based on $n=25$, so that the maximum likelihood estimator is $\hat\theta_y = \sqrt{9689.7} \approx 98.4$.  Figure~\ref{fig:uniform} shows the possibility contours for the IM proposed in \citet{imunif} with that proposed here overlaid.  It's clear that the two solutions are the same, modulo Monte Carlo variation.   
\end{ex}

\begin{figure}[t]
\begin{center}
\scalebox{0.7}{\includegraphics{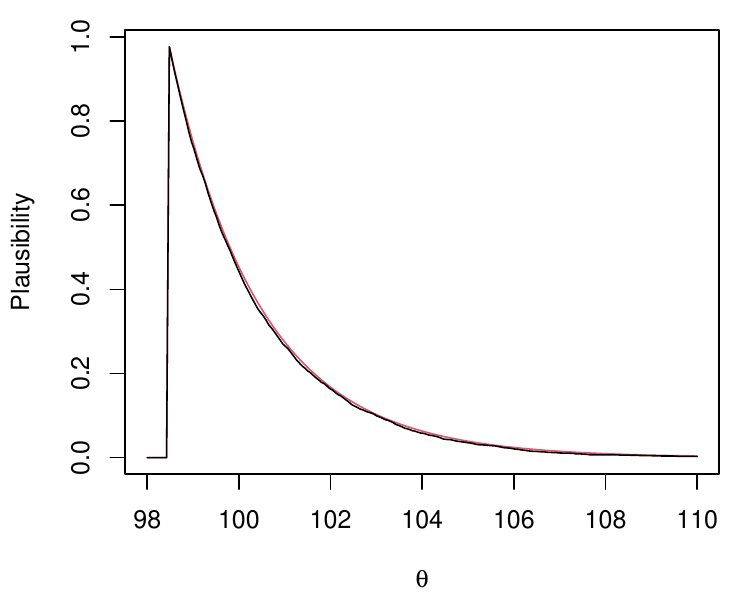}}
\end{center}
\caption{Possibility contours based on the old IM (red) in \citet{imunif} and the new IM here (black) for the $\unif(\theta,\theta^2)$ model, based on a sample of size $n=25$ with sufficient statistic $y=(281.1, 9689.7)$.}
\label{fig:uniform}
\end{figure}

\begin{ex}[Bivariate normal correlation]
\label{ex:bvn.cor}
Consider the classical standard bivariate normal model which has known mean and variance, taken to be 0 and 1, respectively, but with unknown correlation $\Theta \in \TT=[-1,1]$.  This example is considered in \citet{basu1964} and in many other textbook treatments.  More specifically, suppose that, $X=(X_1,\ldots,X_n)$ consists of $n$ iid random variable pairs $X_i = (X_{i1}, X_{i2})$ which are bivariate normal with mean 0, variance 1, and unknown correlation $\Theta \in \TT=[-1,1]$.  It's easy to verify that the minimal sufficient statistic is the pair $Y=(Y_1,Y_2)$ where 
\[ Y_1 = \sum_{i=1}^n (X_{i1}^2+X_{i2}^2) \quad \text{and} \quad Y_2 = \sum_{i=1}^n X_{i1}X_{i2}. \] 
Note, as in the previous example, the mismatch in dimension between the model parameter and the minimal sufficient statistic.  A solution based on the sampling distribution of, say, the maximum likelihood estimator would suffer a loss of information, but there's no clear guidance on what ancillary statistic one should condition on to recover the lost information.  Indeed, both marginal data sets $(X_{i1},\ldots,X_{n1})$ and $(X_{i2},\ldots,X_{n2})$ are legitimate ancillary statistics, so it's not clear how to proceed.  For this reason, care is needed in developing valid and efficient methods for inference on $\Theta$; this includes the default-prior Bayesian solution \citep[e.g.,][Ex.~2]{ong.mukerjee.2010},  likelihood-based solution \citep[][Ex.~4.3]{reid2003}, and generalized fiducial \citep{majumdar.hannig.2015}.  Again, the confidence limits derived from these solutions only offer approximate coverage probability guarantees asymptotically.  The IM solution presented here is relatively easy, exactly valid, and efficient compared to these alternative methods.  

The maximum likelihood estimator has no closed-form expression in this case, but it's easy to find numerically.  Once this is found, the relative likelihood can be evaluated numerically.  So the basic Monte Carlo strategy presented in \eqref{eq:vacuous.MC} can be applied here without any real difficulty to evaluate the IM's possibility contour, and its exact validity follows from the general theory.  For comparison, an alternative IM solution was presented in \citet{imcond} using the localization strategy mentioned briefly in the uniform example above.  The derivation of that solution is non-trivial, but it is exactly valid and was shown in simulations to be more efficient than other existing methods.  

For a numerical comparison, I'll consider a real data example on law school admissions presented in \citet{efron1982} consisting of $n=15$ data pairs $(Y_1,Y_2)$, with $Y_1 = \text{LSAT scores}$ and $Y_2=\text{undergrad GPA}$.  For this analysis, I'll standardize both variables so that the mean zero--unit variance is appropriate.  Of course, this standardization has no effect on the correlation, which is our object of interest.  In this case, the sample correlation is 0.776 and the maximum likelihood estimator is $\hat\theta = 0.789$.  Figure~\ref{fig:bvn.cor} shows the possibility contour plots for the two IM solutions, i.e., the one in \citet{imcond} and the one based on the direct, likelihood-based solution based on the formulation presented here in this paper.  The two IM contour functions are similar but, unlike in the Example~\ref{ex:uniform}, they're not exactly the same.  By definition, the new likelihood-based IM possibility contour is maximized at the maximum likelihood estimator, which is (asymptotically) efficient; the old IM solution gives a contour that's maximized at neither the maximum likelihood estimator nor the sample correlation, so this is difficult to justify.  For this and other reasons, I prefer the new likelihood-based IM solution.  In any case, the range of sufficiently plausible $\theta$ values based on the two solutions in Figure~\ref{fig:bvn.cor} is almost the same and, since the old solution are efficient, the same must be true of the new IM solution.  

To justify the above claim concerning the efficiency of new, likelihood-based IM solution, I carry out a brief simulation study.  I consider two different sample sizes, $n \in \{10, 25\}$, and five different true correlations, $\Theta \in \{0.05, 0.25, 0.50, 0.75, 0.90\}$.  For each $(n,\Theta)$ combination, I generate 1000 data sets and from each data set I extract a 95\% confidence interval based on the new IM solution and the $r^\star$ solution based on details laid out in \citet[][Ex.~4.3]{reid2003}.  Based on the 1000 confidence intervals for each method, I get empirical estimates of their coverage probability and expected length.  These results, summarized in Table~\ref{table:bvn}, suggest that the exactness of the IM solution leads to empirical coverage probabilities closer to the 0.95 target compared to the $r^\star$ intervals; in particular, the $r^\star$ intervals significantly under-cover across the $\Theta$-board in the $n=10$ case.  Naturally, since $r^\star$ isn't exact, the corresponding intervals tend to be shorter, but that's no justification for giving up on exactness.  In any case, the difference between the two methods' lengths closes quickly as $n$ increases.  So, the take-away message is that the IM solution is both exact and (nearly) as efficient as the non-exact $r^\star$ solution. 
\end{ex}

\begin{figure}[t]
\begin{center}
\scalebox{0.7}{\includegraphics{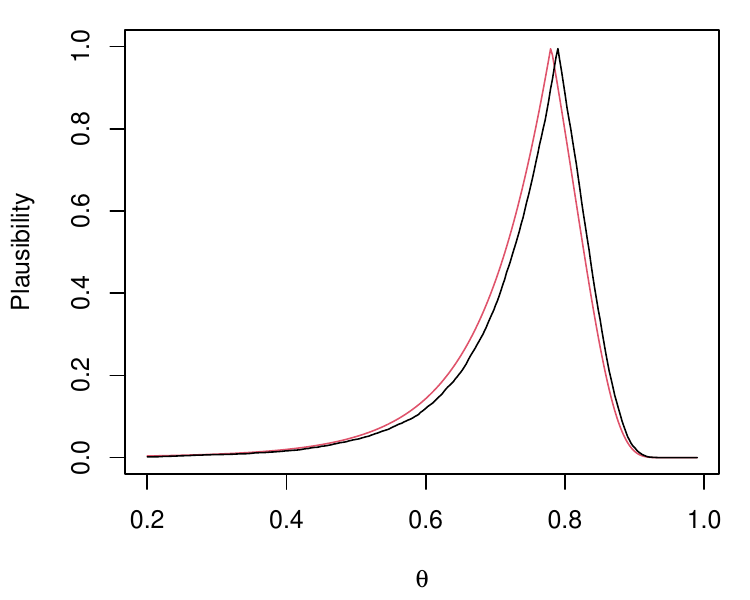}}
\end{center}
\caption{Possibility contours based on the old IM (red) in \citet{imcond} and the new IM here (black) for the standard bivariate normal correlation model, based on Efron's law school admissions data as explained in Example~\ref{ex:bvn.cor}.}
\label{fig:bvn.cor}
\end{figure}

\begin{table}[t]
\begin{center}
\begin{tabular}{cccccccc}
\hline
& & & & & $\Theta$ & & \\
\cline{4-8} 
$n$ & Metric & Method & 0.05 & 0.25 & 0.50 & 0.75 & 0.90 \\
\hline 
10 & Coverage & IM & 0.957 & 0.968 & 0.953 & 0.947 & 0.942 \\
& & $r^\star$ & 0.927 & 0.923 & 0.933 & 0.930 & 0.923 \\
& Length & IM & 1.004 & 0.979 & 0.888 & 0.619 & 0.242 \\
& & $r^\star$ & 0.947 & 0.919 & 0.818 & 0.540 & 0.238 \\
\hline
25 & Coverage & IM & 0.949 & 0.960 & 0.957 & 0.950 & 0.941 \\
& & $r^\star$ & 0.938 & 0.943 & 0.948 & 0.944 & 0.942 \\
& Length & IM & 0.729 & 0.695 & 0.574 & 0.311 & 0.117 \\
& & $r^\star$ & 0.695 & 0.662 & 0.545 & 0.309 & 0.127 \\
\hline
\end{tabular}
\end{center}
\caption{Empirical coverage probabilities and expected lengths for the IM- and $r^\star$-based 95\% confidence intervals in the bivariate normal correlation case described in Example~\ref{ex:bvn.cor}.}
\label{table:bvn}
\end{table}



\subsection{Examples involving conditioning}

The examples above were chosen specifically because conditioning on ancillary statistics either wasn't needed or wasn't possible.  The next few examples, on the other hand, consider cases where conditioning can/should be done.  These are meant to illustrate what conditioning offers and how it can be carried out within the proposed IM framework.  Since the focus here is specifically on conditioning, I'll assume for simplicity that the prior information is vacuous.  This vacuous-prior case also makes it easier to compare with existing solutions in the literature. 

\begin{ex}[First-order autoregression]
\label{ex:ar1}
Here is a simple and somewhat extreme example to highlight the differences that can arise between conditional and unconditional IM solutions; this is based on Example~7.2 in \citet{fraser2004}.  Let $Y=(Y_1,Y_2)$ denote the observables from the first-order autogressive model 
\[ Y_1 \sim \nm(0, \sigma^2) \quad \text{and} \quad (Y_2 \mid Y_1=y_1, \Theta=\theta) \sim \nm(\theta y_1, \sigma^2), \]
where $\sigma^2 > 0$ is known but $\Theta \in (-1,1)$ is unknown; my restriction to the interval $(-1,1)$ is just to avoid challenges associated with non-stationarity.  For simplicity, I'll assume the prior information about $\Theta$ is vacuous.  The likelihood function is 
\[ p_\theta(y) = p(y_1) \, p_\theta(y_2 \mid y_1) \propto e^{-(y_2 - \theta y_1)^2 / 2\sigma^2}, \]
and maximum likelihood over $\theta$, attained at $\hat\theta = y_2 / y_1$, is constant in $y$ and, therefore, can be ignored.  That is, the plausibility order is defined as 
\[ \eta(y, \theta) = e^{-(y_2 - \theta y_1)^2 / 2 \sigma^2}, \quad (y,\theta) \in \RR^2 \times (-1,1). \]
Since $Y_1$ is an ancillary statistic, there's a choice that has to be made:
\begin{itemize}
\item condition on the observed value of the ancillary statistic and carry out the IM construction using the conditional distribution of $S(Y)=Y_2$, given $U(Y)=Y_1$, as described above, or 
\vspace{-2mm}
\item ignore the fact that $Y_1$ is ancillary and carry out the IM construction using the joint distribution of $Y=(Y_1, Y_2)$. 
\end{itemize} 
In this simple case, the IMs based on both the conditional and unconditional perspectives can be worked out in closed-form, and the two contour functions are, respectively, 
\begin{align*}
\pi_{y_2|y_1}(\theta) & = 1 - {\tt pchisq}\Bigl( \frac{(y_2 - \theta y_1)^2}{\sigma^2}, {\tt df} = 1 \Bigr) \\
\pi_y(\theta) & = 1 - {\tt pchisq}\Bigl( \frac{(y_2 - \theta y_1)^2}{\sigma^2\{1 + \theta^2 + \theta/(1-\theta^2)\}}, {\tt df} = 1 \Bigr).
\end{align*}
As described above, both the conditional- and unconditional-based IMs are strongly valid, so the choice between them must be made based on other considerations.  One such consideration is that the conditional-based IM is tailored to the configuration---determined by $y_1$---of the sample, which is a general feature of conditional inference.  More pragmatically, a visualization might be helpful, so a plot of the two contour functions based on data $y=(0.0, 0.2)$ and $\sigma=0.1$ is shown in Figure~\ref{fig:ar1}.   This plot reveals that the unconditional-based IM contour function is sensitive to certain singularities, while the conditional-based IM is much more stable.  Clearly the conditional-based IM is better.  
\end{ex}

\begin{figure}[t]
\begin{center}
\scalebox{0.65}{\includegraphics{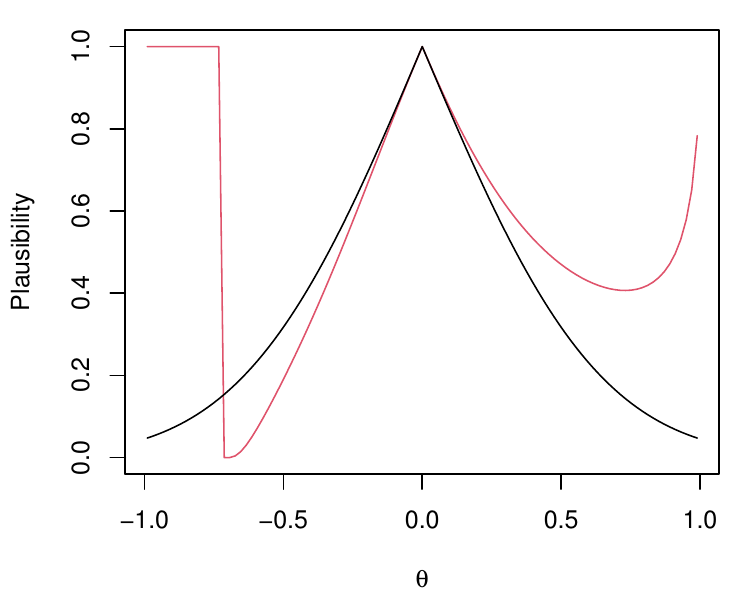}}
\end{center}
\caption{Plot of the conditional- (black) and unconditional-based (red) IM contours in the first-order autoregressive model in Example~\ref{ex:ar1}, with $y=(0, 0.2)$ and $\sigma=0.1$.}
\label{fig:ar1}
\end{figure}

\begin{ex}[Fisher's ``Problem of the Nile'']
\label{ex:nile}
Let $Y=(Y_1,Y_2)$, where $Y_1$ and $Y_2$ themselves are vectors $Y_j = (Y_{j1},\ldots,Y_{jn})$, where 
\[ (Y_{1i} \mid \Theta=\theta) \iid \expo(\theta^{-1}) \quad \text{and} \quad (Y_{2i} \mid \Theta=\theta) \iid \expo(\theta), \quad i=1,\ldots,n. \]
That is, there are $2n$ total samples, all mutually independent, but the $Y_{1i}$'s are exponential with mean $\theta^{-1}$ and the $Y_{2i}$'s are exponential with mean $\theta$.  This problem is a classic one in the conditional inference literature, see Example~1 in \citet{ghoshreidfraser2010}, commonly referred to as (a version of) Fisher's gamma hyperbola \citep[e.g.,][]{efron.hinkley.1978, reid2003, bn.cox.1994}.  Fisher described this model in the context of an  investigation into the fertility of land in the Nile river valley \citep[][p.~169]{fisher1973}.  

The factorization theorem implies that $\{\sum_{i=1}^n Y_{1i}, \sum_{i=1}^n Y_{2i}\}$ is a minimal sufficient statistic.  Equivalently, the pair $\{S(Y), U(Y)\}$, with 
\[ S(Y) = \Bigl( \frac{\sum_{i=1}^n Y_{2i}}{\sum_{i=1}^n Y_{1i}} \Bigr)^{1/2} \quad \text{and} \quad U(Y) = \Bigl(\sum_{i=1}^n Y_{1i} \cdot \sum_{i=1}^n Y_{2i} \Bigr)^{1/2}, \]
make up a minimal sufficient statistic and, moreover, $S(Y)$ is the maximum likelihood estimator and $U(Y)$ is ancillary.  The conditional distribution of $S=S(Y)$, given $U=U(Y)$, under the above model is known to be a generalized inverse Gaussian distribution \citep{bn1977, bn1983} with density 
\[ p_\theta(s \mid u) \propto s^{-1} \exp\{-u(\theta/s + s / \theta)\}, \quad s > 0. \]
This density can be used in the (conditional) IM plausibility contour construction above and, in particular, the suite of {\tt ginvgauss} functions in the R package {\tt rmutil} \citep{R:rmutil} can be used to carry out the probability calculations, either via Monte Carlo and {\tt rginvgauss} or by numerical integration and the density {\tt dginvgauss}.  

Figure~\ref{fig:nile} shows the IM plausibility contour based on data of size $n=20$ with maximum likelihood estimator $S=1$ and with ancillary statistic $U=u$, for $u \in \{10,20\}$.  The plots are on the same scale, so it's easy to see that the larger $u$ corresponds to a slightly more concentrated plausibility contour; so $u$ serves as a sort of ``informativeness'' measure for the sample, which is why, as Fisher argues, conditioning on its value makes intuitive sense.  As an alternative, like in Example~\ref{ex:ar1}, I could've ignored the conditioning and just used the basic likelihood function formulation to construct a (not-conditional) IM for $\Theta$.  Compared to the autoregressive example above, here there's no difference between the two contours---this is due to the scale parameter structure of the problem.  If there's any difference between the two contours, it's that the not-conditional version is a little more wiggly than the conditional version.  The reason for this is the conditional version is based on a lower-dimensional Monte Carlo integral approximation, so there's a Rao--Blackwellization that reduces the variability.  Just to be clear, however, the two contours are the same, the tiny difference in the plot is due to Monte Carlo variability.  
\end{ex}

\begin{figure}[t]
\begin{center}
\subfigure[$u=10$]{\scalebox{0.55}{\includegraphics{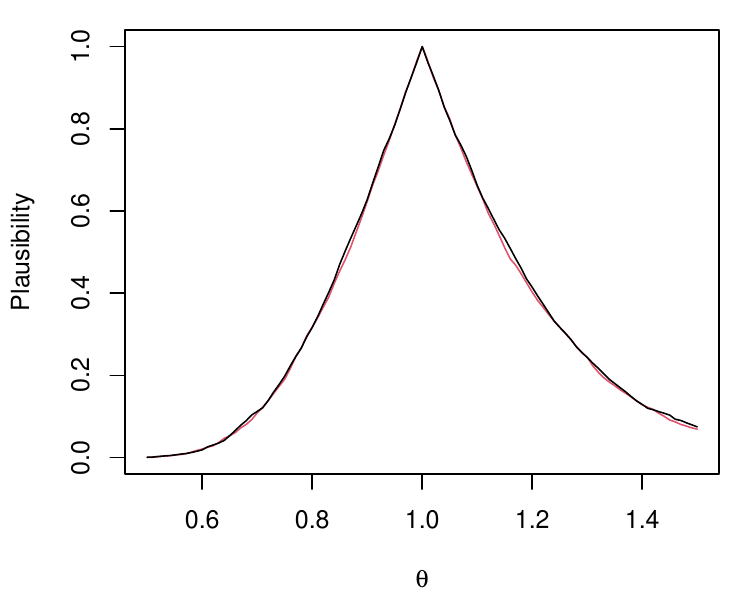}}}
\subfigure[$u=20$]{\scalebox{0.55}{\includegraphics{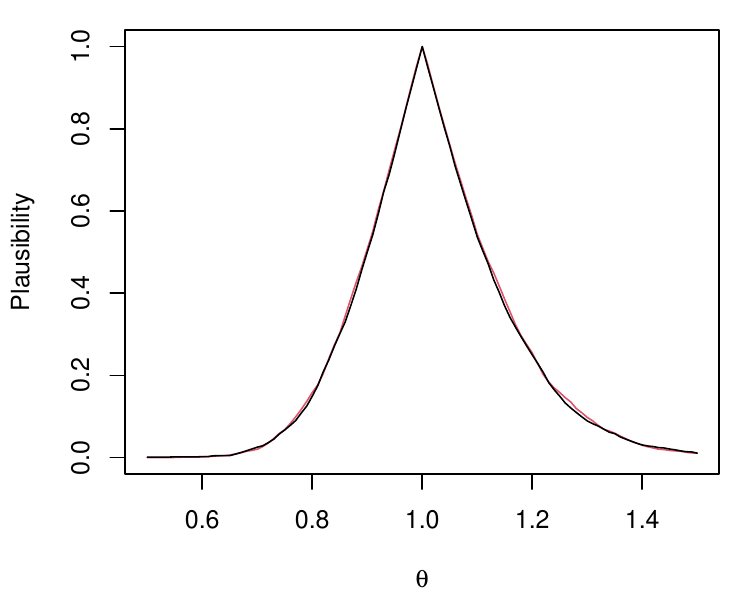}}}
\end{center}
\caption{Plot of the IM plausibility contour for the parameter $\Theta$ in the Nile/hyperbola model in Example~\ref{ex:nile}.  Both plots are based on $n=20$ and the observed maximum likelihood estimator equal to 1, which is where the contour peaks; all that varies between the panels is the value of the ancillary statistic $U=u$.  Black lines are the conditional IM contour, red lines are the not-conditional contour.}
\label{fig:nile}
\end{figure}

\begin{ex}[Location parameters, etc]
\label{ex:location}
For the last example in this subsection, I'd like to show something a little different.  To keep things relatively concrete, I'll consider a location parameter problem where $Y=(Y_1,\ldots,Y_n)$ consists of $n$ iid samples with $Y_i = \Theta + Z_i$, where the $Z_i$'s are iid with a continuous distribution having density $p$; then the density of $Y_i$, given $\Theta=\theta$, is $p_\theta(y) = p(y - \theta)$.  There are lots of examples of this type---in fact, the Nile problem in Example~\ref{ex:nile} is of this form, modulo suitable transformations \citep[][Ex.~3.2]{efron.hinkley.1978}.  Moreover, the location parameter structure is a special case of general kinds of transformation structures \citep[e.g.,][]{fraser1968, eaton1989}, and the IM-specific results summarized here hold in the general case too; for these group-theoretic details, see \citet{martin.isipta2023}. 

A key feature that the location parameter problem---and the aforementioned more general structures---is that there exists a natural/best ancillary statistic upon which to condition.  The exact form is problem-specific, but it generally takes the form of a difference, residual, etc.  For example, if the model is $\{\unif(\theta,\theta+1): \theta \in \RR\}$, then the minimal sufficient statistic is $(Y_{(1)}, Y_{(n)})$, the extreme order statistics, and this is equivalent to the pair $(S,U)=(Y_{(1)}, Y_{(n)}-Y_{(1)})$, where $S=S(Y)$ is the maximum likelihood estimator of the unknown $\Theta$ and $U=U(Y)$ is the natural ancillary statistic.  The difference $U$ quantifies the ``information'' in the same, with large values of the difference indicating sharper inference on $\Theta$; compare this to the role played by ``$U$'' in the Nile example.  All location parameter problems admit a similar decomposition $Y \mapsto (S,U)$ where $U$ is ancillary, though the dimension of the $U$ component can vary from one problem to the next.  Following the general discussion above, all that's needed to construct the valid and efficient IM solution is the conditional distribution of $S$, given $U=u$.  

What's ``different'' about what I want to show here is that there's another way to produce the IM solution in structured problems like these.  \citet{martin.isipta2023} demonstrated a connection between the IM solution and the default-prior Bayes (and fiducial) solutions.  Specifically, the IM's upper probability is just the probability-to-possibility transform applied to $\prior_y$, the Bayesian posterior distribution for $\Theta$, given $Y=y$, based on the default right Haar prior.  In location parameter problems, the right and left Haar priors are the same and equal to Lebesgue measure---an improper prior distribution.  Regardless of the prior impropriety, this formulation admits a proper posterior, so I can get the IM solution exactly like in the complete-prior case discussed in Section~\ref{SS:complete}.  That is, the IM's plausibility contour is given by 
\[ \pi_y(\theta) = \prior_y\{ q_y(\Theta) \leq q_y(\theta) \}, \quad \theta \in \TT, \]
where $q_y$ is the density of $\prior_y$.  This can be approximated via Monte Carlo as 
\[ \pi_y(\theta) \approx \frac1M \sum_{m=1}^M 1\{ q_y(\Theta_m) \leq q_y(\theta)\}, \quad \theta \in \TT, \]
where $\{\Theta_m: m=1,\ldots,M\}$ is an iid sample from the posterior $\prior_y$.  To be clear, this connection to the Bayes solution is just a coincidence resulting from the problem structure---I'm not assuming that the default prior is ``correct'' and the validity achieved by this IM is with respect to the vacuous prior assumption.  Moreover, whatever good properties the default-prior Bayes solution satisfies are inherited by the IM solution, e.g., the Bayesian credible intervals are exact confidence intervals in these structured problems.  

For a quick numerical illustration, I'll consider the Cauchy location problem, where $Y=(Y_1,\ldots,Y_n)$, given $\Theta$, is an iid sample of size $n$ from $\cau(\theta,1)$.  In this case, the minimal sufficient statistic is $Y$, so sufficiency alone provides no reduction in dimension.  However, there are various $(S,U)$ decompositions that all yield the same relevant conditional distribution, e.g., 
\[ S = Y_1 \quad \text{and} \quad U=\{Y_2-Y_1,\ldots,Y_n-Y_1\} \]
or 
\[ S = \text{MLE} \quad \text{and} \quad U=\{Y_1-S, \ldots, Y_n-S\}. \]
The conditional distribution is tedious to work out directly, but there are shortcuts, e.g., the so-called ``magic $p^\star$ formula'' \citep[e.g.,][]{efron1998, fisher1934, bn1983, reid2003} is exact for this and other structured inference problems.  However, the default-prior Bayes solution only requires the (unnormalized) posterior density
\[ q_y(\theta) \propto L_y(\theta) = \prod_{i=1}^n \frac{1}{1 + (Y_i - \theta)^2}, \]
and note that there's no need to directly consider conditioning on the $U$ component; this is because the Bayes solution effectively conditions on all of the observed $Y$.  For a visualization, I simulated $n=10$ samples from a $\cau(0,1)$ model and the plot of the corresponding IM plausibility contour is shown in; the observed maximum likelihood estimator is $\hat\theta=-0.18$.  Thresholding this contour function at level $\alpha \in (0,1)$---the plot shows cutoff $\alpha=0.05$---returns an exact $100(1-\alpha)$\% confidence interval for $\Theta$. 
\end{ex}

\begin{figure}[t]
\begin{center}
\scalebox{0.65}{\includegraphics{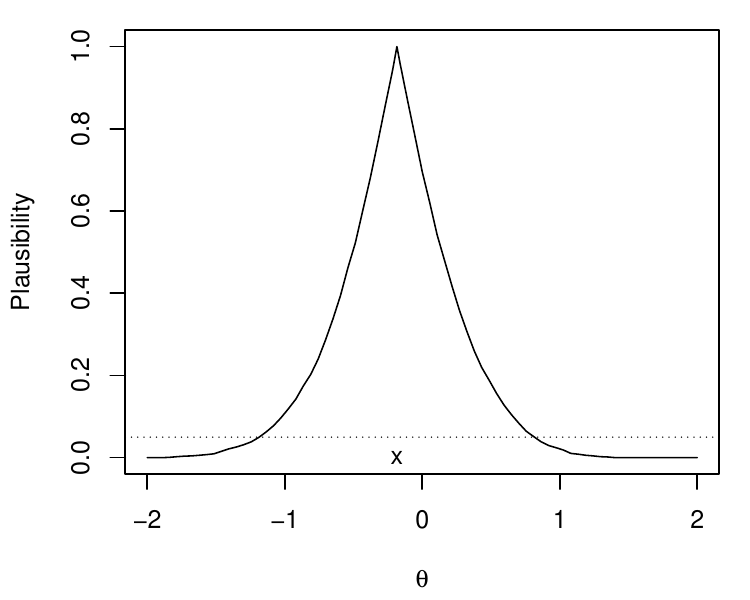}}
\end{center}
\caption{Plot of the plausibility contour function for $\Theta$ in the Cauchy location problem in Example~\ref{ex:normal2}, where $n=10$ and {\sf X} marks $\hat\theta = -0.18$.}
\label{fig:cauchy}
\end{figure}

\subsection{Examples involving multi-parameter models}

The examples presented so far all involve a single, scalar parameter.  Nothing about the methodology being proposed here is specific to scalar parameter models, but there are computational aspects that become more challenging to overcome in multi-parameter cases compared to the one-parameter case.  In particular, the grid on which the plausibility contour is to be evaluated grows exponentially in the dimension of the parameter which, in turn, makes the Monte Carlo approximation \eqref{eq:vacuous.MC} more expensive---there's more distinct data-generating distributions that need to be simulated from.  This is not such a burden in two-parameter cases like considered below, but can easily become too expensive when the parameter is 3 or more dimensions.  There would be various ways to possibly overcome the associated computational challenges, e.g., parallel computation, but a fairly general strategy is what's called {\em importance sampling} \citep[e.g.,][]{tokdar.kass.importance}.  The idea is to use Monte Carlo samples from one data-generating distribution to mimic expected value calculations with respect to another data-dependent distribution.  Specifically, let $\theta^\dagger$ be a distinguished value of the parameter, take samples $\{Y^{(m)}: m=1,\ldots,M\}$ from $\prob_{Y|\theta^\dagger}$, and then compute 
\[ \pi_y(\theta) \approx \frac{\sum_{m=1}^M w_m(\theta) \, 1\{\eta(Y^{(m)}, \theta) \leq \eta(y, \theta)\}}{\sum_{m=1}^M w_m(\theta)}, \]
where 
\[ w_m(\theta) = \frac{p_\theta(Y^{(m)})}{p_{\theta^\dagger}(Y^{(m)})}, \]
is the so-called ``importance weight'' that adjusts for the fact that the $Y^{(m)}$'s are generated from $\prob_{Y|\theta^\dagger}$ instead of $\prob_{Y|\theta}$ as in \eqref{eq:vacuous.MC}, and $p_\theta(\cdot)$ is the density of $\prob_{Y|\theta}$.  I'll employ this importance sampling strategy in each of the examples below, using the maximum likelihood estimator $\hat\theta_y$ based on the observed data $Y=y$ for the distinguished $\theta^\dagger$.  Note that, since the relative likelihood depends only on the model's minimal sufficient statistic, say, $S=S(Y)$; so if the sampling distribution of $S$ is known---like in Example~\ref{ex:normal2} below---then $Y^{(m)}$ above can be replaced by samples $S^{(m)}$ of the minimal sufficient statistic, and the importance weight can be adjusted accordingly.  


\begin{ex}[Normal]
\label{ex:normal2}
Consider a normal model $\{\nm(\theta_1, \theta_2^2): \theta_1 \in \RR, \theta_2 > 0\}$, where $\theta_1$ and $\theta_2$ denote the mean and standard deviation, respectively.  Let $Y=(Y_1,\ldots,Y_n)$ consist of $n$ iid random variables from this normal model with unknown $\Theta=(\Theta_1,\Theta_2)$.  The likelihood function, of course, is 
\[ L_y(\theta) \propto \theta_2^{-n} \exp\Bigl[ -\frac{n}{2\theta_2^2} \Bigl\{ \hat\theta_2^2 + (\theta_1 - \hat\theta_1)^2 \Bigr\} \Bigr], \]
where $\hat\theta_1 = n^{-1} \sum_{i=1}^n y_i$ and $\hat\theta_2 = \{ n^{-1} \sum_{i=1}^n (y_i - \hat\theta_1)^2\}^{1/2}$ are the maximum likelihood estimators of $\Theta_1$ and $\Theta_2$, respectively.  Then the relative likelihood is 
\[ \eta(y,\theta) \propto \Bigl( \frac{\hat\theta_2}{\theta_2} \Bigr)^n \exp\Bigl[ -\frac{n}{2\theta_2^2} \Bigl\{ \hat\theta_2^2 + (\theta_1 - \hat\theta_1)^2 \Bigr\} \Bigr], \]
where ``$\propto$'' only includes fixed constants independent of $(y,\theta)$.  The exact sampling distribution of $\eta(Y,\theta)$, given $\Theta=\theta$, is unavailable, so numerical methods are required.  But since the above expression only depends on data through the maximum likelihood estimators/minimal sufficient statistics, and since the sampling distribution of these is known in the present case, numerical evaluation of the corresponding plausibility contour can be slightly simplified.  For illustration, we consider a data set of size $n=10$ with $(\hat\theta_1,\hat\theta_2) = (0,1)$.  Then the corresponding (importance sampling-based) IM plausibility contour is displayed in Figure~\ref{fig:normal2}.  Notice, of course, that the contour peaks at the maximum likelihood values, and decreases as one moves away from the peak in any direction.  Note, also, that the spread of the contour in the $\theta_1$-direction is narrower for small $\theta_2$ values than it is for large $\theta_2$ values; this makes sense because inference on $\Theta_1$ is more challenging/less precise when $\Theta_2$ is large instead of small.  
\end{ex}

\begin{figure}[t]
\begin{center}
\scalebox{0.65}{\includegraphics{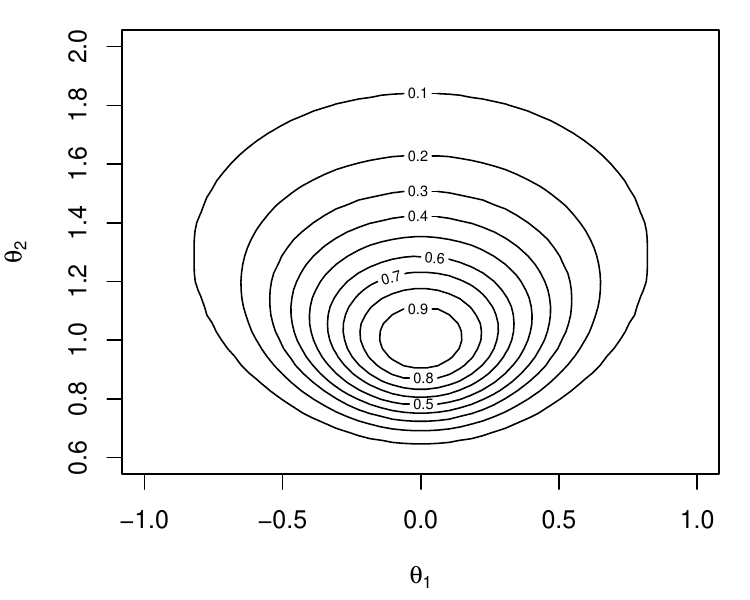}}
\end{center}
\caption{Plot of the plausibility contour function for $\Theta$, the mean and standard deviation of a normal model, from Example~\ref{ex:normal2}, where $n=10$.}
\label{fig:normal2}
\end{figure}

\begin{ex}[Gamma]
\label{ex:gamma2}
Suppose that $Y=(Y_1,\ldots,Y_n)$, given $\Theta=\theta$, is a vector of $n$ iid observations from $\gam(\theta_1, \theta_2)$, where $\theta=(\theta_1, \theta_2)$ represents an unknown and to-be-inferred shape and scale parameter pair.  For simplicity, I'll assume the prior information about $\Theta$ is vacuous, so the details in Section~\ref{SS:vacuous} are relevant here.  There's no closed-form expression for the relative likelihood, $\eta(y,\theta)$, in this example, but it's easy to evaluate numerically.  Consequently, the IM plausibility contour can readily be approximated using the importance sampling strategy described above.  For illustration, I simulated data of size $n=25$ from a gamma distribution with $\theta_1=7$ and $\theta_2=3$.  A plot of the (importance sampling-based) plausibility contour function is shown in Figure~\ref{fig:gamma2}.  The outer-most level corresponds to $\pi_y(\theta)=0.1$, so this is a 90\% plausibility region for $\Theta$, which clearly contains the true $\Theta$ in this case.  
\end{ex}

\begin{figure}[t]
\begin{center}
\scalebox{0.65}{\includegraphics{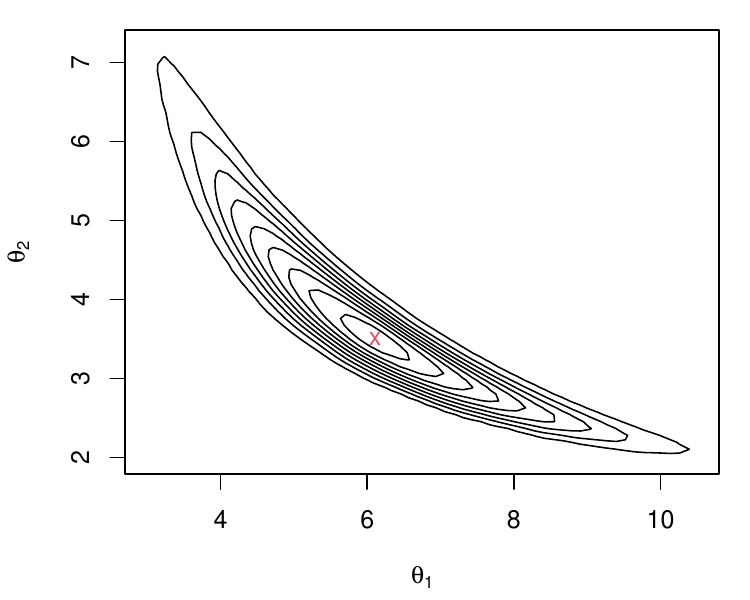}}
\end{center}
\caption{Plot of the plausibility contour function for $\Theta$, the shape and scale parameters of a gamma model, from Example~\ref{ex:gamma2}, where $n=25$ and the true parameter values are $\Theta_1=7$ and $\Theta_2=3$.  The red {\sf X} marks the maximum likelihood estimator.}
\label{fig:gamma2}
\end{figure}

\begin{ex}[Multinomial]
\label{ex:multinomial}
A random sample of size $n$ is taken from a population consisting of $K$-many categories, labeled $1,2,\ldots,K$.  The model specifies a probability vector $\theta=(\theta_1,\ldots,\theta_K)$ with 
\[ \prob_\theta(\text{observation is of category $k$}) = \theta_k, \quad k=1,\ldots,K. \]
The parameter space $\TT$ is the $K$-dimensional probability simplex, i.e., 
\[ \TT = \{ u \in \RR^K: u_k \geq 0, \, \textstyle\sum_{k=1}^K u_k = 1\}. \]
The observable data $Y=(Y_1,\ldots,Y_K)$ consists of a frequency table listing the number of observations in each of the $K$ categories; note that $\sum_{k=1}^K Y_k = n$.  The true probability vector $\Theta$ is unknown and to be inferred based on observations $Y=y$.  Since every discrete distribution on $\{1,\ldots,K\}$ can be described by such a $\Theta$ vector, I refer to this as the ``discrete non-parametric'' model.  For this reason, the multinomial model, while relatively simple, is of fundamental importance.  Many of the more general non-parametric developments, such as Bayesian non-parametrics via the Dirichlet process \citep[e.g.,][]{ferguson1973}, are built upon the multinomial model; see, also, the recent discussion paper published in the {\em Journal of the American Statistical Association} \citep{gong.jasa.mult, gong.jasa.rejoinder}.  

More specifically, let $(Y \mid \Theta=\theta) \sim \mult_K(n, \theta)$.  This determines a likelihood function $L_y(\theta) \propto \prod_{k=1}^K \theta_k^{y_k}$, for $\theta \in \TT$.  Assuming vacuous prior information for $\Theta$, for simplicity, the plausibility ordering is determined by the relative likelihood alone, 
\[ \eta(y, \theta) = \prod_{k=1}^K \Bigl( \frac{n \theta_k}{y_k} \Bigr)^{y_k}, \quad \theta \in \TT, \]
where I've plugged in the likelihood function maximizer, which is available in closed form.  From here it's straightforward to evaluate the contour function of the (strongly valid) IM for $\Theta$, via Monte Carlo plus importance sampling.  This is easy to implement but difficult to visualize---like the binomial problem in Example~\ref{ex:binomial} above, the vacuous-prior contour isn't smooth, so a contour plot on a two-dimensional slice looks very messy.  But various features of the IM for $\Theta$ can be visualized, e.g., the derived marginal plausibility contours for the individual components; that is, consider separately $\Theta_1 = f(\Theta)$ and $\Theta_2 = f(\Theta)$ and apply the marginalization formula in \eqref{eq:marginal.contour}.  For illustration, let $K=3$---so that $\Theta$ consists of three components with a linear constraint---and suppose that $y=(5,3,2)$ based on a sample of size $n=10$. Figure~\ref{fig:multinomial} plots the marginal plausibility contours for $\Theta_1$ and $\Theta_2$.  More efficient marginal inference on the components of $\Theta$ is possible, but this requires special considerations; see Part~III of the series. 
\end{ex}

\begin{figure}[t]
\begin{center}
\subfigure[$\hat\theta_1 = 0.5$]{\scalebox{0.55}{\includegraphics{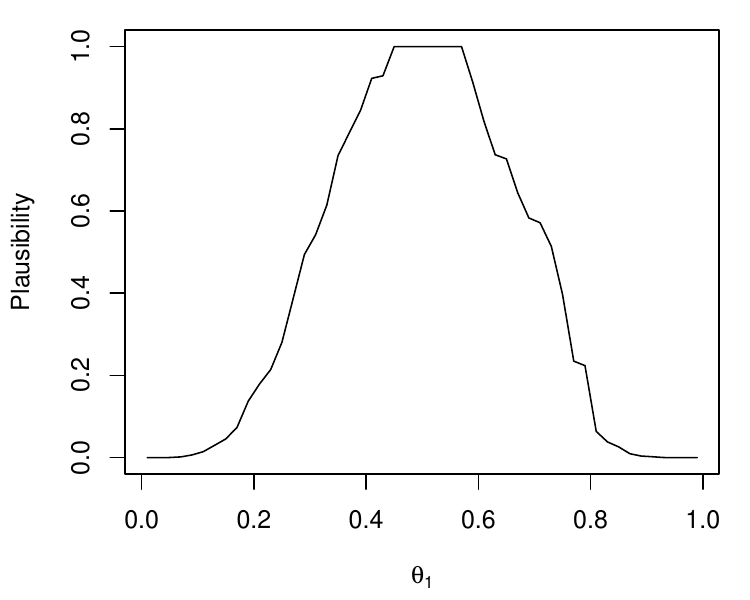}}}
\subfigure[$\hat\theta_2=0.3$]{\scalebox{0.55}{\includegraphics{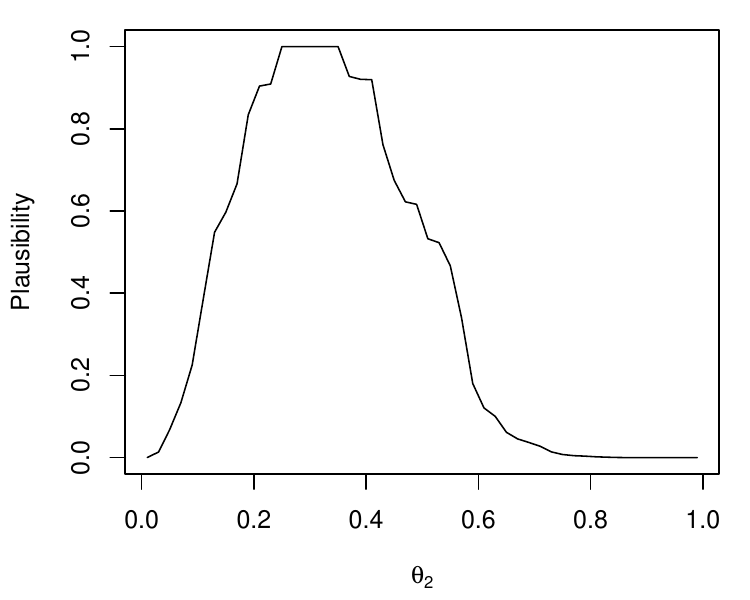}}}
\end{center}
\caption{Plots of the marginal plausibility contours for $\Theta_1$ and $\Theta_2$ for the multinomial model in Example~\ref{ex:multinomial} based on $K=3$ categories and data $y=(5, 3, 2)$.}
\label{fig:multinomial}
\end{figure}

\begin{ex}[Generalized linear models]
\label{ex:glm}
A powerful tool in the applied statisticians' toolbox is the family of {\em generalized linear models}, or GLMs for short \citep[e.g.,][]{McCullaghNelder:1989}.  This includes the usual Gaussian linear regression as a special case, as well as the commonly used {\em logistic regression model} for binary response variables.  For concreteness, here I'll focus on the logistic regression case.  Suppose the data consists of independent pairs $Y_i=(X_i, u_i)$, for $i=1,\ldots,n$, where $u_i \in \RR^d$ is an independent/explanatory variable---which I take to be non-random here---and $X_i \in \{0,1\}$ is the binary dependent/response variable.  The goal is to understand the relationship between the explanatory and response variables.  The logistic regression model assumes that 
\[ (X_i \mid u_i, \Theta=\theta) \ind \ber\{ H(u_i^\top \theta) \}, \quad i=1,\ldots,n, \]
where $H(z) = (1 + e^{-z})^{-1}$ is the logistic distribution function and $\Theta \in \TT \subseteq \RR^d$ is an unknown coefficient vector.  This determines a likelihood function 
\[ L_y(\theta) = \prod_{i=1}^n H(u_i^\top \theta)^{x_i} \{1 - H(u_i^\top \theta)\}^{1-x_i}, \quad \theta \in \TT. \]
Closed-form expressions aren't available here but, as explained in more detail above, it's possible to evaluate the IM's plausibility contour numerically via Monte Carlo.  

For illustration, consider the data in Table~8.4 of \citet[][p.~252]{ghosh-etal-book}, on the relationship between exposure to chloracetic acid ($u$) and the death of mice ($x$).  In particular, $u$ is the acid dosage and $x=1$ if the exposed mice dies and $x=0$ otherwise.  A total of $n=120$ mice are exposed, ten at each of the twelve dosage levels.  Figure~\ref{fig:logistic} shows the (vacuous-prior) IM's plausibility contour for the unknown $\Theta=(\Theta_1,\Theta_2)^\top$.  As expected, the results of this analysis suggests that the death probability is increasing in the dose level---the marginal contour for $\Theta_2$ concentrates on the positive side of 0.  For comparison, the 90\% confidence region based on the asymptotic normality of the maximum likelihood estimator is also shown.  In this case, the two 90\% regions are similar, which suggests that the IM's validity doesn't come at the cost of efficiency.
\end{ex}

\begin{figure}
\begin{center}
\scalebox{0.65}{\includegraphics{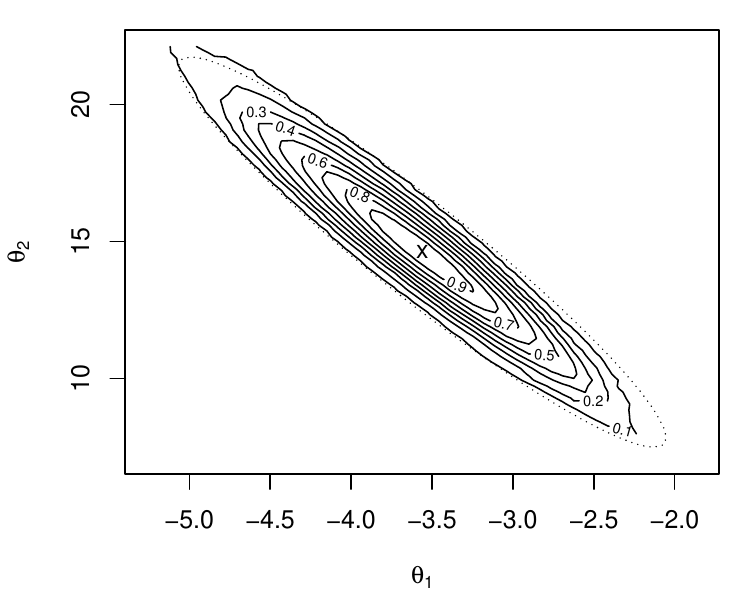}}
\end{center}
\caption{Plot of the plausibility contour for the two-parameter logistic regression illustration in Example~\ref{ex:glm}; dotted line gives the 90\% confidence region for $\theta$ based on asymptotic normality of the maximum likelihood estimator.  }
\label{fig:logistic}
\end{figure}

\section{Imprecise statistical model}
\label{S:imprecise}

The formulation and results in Section~\ref{S:general} cover very general kinds of models, but, so far, I've only directly considered cases where the statistical model is precise and imprecision enters the analysis through the availability of partial prior information about the uncertain quantity of interest.  That is, the starting point is a collection of {\em precise} probability distributions $\model := \{\prob_{Y|\theta}: \theta \in \TT\}$ and imprecision comes from the handling of partial prior information about the unknown $\Theta \in \TT$.  These cases are the most familiar to me and probably to many other statisticians, but these aren't the only cases.  In fact, as I explain below, arguably all real applications should have some degree of imprecision in the model itself, but this is rarely done because (possibly unjustified) assumptions are made to simplify the model and analysis.  In this section, I'll dig into a few such cases.  This will be a rather shallow dig, however, because these cases are important and deserve their own separate and more thorough investigation.  

Start with a common situation in sample surveys and other observational studies where data can be missing \citep[e.g.,][]{little.rubin.book.v2}.  Of course, the data analyst knows whether a data point $Y$ is observed or not, so there's an observable ``observed or missing'' indicator variable $Z$, with $Z=1$ if $Y$ is observed and $Z=0$ if $Y$ is missing.  By the law of total probability, the (marginal) distribution of $Y$ can be expressed as 
\begin{align*}
\prob_{Y|\theta}( B) & = \prob_{Y|\theta}(B \mid Z=1) \, \prob_{Z|\theta}(1) + \underbrace{\prob_{Y|\theta}(B \mid Z=0)}_{\text{unidentified}} \, \prob_{Z|\theta}(0), \quad B \subseteq \YY, 
\end{align*}
where $\prob_{Y|\theta}(\cdot \mid Z=1)$ and $\prob_{Z|\theta}(\cdot)$ are both (precise and) determined by the posited statistical model.  As \citet{manski.book} explains, the highlighted term is ``unidentified'' because, by definition, the values of $Y$ are unobserved when $Z=0$.  Therefore, there's an inherent degree of imprecision that can't be overcome by simply collecting more data.  Manski calls this a {\em partially identified} model, since the conditional distribution of $Y$, given $Z=1$, and the marginal distribution of $Z$ are identified, but the conditional distribution of $Y$, given $Z=0$, is unidentified.  So, without imposing any additional structure/assumptions, the model is genuinely imprecise, i.e., 
\begin{align*}
\lprob_{Y|\theta}(B) & = \prob_{Z|\theta}(1) \, \prob_{Y|\theta}(B \mid Z=1) \\
\uprob_{Y|\theta}(B) & = \prob_{Z|\theta}(1) \, \prob_{Y|\theta}(B \mid Z=1) + \prob_{Z|\theta}(0).
\end{align*}
The degree of imprecision or partial identification is controlled by the magnitude of $\prob_{Z|\theta}(1)$: smaller values mean more imprecision, a larger credal set; in particular, if the value is 1, then the model is fully identified and we recover the precise case discussed in the previous sections and, if it's 0, then the model is fully unidentified or vacuous.  

Naturally, if $Y$ is a vector $(Y_1,\ldots,Y_n)$ of independent data points, then it might be that some $Y_i$'s are observed and some are missing.  In that case, the above discussion applies to each individual $Y_i$, by introducing a corresponding $Z_i$, and then the model for the vector $Y$ is the credal set consisting of all product measures compatible with the imprecise marginals above; that this collection is closed and convex isn't too difficult to check, given that $\prob_{Y|\theta}(\cdot \mid Z=1)$ and $\prob_{Z|\theta}(\cdot)$ are fixed by the posited statistical model.  This upper joint distribution for $Y$, given $\Theta=\theta$, is what's needed to carry out the IM construction described in Section~\ref{SS:construction}, which I explain below.  

To avoid measure-theoretic technicalities, here I'll suppose that each $Y_i$ is discrete with the same finite/countable support $\YY$.  If $\Theta$ has a partial prior, then the ``joint distribution'' for $(Y,\Theta)$ is a credal set consisting of all precise joint distributions corresponding to the product of a $\prob_{Y|\theta}$ in the credal set described above and a probability $\prob_\Theta$ in the prior credal set.  Again, that this collection is closed and convex follows from the discussion above and the assumed closedness and convexity of the prior credal set.  Then the imprecise joint distribution of $(Y,\Theta)$ has a contour function that can be factored as 
\[ \uprob_{Y,\Theta}(\{y,\theta\}) = \uprob_{Y|\theta}(\{y\}) \, \uprob_\Theta(\{\theta\}), \]
the product of the two corresponding credal set contour functions. The first term on the right-hand side can be interpreted as an ``upper likelihood,'' which plays a fundamental role in the literature on statistical inference with imprecise models, e.g., \citet[][Ch.~8.5.3]{walley1991} and \citet{zhang.isr.2010}.  After a closer look at the upper-likelihood term, it's not difficult to see that this simplifies to 
\[ \uprob_{Y|\theta}(\{y\}) =  \prob_{Z|\theta}(1)^{N_1(y)} \prob_{Z|\theta}(0)^{n-N_1(y)} \prod_{i: z_i=1} p_\theta(y_i \mid z_i=1), \quad y=(y_i: z_i=1), \]
where $z_i$ is the not-missing/missing indicator, $p_\theta(\cdot \mid z_i=1)$ is the posited mass function for non-missing $Y_i$'s, and $N_1(y) = |\{i: z_i=1\}|$.  The joint contour will inevitably be normalized as a function of $\theta$, so any factors that depend only on $y$ can be ignored.  For example, if $\prob_{Z|\theta}$ doesn't depend on $\theta$, i.e., if data are {\em missing at random} (MAR), then the leading two terms in the above display can be dropped.  In that case, the missingness indicators are ancillary statistics, cf.~Section~\ref{SS:partial}, so they should be conditioned on like if, say, the sample size were drawn at random from some fully known distribution.  From this point, the IM construction described in Section~\ref{SS:construction} can proceed exactly as before, and the corresponding theoretical properties stated in Section~\ref{SS:properties} apply.  For example, in the discrete uniform problem described in Example~\ref{ex:discrete.unif}, if a simple coin flip---with possibly unknown probability, just not depending on $\Theta$---determines whether an observation is missing or not, then the results presented there don't change.  That is, the solution is just to ignore the missing data and base the IM solution on just the available non-missing data.  If what determines the missing/not-missing characteristic of an observation depends on the unknown $\Theta$, e.g., if extreme observations are more likely to be missing than non-extreme, then that creates some additional challenges.  One approach to overcome these challenges was presented in \citet{imcens} in the context of censored data (see below), but I'll address this more formally elsewhere.  

The all-or-nothing, missing-or-not case discussed above is a special case of a more general situation of what's called {\em coarse data} \citep[e.g.,][]{couso.dubois.2018, guillaume.dubois.2020, heitjan.rubin.1991, schreiber2000}.  An interesting, imprecise-probabilistic perspective is to treat the (coarse) data as set-valued observations that, by definition, contain the respective precise-but-unobservable data point.  A simple example of this case is rounding: all measurement devices have limited precision, so every numerical data point that's ever been recorded is really an interval determined by the measurement device's precision.  Another case is {\em censored} observations common in reliability and biomedical applications, for example, a patient's time of remission occurred sometime between two consecutive check-ups at the clinic \citep[e.g.,][]{km.book}.  Let $Y_i \subseteq \YY$ denote the coarsened version of the ``ideal'' data point $X_i$, with the link between them being that $X_i$ is an {\em almost sure selector} of $Y_i$, which means that $Y_i \ni X_i$ with probability~1 \citep[][Ch.~2.5]{nguyen.book}.  The missing-data problem described above corresponds to the special case where $Y_i$ equals either $\{X_i\}$ or $\YY$ and the missingness indicator $Z_i$ determines which of the two cases it is.  Under this coarsening framework, there would often be a posited model $\prob_{X|\theta}$ for the ideal data point $X_i$ and some description of the coarsening process $X_i \mapsto Y_i$, which may be stochastic and might depend on $\theta$.  These together determine a marginal distribution $\prob_{Y|\theta}$ for the observable coarsened data $Y$, whose probability mass function---with admittedly sloppy notation---is given by
\[ p_\theta(y_i) = \sum_{x_i \in y_i} p_\theta(y_i \mid x_i) \, p_\theta(x_i), \quad y_i \subseteq \YY. \]
This expression pools together the model for $X_i$ and a conditional distribution of $Y_i$, given $X_i$, which describes the coarsening process; in general, both of these pieces depend on $\theta$.  Without additional assumptions, one must deal directly with random sets and their distributional properties, which is equivalent to working with an imprecise model or a set of precise distributions for the unobservable $X_i$'s.  Those statisticians who frequently deal with coarsened or censored data might not recognize this need to consider random sets, and that's because the typical applications are presented with additional assumptions, implicitly or explicitly stated, that simplify the problem.  For example, the common {\em coarsening at random} (CAR) assumption implies that $p_\theta(y_i \mid x_i)$ is constant in $(x_i,\theta)$.  In that case, when considering likelihoods, as in this paper's approach, proportionality constants can be ignored, and the expression in the above display simplifies to 
\[ p_\theta(y_i) \propto \sum_{x_i \in y_i} p_\theta(x_i), \quad y_i \subseteq \YY, \]
which is just a model-based (precise) probability calculation, namely, that the ideal data $X_i$ happens to be contained in $y_i$.  This amounts to ignoring the coarsening process and, in turn, the random set aspect of the problem, because CAR assumes that the coarsening is non-informative.  Whether CAR is a justifiable assumption is determined by the context of the problem, not by the convenience it affords the data analyst.  

These problems, and the corresponding IM constructions, are more nuanced when covariates are involved, as is very often the case in applications.  The theory developed here covers such cases, but the details require more care and will be addressed in a follow-up paper.  In particular, causal inference \citep{imbens.rubin.book} is an important application involving structurally-missing data that will be investigated elsewhere.  

Related to the above discussion, a terrifying reality is that there surely are variables relevant to the phenomenon under investigation that were not measured and, therefore, entirely missing.  I say this is ``terrifying'' because ignoring the unmeasured variables implies that there's always a risk of getting trapped by Simpson's paradox---where the inference would be entirely different had the unmeasured variable been considered in the analysis.  In principle, an imprecise model can accommodate such cases by, say, considering simultaneously all the joint distributions for measured and unmeasured variables that are compatible with the posited model for the measured variables.  I'm not suggesting this as a practical strategy, however, because one can quickly get deep into the ``what if'' rabbit hole and end up unable to do anything with the data that were collected.  

I've focused here on imprecise statistical models that result from imprecision in the data itself.  But that's not the only way that imprecision in the model might appear.  For example, one might posit a (parametric) neighborhood model to achieve a certain degree of robustness to model misspecification bias.  This might be an attractive and perhaps a computationally simpler alternative to non-parametric approaches that achieve this flexibility by identifying the statistical model as the parameter.

\section{Conclusion}
\label{S:discuss}

This paper develops a general framework for constructing an IM for data-driven uncertainty quantification about unknowns---model parameters, future observables, and general functionals of the underlying distribution---that can accommodate very general kinds of models, including cases with partial prior information, and has certain statistical and behavioral guarantees.  More specifically, on the statistical side, the IM output is strongly valid in the sense that the random variable $\pi_Y(\Theta)$, the IM's plausibility contour, as a function of $(Y,\Theta)$ with distribution compatible with the posited imprecise probability model, is stochastically no smaller than uniform; this property holds exactly, no asymptotic approximations necessary.  This result has two important consequences:
\begin{itemize}
\item In terms of traditional statistical methods, i.e., procedures for making decisions, this result makes crystal clear how the IM's output can be used to construct such procedures with error rate control guarantees.  And don't forget that this error rate control property is {\em relative to the posited $\uprob_{Y,\Theta}$}, with the most stringent case corresponding to a vacuous prior.  So, when non-vacuous partial prior information is incorporated, the strong validity property is weaker, easier to satisfy, and, therefore, can be achieved by a more efficient IM with tighter contour.  This makes clear the trade-off: my IM will be valid no matter what, but if I can justify stronger model assumptions, then I can expect greater efficiency.  
\vspace{-2mm}
\item Perhaps more importantly, in terms of probabilistic reasoning, inference can go wrong when the event ``hypothesis $A$ is true but the IM's upper probability $\uPi_y(A)$ is small'' occurs.  In applications, of course, there's no way to know if $\uPi_y(A)$ is rightfully small or if the bad event above has occurred.  Since probabilistic reasoning only works if the (imprecise) probabilities it's based on are reliable, having some calibration guarantees is absolutely essential.  The results presented here show how this can be achieved exactly, not approximately, in virtually any application, any kind of model, and any degree of precision/imprecision.  
\end{itemize}
On the behavioral side, first, the IM output has the mathematical form of a coherent lower/upper probability, so there's an automatic ``avoids sure loss'' conclusion.  Second, the IM output also partially (and often fully) achieves a certain updating coherence property that relates to how prior beliefs are modified in light of observations.  These two coherence properties together imply a sort of internal rationality behind the proposed framework, in addition to the external rationality derived from its validity/reliability properties.  That the IM framework is unable to achieve both valdity and full-blown update coherence is not an issue for me.  My conjecture is that it's impossible to achieve both simultaneously, at least not without loss of efficiency, so since reliability is my top priority, I'd prefer to sacrifice (a little) on the behavioral side if it means validity and efficiency are in reach on the statistical side.  

The driving force behind these new developments is the construction of a suitable outer consonant approximation of the posited (imprecise) joint distribution for $(Y,\Theta)$.  The idea itself isn't new, but it's not been used in a statistical context like this before, at least not to my knowledge.  In retrospect, it's really an obvious solution: consonance is needed in order to achieve strong validity, and some kind of ``dominance'' is required to maintain a connection to the posited (imprecise) model, but the model itself isn't consonant, so just take an upper consonant approximation thereof!  An important consequence of this outer approximation is that it ends up being likelihood-driven, which has a number of advantages.  For one thing, it makes intuitive sense that the IM ought to rely heavily on the (relative) likelihood, since that's what contains all the relevant information in the data.  More than that, it's well-known that likelihood-based methods enjoy certain optimality properties, at least under regularity conditions, so the IM that's likelihood-driven can expect to share those same optimal statistical guarantees.  

Despite covering a lot of ground in this paper, there are still many unanswered questions.  Below is a short list of a few things I have in mind.  These will likely be topics of follow-up papers in this series.  

\begin{question}
In Section~\ref{SS:part.naive}, I explained how there's a gap in my/our present understanding of how to reduce the partial-prior IM's complexity to boost efficiency.  The vanilla strategy in Section~\ref{SS:partial} is strongly valid and fully incorporates the available prior information, but is generally too conservative, i.e., puts too much weight on the prior.  The naive, complexity-reduced partial-prior IM is simpler and still strongly valid but arguably sits too close to the vacuous-prior IM solution and, therefore, doesn't take full advantage of the available partial prior information.  The question is: how to leverage the partial prior structure and tailor a dimension-reduction strategy that is as efficient as possible?  Looking at the two extremes---vacuous and complete priors---we see the following pattern: in the vacuous prior case I fix the value of $\theta$ and, conversely, in the complete prior case I fix the value of $y$.  Then my intuition is that, in the partial prior case, one needs to fix the value of some feature of both $y$ and $\theta$.  This seems quite interesting, adding some nuance to the old-school conditional/unconditional debates that used to take place in statistics journals and at conferences.  And since the opportunities created by the IM's ability to incorporate partial prior information is a major motivation for this effort, I'd say this is the most important open question at the time I'm writing this.  
\end{question}

\begin{question}
The computations that I carried out\footnote{I'll eventually post my R codes for (at least a subset of) the examples presented in this paper on my website, \url{https://www4.stat.ncsu.edu/~rmartin/}, for anyone who's interested.} for the examples presented in this paper are all relatively simple---either closed-form expressions were available or I used the basic Monte Carlo (or importance sampling) strategy to evaluate the Choquet integral; see, also, \citet{hose.hanss.martin.belief2022}.  This isn't difficult to do in the low-dimensional problems presented here, but how to scale these up to higher-dimensional cases efficiently and without losing accuracy?  There are a number of different ideas that could be pursued, e.g., using better and more efficient Monte Carlo methods and designing problem-specific computational strategies one at a time.  Here especially I'd welcome contributions from others who have more experience on the computational side.  
\end{question}


\begin{question}
In Section~\ref{SS:vacuous}, I pointed out the incompatibilities between methods that satisfy the likelihood principle and those that are valid and efficient.  As I briefly argued there, if I know what stopping rule was used when the data were collected, my choice of statistical framework shouldn't prevent me from using that information if it can improve the efficiency of my inferences.  But these considerations raise an interesting question: if I really don't know what stopping rule was used, then I'd require that my inference be valid across all those data-generating distributions that lead to the same observed likelihood function, so how can I achieve this?  Not knowing the stopping rule boils down to imprecision in the model formulation and, therefore, a natural way to approach this would be to consider an imprecise model in the sense of Section~\ref{S:imprecise}.  A first attempt along this line is presented in \citet{martin.basu}.  
\end{question}

\begin{question}
The IM framework proposed here is, for good reasons (Sections~\ref{S:blocks}--\ref{S:general}), fully likelihood-based.  But there is an important class of problems where the likelihood function itself is intractable, or perhaps there is no likelihood function at all.  The former issue often emerges in complex problems involving hierarchical models or mixtures where the likelihood involves integrals that can't be written in closed-form; the latter emerges in cases, common in machine learning, where the quantity of interest isn't defines as a parameter of a statistical model but as the minimizer of some expected loss function.  If there is effectively no likelihood function available, then how can the IM framework proposed here be applied?  While there is good reason to build up an IM based on the likelihood function when it's available, fortunately, this isn't the only way to proceed with the IM construction.  At its core, the above construction treats the relative likelihood as a (justifiable) measure of the compatibility between a generic value $\theta$ of $\Theta$ and the observed data.  So the data analyst is free to choose a different measure of compatibility, whether this is out of necessity (as in the above explanation) or for the sake of novelty, computational efficiency, etc.  
This latter perspective is very much in line with the IM-based solutions put forward recently in \citet{imconformal.supervised, cella.martin.imrisk} so I think there are promising opportunities to extend the proposal here in this paper to those modern non-likelihood-based settings.  
\end{question}

\begin{question}
So far, I've not felt compelled to explore the asymptotic properties of IMs because strong validity properties are exact in finite samples.  But an asymptotic investigation would be valuable, if only to demonstrate that the same kind of limiting properties users expect still hold for IMs.  In particular, I'm envisioning an imprecise-probabilistic version of the so-called {\em Bernstein--von Mises theorem} showing that the IM's consonant plausibility function asymptotically merges with a suitable ``Gaussian'' plausibility function, e.g., maybe like the Gaussian random fuzzy sets in \citet{denoeux.fuzzy.2022}.  
\end{question}

\begin{question}
The focus here, and in virtually all of the previous IM-related work, is on inference.  But decision-theoretic considerations are interesting and within reach.  Some first developments are presented in \citet{imdec} but this only considered the case with vacuous prior information, so I believe that much more can be done.  This belief is based on an analogy with classical decision theory: there's not much one can say about admissibility or optimality directly, these results are usually established indirectly by connecting it to a Bayes decision problem.  What more/new can be done when the IM incorporates a {\em partial prior}? This adds an interesting new angle to this problem.
\end{question}

Finally, I'd like to briefly comment on the scope of these developments, to follow up on ``big-picture implications'' discussed in Section~3.4 of Part~I.  Data science is booming, as both government and industry are placing increased emphasis on data-driven decision-making.  I can't complain about the opportunities this has afforded those of us in quantitative fields, but I'm also embarrassed about the statistical community's major failure that this attention has shone a spotlight on.  By not resolving our subject's foundational questions, statisticians are forced to admit that there are at least two acceptable approaches/solutions to any given problem.  But if there's two, then why not three, four, ...?  The result is an anything-goes-as-long-as-it-works mentality, a veritable Wild-West of new methods claiming to be ``optimal'' in one sense or another.  This is an absolutely perfect environment for individual researchers publishing papers, and I can appreciate the creative spirit such an environment inspires, but what are the effects this has on the field of statistics and, more importantly, on science and society as a whole?  

The general public is far removed from optimality claims and the Bayesian-versus-frequentist debate, but these are still relevant.  Experts tout the power of data-analytic technology and, based on this, ask the public to ``trust the science.''  Then, time after time, in high-stakes settings, those conclusions that the public was asked to trust are later revised, accompanied by excuses that {\em the data were wrong}, etc.  It's no wonder that trust is lacking!  My complaint isn't that the conclusions are sometimes wrong---this is unavoidable in our line of business.  My complaint is that at least some of what the public is being asked to trust isn't scientific, it's the aforementioned Wild-West.  As a result, the experts exaggerate the degree of certainty or precision in the inferences they draw, in part because virtually all of the existing data-analytic methods insist on precision from start to finish, whether it's justified or not.  Presumably, the experts don't exaggerate precision intentionally, they just don't know how to accommodate imprecision and to make more cautious or conservative inference in a principled way.  Technicalities aside, I believe that what I'm proposing in this series of papers can address this challenge.  First, it provides data analysts with the opportunity to incorporate imprecision if/where it's warranted, and the imprecision they introduce propagates through the analysis and is reflected in more conservative inferences.  Second, thanks to the validity constraint, the IM's output is always imprecise, reflecting the fact that there are inherent precision limits when making justifiable data-driven decisions.  While I believe that more caution and levelheadedness is needed in data science, I understand that there are lots of sociological factors in play (e.g., there's no fortune and fame in making cautious predictions).  Levelheadedness is sure to prevail, and I hope that victory comes sooner rather than later.

\section*{Acknowledgments}

First, thanks to Leonardo Cella and Dominik Hose for helpful discussions and feedback on earlier drafts of this manuscript.  Second, thanks to the students who participated in my Fall 2022 special topics course\footnote{All the course materials, including lecture videos and slides, are publicly available at  \url{https://wordpress-courses2223.wolfware.ncsu.edu/st-790-001-fall-2022/}.}---entitled {\em Imprecise-Probabilistic Foundations of Statistics \& Data Science}---for bearing with me during my lectures on early versions of the material presented here.  Finally, thanks to the U.S.~National Science Foundation, grant SES--2051225, for partially supporting this effort.

\bibliographystyle{apalike}
\bibliography{/Users/rgmarti3/Dropbox/Research/mybib}

\end{document}